\documentclass[preprint]{elsarticle}
\usepackage{graphicx}
\usepackage{amsmath}
\usepackage{amssymb}
\usepackage{amsthm}
\usepackage{hyperref}

 \usepackage{float}
\usepackage{epstopdf}
\usepackage{xcolor} 
\usepackage{multirow}
\usepackage{soul}

\usepackage{subfigure}

\newtheorem{theorem}{Theorem}[section]

\newtheorem{remark}{Remark}

\newtheorem{scheme}{Scheme}

\newtheorem{system}{System}

\newcommand{\tr}{\mbox{Tr}}
\newcommand{\ketbra}[2]{\left | #1\right \rangle \left \langle #2\right | }

\begin{document}

\title{Numerical solution of stochastic master equations using stochastic interacting wave functions}

\author[dim-ci2ma]{C. M. Mora}
    \ead{cmora@ing-mat.udec.cl}

\author[dim]{J. Fern\'andez}
    \ead{joafernandez@udec.cl}

\author[cimat]{R. Biscay}
    \ead{rolando.biscay@cimat.mx}

\address[dim-ci2ma]{Departamento de Ingenier\'{\i}a Matem\'{a}tica and CI$^2$MA, Universidad de Concepci\'{o}n, Chile.}

\address[dim]{Departamento de Ingenier\'{\i}a Matem\'{a}tica, Universidad de Concepci\'{o}n, Chile.}

\address[cimat]{Departamento de Probabilidad y Estad\'istica, Centro de Investigaci\'on en Matem\'atica, Guanajuato, M\'exico.}

\begin{abstract}
We develop a new approach for solving stochastic quantum master equations
with mixed initial states.
First,
we obtain that the solution of the jump-diffusion stochastic master equation
is represented by a mixture of pure states 
satisfying a system of stochastic differential equations of Schr\"odinger type.
Then, 
we design three exponential schemes for these 
coupled stochastic Schr\"odinger equations,
which are driven by Brownian motions and jump processes.
Hence,
we have constructed efficient numerical methods for the stochastic master equations
based on quantum trajectories.
The good performance of the new numerical integrators 
is illustrated by simulations of two quantum measurement processes.
\end{abstract}


\maketitle

\section{Introduction}

This paper addresses the numerical simulation of open quantum systems.
We consider a small quantum system 
described by the time-dependent Hamiltonian $\mathbf{H} \left( t \right)$
that interacts  with the environment via the Gorini-Kossakowski-Sudarshan-Lindblad operators $\mathbf{L}_{j} \left( t \right) $ and  $\mathbf{R}_{m} \left( t \right) $
(see, e.g., \cite{AlickiLendi2007,BreuerPetruccione2002,Carmichael2008,HarocheRaimond2006}).
Here, 
for any $t \geq 0$,
the linear operators
 $\mathbf{H} \left( t \right)$, $\mathbf{L}_{j} \left( t \right) $ and  $\mathbf{R}_{m} \left( t \right) $
 act on the complex Hilbert space $\left( \boldsymbol{\mathfrak{h}}, \langle \cdot, \cdot \rangle \right)$.
The main goal of this article is to develop the numerical solution of the stochastic master equation
\begin{equation}
  \label{eq:StochasticQME}
  \begin{aligned}
 d\mathbf{\boldsymbol{\rho}}_{t}	
 & =	
\boldsymbol{\mathcal{L}} \left( t \right)  \boldsymbol{\rho}_{t-} dt 
+ \sum_{j =1}^J \left( \mathbf{L}_{j} \left( t \right) \boldsymbol{\rho}_{t-} + \boldsymbol{\rho}_{t-} \mathbf{L}_{j} \left( t \right)^* 
 - 2 \Re  \left(  \tr \left( \mathbf{L}_{j} \left(  t \right) \boldsymbol{\rho}_{t-} \right) \right) \boldsymbol{\rho}_{t-}  \right) dW^{j}_{t}
\\
& \quad
+ 
\sum_{m=1}^M   \left( 
\frac{ \mathbf{R}_{m} \left( t \right) \boldsymbol{\rho}_{t-}  \mathbf{R}_{m} \left(  t \right) ^*}
{ \tr \left( \mathbf{R}_{m} \left( t \right) ^*\mathbf{R}_{m} \left( t \right) \boldsymbol{\rho}_{t-} \right) } 
- \boldsymbol{\rho}_{t-} \right) \biggl( dN^{m}_{t} -  \tr \left( \mathbf{R}_{m} \left( t \right) ^*\mathbf{R}_{m} \left( t \right) \boldsymbol{\rho}_{t-} \right) dt \biggr) ,
\end{aligned}
\end{equation}
where 
$ \boldsymbol{\rho}_{t} $ is a  random density operator
(i.e., a random non-negative operator on $\boldsymbol{\mathfrak{h}}$ with unit trace),
$W^1, \ldots, W^J$ are independent real Brownian motions,
the $N^m$'s are doubly stochastic Poisson processes (also known as  Cox processes) with predictable compensator 
$
t \rightarrow \int_{0}^t \tr \left( \mathbf{R}_{m} \left( s \right) ^*\mathbf{R}_{m} \left( s \right) \boldsymbol{\rho}_{s-} \right) ds
$
(see, e.g., \cite{Bremaud1981,JacodProtter1982,Pellegrini2010})
such that $N^1,\ldots, N^M$ have no common jumps, 
and
$$
\boldsymbol{\mathcal{L}} \left( t \right) \boldsymbol{\varrho}
=
\mathbf{G} \left( t \right) \boldsymbol{\varrho} + \boldsymbol{\varrho} \, \mathbf{G} \left( t \right) ^* 
+ \sum_{j = 1}^{J} \mathbf{L}_{j}  \left( t \right) \boldsymbol{\varrho} \, \mathbf{L}_{j} \left( t \right)^* 
+ \sum_{m=1}^{M} \mathbf{R}_{m} \left( t \right) \boldsymbol{\varrho} \, \mathbf{R}_{m} \left( t \right)^*
$$
with 
\begin{equation}
\label{eq:1.1}
 \mathbf{G} \left( t \right)  = -  \mathrm{i}  \, \mathbf{H} \left( t \right) 
- \dfrac{1}{2}\, \sum_{j = 1}^{J} \mathbf{L}_{j}  \left( t \right)^*\mathbf{L}_{j}  \left( t \right) 
- \dfrac{1}{2}\, \sum_{m = 1}^{M} \mathbf{R}_{m}  \left( t \right)^*\mathbf{R}_{m}  \left( t \right) .
\end{equation}
It is worth pointing out that the  unknown 
$\left( \boldsymbol{\rho}_{t} \right)_{t \geq 0}$ 
is an $\boldsymbol{\mathfrak{L}} \left(\boldsymbol{\mathfrak{h}}\right)$-valued adapted stochastic process 
on the underlying filtered complete probability space 
$\left( \Omega ,\mathfrak{F}, \left(\mathfrak{F}_{t}\right) _{t\geq 0},\mathbb{P}\right) $,
where 
$\boldsymbol{\mathfrak{L}} \left(\boldsymbol{\mathfrak{h}}\right)$  stands for the space of all linear operators from $\boldsymbol{\mathfrak{h}}$ to $\boldsymbol{\mathfrak{h}}$,
and that $\mathbf{H} \left( t \right)$ is a symmetric operator.
In this paper,
we design efficient numerical methods for the stochastic evolution equation \eqref{eq:StochasticQME}.

The stochastic operator equation \eqref{eq:StochasticQME} describes
the dynamics of several quantum systems interacting with the environment
under the  Born approximation.
In the quantum measurement  process,
$ \boldsymbol{\rho}_{t}$ represents the system density operator conditioned on the measurement outcomes
(see, e.g., \cite{BarchielliBelavkin1991,BarchielliGregoratti2009,BarchielliGregoratti2012,BarchielliPellegriniPetruccione2012,BreuerPetruccione2002,WisemanMilburn2009}).
Moreover,
$N^m_t$ is the number of detections registered  up to time $t$ by the counter associated to the observable 
$\mathbf{R}_{m} \left( t \right)$,
and
the integral from $0$ to $t$ of the photocurrent in the homodyne or heterodyne measurement 
of the observable $\mathbf{L}_{j} \left(  t \right)$ is proporcional to 
$
W_t^j  + 2 \int_{0}^t  \Re  \left(  \tr \left( \mathbf{L}_{j} \left(  s \right) \boldsymbol{\rho}_{s-} \right)  \right) ds
$
(see, e.g, \cite{BarchielliBelavkin1991,BarchielliGregoratti2009,BarchielliGregoratti2012,BreuerPetruccione2002,WisemanMilburn2009}).
Since \eqref{eq:StochasticQME} describes the continuous weak measurements on the small quantum system,
the stochastic master equation  \eqref{eq:StochasticQME} and its versions
have been used to study and design 
quantum feedback control systems (see, e.g.,  \cite{BarchielliGregoratti2012,Doherty2000,Mirrahimi2007,Rouchon2015,Sarovar2004,WisemanMilburn1993,WisemanMilburn2009}).

Let $\boldsymbol{\rho}_{0}$ be a random pure state,
which means, in the the Dirac notation, that 
$
\boldsymbol{\rho}_{0} =\ketbra{\mathbf{X}_0}{\mathbf{X}_0}
$,
where 
$\mathbf{X}_0$ is a $\mathfrak{F}_{0}$-measurable random variable taking values in $\boldsymbol{\mathfrak{h}}$
such that $\left\Vert \mathbf{X}_0 \right\Vert  = 1$.
Then
\begin{equation}
 \label{eq:1.2}
 \boldsymbol{\rho}_{t} = \ketbra{\mathbf{X}_t}{\mathbf{X}_t}
\end{equation}
(see, e.g., Remark \ref{rem:PureInitialState} given below),
where $\mathbf{X}_t$ is an adapted stochastic process  taking values in $\boldsymbol{\mathfrak{h}}$
that satisfies the non-linear stochastic differential equation (SDEs)  
\begin{equation}
\begin{aligned} 
\label{eq:SSE}
 d\mathbf{X}_t
& =
 \left( \mathbf{G} \left( t \right)   \mathbf{X}_{t-} + \mathbf{g} \left ( t,  \mathbf{X}_{t-} \right)   \right) dt
 + \sum_{j=1}^J
  \left( \mathbf{L}_{j}  \left( t \right) - \Re \left( \langle \mathbf{X}_{t-}, \mathbf{L}_{j}  \left( t \right)  \mathbf{X}_{t-} \rangle \right)  \right) \mathbf{X}_{t-} dW^j_t
  \\
& \quad 
 +
 \sum_{m=1}^M \left(  \dfrac{ \mathbf{R}_{m}  \left( t \right) \mathbf{X}_{t-}}{ \left\Vert \mathbf{R}_{m}  
 \left( t \right) \mathbf{X}_{t-} \right\Vert } - \mathbf{X}_{t-} \right) dN^m_t 
\end{aligned}
\end{equation}
with initial datum $\mathbf{X}_0$,
$
t \rightarrow \int_{0}^t  \left\Vert  \mathbf{R}_{m}  \left( s \right) \mathbf{X}_{s-} \right\Vert^2 ds
$
the  predictable compensator of $N^m$,
and
$$
\mathbf{g} \left ( t,  \mathbf{x} \right) 
=
\sum_{j=1}^J \left( 
  \Re \langle  \mathbf{x}, \mathbf{L}_{j}  \left( t \right) \mathbf{x} \rangle \mathbf{L}_{j} \left( t \right) \mathbf{x} - \dfrac{1}{2}\Re^2 \left( \langle \mathbf{x}, \mathbf{L}_{j}  \left( t \right) \mathbf{x} \rangle \right)  \mathbf{x} \right)
+\dfrac{1}{2} \sum_{m=1}^M  \left\Vert  \mathbf{R}_{m}  \left( t \right) \mathbf{x} \right\Vert^2  \mathbf{x} .
$$
In case  \eqref{eq:SSE} does not include the doubly stochastic Poisson noises
(i.e., $M=0$ and $J \geq 1$),
the non-linear stochastic Schr\"odinger equation \eqref{eq:SSE} has been solved 
by combining  finite-dimensional approximations and
projections on the unit sphere with one of the following numerical methods:
\begin{itemize}
 
 \item Versions of the Euler-Maruyama scheme (see, e.g., \cite{BreuerPetruccione2002,Mora2005,Schack1995}).
 
 \item A second order weak It\^o-Taylor scheme \cite{Breuer2000,BreuerPetruccione2002}.
 
 \item An exponential scheme \cite{Mora2005}.
 
\end{itemize}
On the other hand,
if  \eqref{eq:SSE} does not involve Brownian motions
(i.e., $M  \geq 1$ and $J = 0$),
then 
$\mathbf{X}_t$ has been approximated numerically by simulating the quantum jumps
and solving ordinary differential equations
(see, e.g, \cite{Carmichael2008,PlenioKnight1998,WisemanMilburn2009}).

In addition to the applications of \eqref{eq:StochasticQME},
the non-linear stochastic Schr\"odinger equation \eqref{eq:SSE}
plays a key role in the computation of the mean values of the quantum observables
(see, e.g., \cite{BreuerPetruccione2002,Breuer2000,HarocheRaimond2006,Percival1998,Schack1995,WisemanMilburn2009}).
The mean value of the observable $\mathbf{A} \in \boldsymbol{\mathfrak{L}} \left(\boldsymbol{\mathfrak{h}}\right)$ at time $t$ is given by 
the trace of the operator $ \mathbf{A} \, \boldsymbol{\varrho}_{t}$,
where  $\boldsymbol{\varrho}_{t} \in \boldsymbol{\mathfrak{L}} \left(\boldsymbol{\mathfrak{h}}\right)$ satisfies 
the quantum master equation 
\begin{equation}
\label{eq:1.3}
\dfrac{d}{dt} \boldsymbol{\varrho}_{t} = \boldsymbol{\mathcal{L}} \left( t \right)  \boldsymbol{\varrho}_{t} .
\end{equation}
The numerical solution of  \eqref{eq:1.3}
by  schemes for ordinary differential equations
presents drawbacks 
when the dimension $d$ of the Hilbert space required for representing accurately the solution of \eqref{eq:1.3} 
is not small with respect to the available computational resource
(see, e.g., \cite{HarocheRaimond2006,Mora2005,Percival1998,Schack1995,WisemanMilburn2009}).
For instance,
a current notebook  can compute the explicit solution of \eqref{eq:1.3} using Pad\'e approximants
only when $d$ is in the range of tens.
It is a common practice to obtain $Tr \left( \mathbf{A} \, \boldsymbol{\varrho}_{t} \right) $
by computing the solution to \eqref{eq:SSE},
and using the fact that 
$$
\mathbb{E} \langle \mathbf{X}_t , \mathbf{A} \mathbf{X}_t \rangle
=
Tr \left(   \mathbf{A} \, \boldsymbol{\varrho}_{t} \right) 
$$
whenever $\mathbf{X}_0$ is such that
$\boldsymbol{\varrho}_{0} =  \mathbb{E} \ketbra{\mathbf{X}_0}{\mathbf{X}_0}$
(see, e.g., \cite{BarchielliGregoratti2009,MoraJFA,MoraAP,Percival1998}).
This method is called unraveling of \eqref{eq:1.3}.
Two advantages of  solving \eqref{eq:SSE} rather than \eqref{eq:1.3}
are the following: 
(i) the number of unknown complex functions in \eqref{eq:1.3} is the square of 
the components of the vector $\mathbf{X}_t$;
and 
(ii) in many physical systems the values of $\mathbf{X}_t$ are localized into small time-dependent  regions of  $\mathfrak{h}$
(see, e.g, \cite{Percival1998,Schack1995}).
For instance, 
using a current desktop computer 
we can get $\mathbb{E} \langle \mathbf{X}_t , \mathbf{A} \mathbf{X}_t \rangle$,
with $M=0$, 
even if the required basis has thousand of elements.

Now, suppose that the dimension of $ \boldsymbol{\mathfrak{h}}$ is finite, 
i.e., $\dim \left(  \boldsymbol{\mathfrak{h}} \right) < + \infty$,
and that $\boldsymbol{\rho}_{0}$ is a random mixed state,
i.e., $\boldsymbol{\rho}_{0}$ is not a random pure state.
We can transform \eqref{eq:StochasticQME} into a system of 
$ \dim \left( \boldsymbol{\mathfrak{h}} \right)^2$ complex SDEs,
and hence we can compute $\boldsymbol{\rho}_{t}$
by applying classical numerical schemes for SDEs with multiplicative noise 
(see, e.g., \cite{GrahamTalay2013,Kloeden1992,Milstein2004}).
Worse than the mentioned approach of solving \eqref{eq:1.3},
this procedure presents scale issues if $ \dim \left( \boldsymbol{\mathfrak{h}} \right)$ is not in the range of tens,
together with difficulties to yield semi-positive definite  numerical solutions.
In the pure diffusive case (i.e., $M=0$),
Amini, Mirrahimi and Rouchon \cite{Amini2011} introduced 
Scheme \ref{scheme:AminiMirrahimiRouchon} given in Section \ref{subsec:DiffusiveSQMEs}
that is a numerical method tailored to the specific characteristics of \eqref{eq:StochasticQME} 
(see also \cite{Rouchon2015}).
Scheme \ref{scheme:AminiMirrahimiRouchon} preserves the positivity of $\boldsymbol{\rho}_{t}$,
but is inaccurate and slow in our numerical experiments with high dimensional physical systems
(see, e.g., Section \ref{subsec:DiffusiveSQMEs}).

This paper develops a quantum trajectory approach for solving 
\eqref{eq:StochasticQME} with $\dim \left(  \boldsymbol{\mathfrak{h}} \right) < + \infty$.
This makes possible not only the simulation of quantum systems with finite dimensional state spaces 
but also the approximate solution of infinite dimensional stochastic master equations  (see, e.g., Remark \ref{rem:aproximacion}).
In Section \ref{sec:Representation},
we extend the representation \eqref{eq:1.2} to random mixed initial states.
Roughly speaking,
we deduce that
$$
\boldsymbol{\rho}_{t} =  \sum_{k} \ketbra{\mathbf{X}^k_t}{\mathbf{X}^k_t} 
\hspace{2cm}
\forall t \geq 0 ,
$$
where the $\boldsymbol{\mathfrak{h}}$-valued stochastic processes $\mathbf{X}^k_t$'s satisfy a system of 
weakly coupled complex stochastic differential equations of type \eqref{eq:SSE}.
In Section \ref{sec:NumericalSolution},
we design exponential schemes for computing the $\mathbf{X}^k_t$'s,
which are novel even for \eqref{eq:SSE}.
Section \ref{sec:Simulation results} illustrates the very good numerical performance of the new exponential methods
by simulating a quantized electromagnetic field coupled to a two-level system,
which interact among them and with the reservoir.
Section \ref{sec:Proofs} is devoted to the proofs of our theoretical results,
and Section \ref{sec:Conclusion} presents the conclusions.


\section{Representation of the solution to the stochastic master equation}
\label{sec:Representation}

For simplicity, 
we here assume that  the dimension of the space state $\boldsymbol{\mathfrak{h}}$ is finite
(see, e.g., Remark \ref{rem:aproximacion}),
and that
$\mathbf{H}, \mathbf{L}_{j}, \mathbf{R}_{m} : \left[ 0, + \infty \right[ \rightarrow  \boldsymbol{\mathfrak{L}} \left(\boldsymbol{\mathfrak{h}}\right)$
are continuous functions.
Let 
\begin{equation}
 \label{eq:3.1}
 \boldsymbol{\rho}_{0} = \sum_{k=1}^{\mu} \ketbra{\mathbf{X}_0^k}{\mathbf{X}_0^k}
\end{equation}
where 
$\mu \in \mathbb{N}$ 
and 
$\mathbf{X}_0^1,\ldots,\mathbf{X}_0^{\mu}$ are $\mathfrak{F}_0$-measurable $\boldsymbol{\mathfrak{h}}$-valued random variables
satisfying
$
 \sum_{k=1}^{\mu} \left\Vert \mathbf{X}_0^k \right\Vert ^2 = 1
$.
 In practice, there is no loss of generality in assuming \eqref{eq:3.1}.
In many situations,
using the physical meaning of  $\boldsymbol{\rho}_{0}$ we obtain
\begin{equation}
 \label{eq:3.2}
 \boldsymbol{\rho}_{0} = \sum_{k=1}^{\mu} p^k \ketbra{\mathbf{Y}^k}{\mathbf{Y}^k} ,
\end{equation}
with $\mathbf{Y}^k \in \boldsymbol{\mathfrak{h}}$, $\left\Vert \mathbf{Y}^k \right\Vert = 1$, $p^k > 0$
and  $ \sum_{k=1}^{\mu}  p^k  = 1 $,
that is, $ \boldsymbol{\rho}_{0}$ can be represented 
as the mixture of the quantum states $\ketbra{\mathbf{Y}^k}{\mathbf{Y}^k}$ with probability $p^k$.
Taking $\mathbf{X}_0^k = \sqrt{p^k} \, \mathbf{Y}^k$ we obtain \eqref{eq:3.1}.
Otherwise,  
conditioning on $\mathfrak{F}_{0}$ leads to solve \eqref{eq:StochasticQME}
with  the initial density operator $\boldsymbol{\rho}_0$ deterministic,
and so applying the spectral decomposition of $\boldsymbol{\rho}_0$ yields \eqref{eq:3.2},
where $\mu$ is less than or equal to the dimension of $ \boldsymbol{\mathfrak{h}}$,
$p^1, p^2, \ldots, p^{\mu}$ are the positive eigenvalues of 
the non-negative operator $\boldsymbol{\rho}_{0}$
and 
$\mathbf{Y}^1,\ldots, \mathbf{Y}^{\mu}$ are the orthonormal eigenvectors of $\boldsymbol{\rho}_{0}$.

We associate to \eqref{eq:3.1} 
the following system of classical stochastic differential equations  in $\boldsymbol{\mathfrak{h}}$:
\begin{equation}  
\label{eq:SistemSSE}
\begin{aligned}
 \mathbf{X}^k _t   
& =   
\mathbf{X}^k _0
+ \int_{0}^t \mathbf{G} \left ( s-, \mathbf{X}^k_{s-}  \right)  ds
 + \sum_{j=1}^J \int_{0}^t \left( \mathbf{L}_{j}  \left( s \right) \mathbf{X}^k_{s-}  - \ell_j \left( s- \right) \mathbf{X}^k_{s-} \right) dW^j_s
  \\
 & \quad 
 +  \sum_{m=1}^M \int_{0+}^t   \left(  \dfrac{ \mathbf{R}_{m}  \left( s \right) \mathbf{X}^k_{s-}}{ \sqrt{r_{m} \left( s - \right)} } - \mathbf{X}^k_{s-} \right) dN^m_s ,
\end{aligned}
\end{equation}
where
$k=1,\ldots, \mu$,
$ r_{m} \left( s \right) := \sum_{k=1}^{\mu}  \left\Vert  \mathbf{R}_{m}  \left( s \right) \mathbf{X}^k_{s} \right\Vert^2$,
the $N^m$'s are doubly stochastic Poisson processes with predictable compensator 
$
t \rightarrow \int_{0}^t r_{m} \left( s - \right) ds 
$,
$$
\mathbf{G} \left ( s , \mathbf{x} \right)
 =
\mathbf{G} \left ( s \right)  \mathbf{x}
 + \sum_{j=1}^J \left( \ell_j \left( s \right)  \mathbf{L}_{j} \left( s \right) \mathbf{x}  - \dfrac{1}{2}\ell_j \left( s \right) ^2  \mathbf{x} \right)
 + \dfrac{1}{2} \sum_{m=1}^M  \mathbf{R}_{m} \left( s \right) \mathbf{x}
 $$
and 
$\ell_j \left( s \right)  := \sum_{k=1}^{\mu}  \Re \langle  \mathbf{X}^k_{s}, \mathbf{L}_{j}  \left( s \right) \mathbf{X}^k_{s} \rangle$.
The equations of \eqref{eq:SistemSSE} are coupled by
the functions $\ell_j$ and $r_{m}$.
We next represent the solution of \eqref{eq:StochasticQME} 
by means of the generalized non-linear stochastic Schr\"odinger equations \eqref{eq:SistemSSE}.

\begin{theorem}
\label{th:Representation}
Suppose that the dimension of $\boldsymbol{\mathfrak{h}}$ is finite.
Let $\mathbf{X}^1_t, \ldots, \mathbf{X}^{\mu}_t$ satisfy the system \eqref{eq:SistemSSE}.
Set
\begin{equation}
\label{eq:Representation}
 \boldsymbol{\rho}_{t} =  \sum_{k=1}^{\mu}  \ketbra{\mathbf{X}^k_t}{\mathbf{X}^k_t} 
\hspace{2cm}
\forall t \geq 0 .
\end{equation}
Then $\left( \boldsymbol{\rho}_t \right)_{t \geq 0}$ solves 
the stochastic master equation \eqref{eq:StochasticQME} with initial datum \eqref{eq:3.1}.
\end{theorem}

\begin{proof}
 Deferred to Subsection \ref{sec:Proof:Representation}.
\end{proof}

\begin{remark}
\label{rem:PureInitialState}
Let $\mu =1$, that is, $\boldsymbol{\rho}_{0}$ is the random pure state $\ketbra{\mathbf{X}_0}{\mathbf{X}_0} $.
Then \eqref{eq:SistemSSE}  becomes \eqref{eq:SSE} with initial condition $\mathbf{X}_0$.
According to Theorem \ref{th:Representation} we have  
$\boldsymbol{\rho}_{t} = \ketbra{\mathbf{X}_t}{\mathbf{X}_t}$,
which is a well-known relation (see, e.g., \cite{BarchielliBelavkin1991}). 
\end{remark}

\begin{remark}
Using  Girsanov's theorem,
we can obtain one solution of \eqref{eq:StochasticQME},
also called Belavkin equation,
as the normalized solution to certain linear SDE 
after changing the original probability measure
(see, e.g., \cite{BarchielliHolevo1995} for details).
Pellegrini \cite{Pellegrini2010,PellegriniIHP2010} proves 
the existence and uniqueness of the solution to the following version of  \eqref{eq:StochasticQME}:
 \begin{align*}
 d\boldsymbol{\rho}_{t}	
 & =	
 \boldsymbol{\rho}_{0} 
 + \int_{0}^t \left( \boldsymbol{\mathcal{L}} \left( s \right)  \boldsymbol{\rho}_{s-} 
 + \tr \left( \mathbf{R}_{m} \left( s \right) ^*\mathbf{R}_{m} \left( s \right) \boldsymbol{\rho}_{s-} \right) \boldsymbol{\rho}_{s-}
 - \mathbf{R}_{m} \left( s \right) \boldsymbol{\rho}_{s-}  \mathbf{R}_{m} \left(  s \right) ^*
 \right) ds 
 \\
& \quad 
+ \sum_{j =1}^J 
\int_{0}^t  \left( \mathbf{L}_{j} \left( s \right) \boldsymbol{\rho}_{s-} + \boldsymbol{\rho}_{s-} \mathbf{L}_{j} \left( s \right)^* 
 - 2 \Re  \left(  \tr \left( \mathbf{L}_{j} \left(  s \right) \boldsymbol{\rho}_{s-} \right) \right)  \boldsymbol{\rho}_{s-} \right) dW^{j}_{s}
\\
& \quad
+ 
\sum_{m=1}^M  \int_{0+}^t \int_{\mathbb{R}_+}
 \left( 
\frac{ \mathbf{R}_{m} \left( s \right) \boldsymbol{\rho}_{s-}  \mathbf{R}_{m} \left(  s \right) ^*}
{ \tr \left( \mathbf{R}_{m} \left( s \right) ^*\mathbf{R}_{m} \left( s \right) \boldsymbol{\rho}_{s-} \right) } 
- \boldsymbol{\rho}_{s-} \right)  \mathbf{1}_{0 \leq x \leq \tr \left( \mathbf{R}_{m} \left( s \right) ^*\mathbf{R}_{m} \left( s \right) \boldsymbol{\rho}_{s-} \right) }
N^m \left( ds, dx \right) ,
\end{align*}
where  $ N^1, \ldots, N^M$ are independent adapted Poisson point processes of intensity $dt \otimes dx$
that are independent of $W^1,\ldots, W^J$.
Here, $\boldsymbol{\rho}_{t}$ satisfies \eqref{eq:StochasticQME} with 
 $
 N^{m}_{t} 
 = 
 \int_{0}^t \int_{\mathbb{R}} 
  \mathbf{1}_{0 \leq x \leq \tr \left( \mathbf{R}_{m} \left( s \right) ^*\mathbf{R}_{m} \left( s \right) \boldsymbol{\rho}_{s-} \right) }
  N^m \left( ds, dx \right)
$.
Analysis similar to that in \cite{Pellegrini2010,PellegriniIHP2010} yields 
the existence and uniqueness of the following version of \eqref{eq:SistemSSE}:
for all $k=1,\ldots, \mu$,
\begin{align*}
 \mathbf{X}^k _t   
& =   
\mathbf{X}^k _0
+ \int_{0}^t \mathbf{G} \left ( s-, \mathbf{X}^k_{s-}  \right)  ds
 + \sum_{j=1}^J \int_{0}^t \left( \mathbf{L}_{j}  \left( s \right) \mathbf{X}^k_{s-}  - \ell_j \left( s- \right) \mathbf{X}^k_{s-} \right) dW^j_s
  \\
 & \quad 
 +  \sum_{m=1}^M \int_{0+}^t   \int_{\mathbb{R}_+}
 \left(  \dfrac{ \mathbf{R}_{m}  \left( s \right) \mathbf{X}^k_{s-}}{ \sqrt{r_{m} \left( s - \right)} } - \mathbf{X}^k_{s-} \right) 
 \mathbf{1}_{0 \leq x \leq \mathbf{R}_{m} \left( s - \right) } N^m \left( ds, dx \right) .
 \end{align*}
\end{remark}

\begin{remark}
\label{rem:aproximacion}
The dynamics of quantum physical systems with infinite-dimensional state space
can be approximated by finite-dimensional stochastic quantum master equations. 
Indeed,
we can approximate \eqref{eq:StochasticQME} with $\dim \left(  \boldsymbol{\mathfrak{h}} \right) = \infty$
by finite-dimensional stochastic master equations
in a similar way to that of the numerical solution of the quantum master equations
(see, e.g., \cite{MoraMC2004}).
To this end,
we should find a finite dimensional subspace $\boldsymbol{\mathfrak{h}}_d$ of $\boldsymbol{\mathfrak{h}}$
such that 
$\mathbf{H} \left( t \right) \approx \overline{\mathbf{H}} \left( t \right) := \mathbf{P}_d \, \mathbf{H} \left( t \right) \mathbf{P}_d$,
$ \mathbf{L}_{j} \left( t \right) \approx \overline{\mathbf{L}}_{j} \left( t \right) := \mathbf{P}_d \, \mathbf{L}_{j} \left( t \right) \mathbf{P}_d$
and
$\mathbf{R}_{j} \left( t \right) \approx \overline{\mathbf{R}}_{j} \left( t \right) := \mathbf{P}_d \, \mathbf{R}_{j} \left( t \right) \mathbf{P}_d$,
where 
$\mathbf{P}_d$ is the orthogonal projection of  $\boldsymbol{\mathfrak{h}}$ onto $\boldsymbol{\mathfrak{h}}_d$.
Then $\boldsymbol{\rho}_{t}$ is approximated by the solution of \eqref{eq:StochasticQME}
with initial datum $\mathbf{P}_d \boldsymbol{\rho}_0 \mathbf{P}_d$
and 
$\mathbf{H} \left( t \right)$, $ \mathbf{L}_{j} \left( t \right) $, $\mathbf{R}_{j} \left( t \right)$
replaced by the operators
$\overline{\mathbf{H}} \left( t \right) $, 
$ \overline{\mathbf{L}}_{j} \left( t \right) $, 
$\overline{\mathbf{R}}_{j} \left( t \right) $.
\end{remark}

\section{Numerical solution of the quantum master equations}
\label{sec:NumericalSolution}

\subsection{The Euler-exponential scheme}
\label{subsec:EulerExponential}

In this section we develop 
the numerical solution of the system of non-linear stochastic Schr\"odinger equations \eqref{eq:SistemSSE}
with $\boldsymbol{\mathfrak{h}}$ finite-dimensional.
To this end,
we consider the time discretization 
$ 0 = T_0 < T_1 < T_2 < \cdots $ of the interval $\left[ 0 , \infty \right[ $,
where the $T_n$'s are stopping times. 
Suppose that $ t \in \left[ T_n , T_{n+1} \right] $.
Decomposing  
$\mathbf{G} \left( t \right) = \mathbf{G} \left( T_{n} \right) + \left( \mathbf{G} \left( t \right) - \mathbf{G} \left( T_{n} \right) \right)$
we obtain from \eqref{eq:SistemSSE} that 
$$
\mathbf{X}^k_t
=
\mathbf{S}^k_t + \int_{T_n}^t \mathbf{G} \left( T_{n} \right)   \mathbf{X}^k_{s-}  ds,
$$
where $k=1,\ldots, \mu$ and 
\begin{align*}
\mathbf{S}^k_t
& =
\mathbf{X}^k_{T_n}
+
\int_{T_n }^{t}  \left(  \left( \mathbf{G} \left ( s \right)  - \mathbf{G} \left( T_{n} \right)  \right) \mathbf{X}^k_{s-}  +  \widetilde{\mathbf{g}}  \left( s-,  \mathbf{X}_{s-} \right)   \right) ds
 \\
 & \quad 
 + \sum_{j=1}^J \int_{T_n}^t \left( \mathbf{L}_{j}  \left( s \right) \mathbf{X}^k_{s-}  - \ell_j \left( s - \right) \mathbf{X}^k_{s-} \right) dW^j_s 
 +  \sum_{m=1}^M \int_{T_n+}^t   \left(  \dfrac{ \mathbf{R}_{m}  \left( s \right) \mathbf{X}^k_{s-}}{ \sqrt{r_{m} \left( s - \right)} } - \mathbf{X}^k_{s-} \right) dN^m_s 
\end{align*}
with 
$
\widetilde{\mathbf{g}}  \left( s, \mathbf{x} \right)
=
 \sum_{j=1}^J \left( \ell_j \left( s \right)  \mathbf{L}_{j} \left( s \right) \mathbf{x}  - \ell_j \left( s \right) ^2  \mathbf{x}/2   \right)
 +  \sum_{m=1}^M  \mathbf{R}_{m} \left( s \right) \mathbf{x}/2 .
$
Hence
$$
\mathbf{X}^k_t 
=
\exp \left( \mathbf{G} \left ( T_n \right)   \left( t - T_n  \right) \right) \mathbf{S}^k_{T_n} 
+
\int_{T_n +}^t \exp \left( \mathbf{G} \left ( T_n \right)   \left( t - s  \right) \right) d \mathbf{S}^k_s
$$
(see, e.g., \cite{Protter2005}).
This gives
\begin{equation}
 \label{eq:2.5}
 \begin{aligned}
 \mathbf{X}^k_t
 = &
\exp{\left( \mathbf{G} \left ( T_n \right)   \left( t - T_n  \right)  \right )} \mathbf{X}^k_{T_n}
\\
&+
\int_{T_n }^{t}  
\exp{\left( \mathbf{G} \left ( T_n \right)   \left( t - s  \right)\right )  }
\left(  \left( \mathbf{G} \left ( s \right)   - \mathbf{G} \left( T_{n} \right)  \right) \mathbf{X}^k_{s-}  +  \widetilde{\mathbf{g}}  \left( s -,  \mathbf{X}^k_{s-} \right)   \right) ds
 \\
 & 
 + \sum_{j=1}^J
\int_{T_n+}^{t}  
 \exp{ \left( \mathbf{G} \left ( T_n \right)   \left( t - s  \right) \right ) } \left( \mathbf{L}_{j}  \left( s \right) \mathbf{X}^k_{s-}  - \ell_j \left( s - \right) \mathbf{X}^k_{s-} \right) 
dW^j_s
\\
&
 +  
 \sum_{m=1}^M 
 \int_{T_n+}^{t}  
\exp{ \left( \mathbf{G} \left ( T_n \right)   \left( t - s  \right)\right )  }  \left(  \dfrac{ \mathbf{R}_{m}  \left( s \right) \mathbf{X}^k_{s-}}{ \sqrt{r_{m} \left( s - \right)} } - \mathbf{X}^k_{s-} \right)
dN^m_s  .
\end{aligned}
\end{equation}

We now approximate the terms in \eqref{eq:2.5}.
Applying the Euler approximation yields
\begin{equation}
\label{eq:2.2}
\begin{aligned}
&
\int_{T_n }^{T_{n+1}}  
\exp{\left( \mathbf{G} \left ( T_n \right)   \left( T_{n+1} - s  \right) \right ) }
\left(  \left( \mathbf{G} \left ( s \right)   - \mathbf{G} \left( T_{n} \right)  \right) \mathbf{X}^k_{s-}  +  \widetilde{\mathbf{g}}  \left( s -,  \mathbf{X}^k_{s-} \right)   \right) ds
\\
& \quad \approx
\exp{ \left( \mathbf{G} \left ( T_n \right)   \left( T_{n+1} - T_n  \right) \right ) }  \widetilde{\mathbf{g}}  \left( T_n,  \mathbf{X}^k_{T_n}  \right) \left(  T_{n+1} - T_n  \right)
\end{aligned}
\end{equation}
and 
\begin{equation}
 \label{eq:2.1}
\begin{aligned}
 & \int_{T_n+}^{ T_{n+1} }  
 \exp{\left( \mathbf{G} \left ( T_n \right)   \left( T_{n+1} - s  \right)  \right )} \left( \mathbf{L}_{j}  \left( s \right) \mathbf{X}^k_{s-}  - \ell_j \left( s - \right) \mathbf{X}^k_{s-} \right) 
dW^j_s
\\
& \quad  \approx
\exp{\left( \mathbf{G} \left ( T_n \right)   \left( T_{n+1} - T_n  \right)  \right )} 
\left( \mathbf{L}_{j}  \left( T_n \right) \mathbf{X}^k_{T_n}  - \ell_j \left( T_n \right) \mathbf{X}^k_{T_n} \right) 
\left( W^j_{T_{n+1}} - W^j_{T_{n}}  \right) .
\end{aligned}
\end{equation}
Moreover, 
in case $r_{m} \left( T_n \right) \neq 0$ we have
\begin{equation}
 \label{eq:2.3}
\begin{aligned}
&
\int_{T_n+}^{ T_{n+1} }  
\exp{ \left( \mathbf{G} \left ( T_n \right)   \left( T_{n+1} - s  \right) \right ) }  \left(  \dfrac{ \mathbf{R}_{m}  \left( s \right) \mathbf{X}^k_{s-}}{ \sqrt{r_{m} \left( s - \right)} } - \mathbf{X}^k_{s-} \right)
dN^m_s 
\\
& \quad  \approx
\exp{\left( \mathbf{G} \left ( T_n \right)   \left( T_{n+1} - T_n  \right) \right) }  
\left(  \dfrac{ \mathbf{R}_{m}  \left( T_n \right) \mathbf{X}^k_{ T_n }}{ \sqrt{r_{m} \left( T_n \right)} } - \mathbf{X}^k_{T_n} \right)
\left( N^m_{T_{n+1}}  - N^m_{T_n}  \right) .
\end{aligned}
\end{equation}
If $r_{m} \left( T_n \right) = 0$,
then 
$
N^m_{s}  - N^m_{T_n} \approx 0
$
for all $s \in \left[ T_n , T_{n+1} \right]$, and so
\begin{equation}
\label{eq:2.10}
\int_{T_n+}^{ T_{n+1} }  
\exp{\left( \mathbf{G} \left ( T_n \right)   \left( T_{n+1} - s  \right) \right ) }  \left(  \dfrac{ \mathbf{R}_{m}  \left( s \right) \mathbf{X}^k_{s-}}{ \sqrt{r_{m} \left( s - \right)} } - \mathbf{X}^k_{s-} \right)
dN^m_s 
\approx 0 .
\end{equation}

Suppose that we are able to get the approximations 
$\bar{\mathbf{X}}^1_n, \ldots, \bar{\mathbf{X}}^{\mu}_n $
of the solution of \eqref{eq:SistemSSE} at time $T_n$.
To make this more precise we use the symbol $\approx_w$ to denote 
the approximation in the weak sense 
(see, e.g., \cite{GrahamTalay2013,Kloeden1992,Milstein2004}),
that is,  
$\mathbf{Y} \approx_w \mathbf{Z}$ means, roughly speaking, that
the distributions of the random variables  $\mathbf{Y}$, $\mathbf{Z}$ approximate each other.
Then,
we actually assume that
$\bar{\mathbf{X}}^1_n, \ldots, \bar{\mathbf{X}}^{\mu}_n $ are $\mathfrak{F}_{T_n}$-measurable random variables 
such that 
$
 \sum_{k=1}^{\mu} \left\Vert \bar{\mathbf{X}}_n^k \right\Vert ^2 = 1
$
and 
$\bar{\mathbf{X}}^k_n \approx_w \mathbf{X}^k_{T_n}$ for any $k=1,\ldots, \mu$.
Therefore,
$
r_{m} \left( s \right) \approx_w \bar{r}_n^m 
:= \sum_{k=1}^{\mu}  \left\Vert  \mathbf{R}_{m}  \left( T_n \right) \bar{\mathbf{X}}^k_n \right\Vert^2 
$
and 
$
\ell_j \left( s \right) \approx_w \bar{\ell}^n_j 
:= \sum_{k=1}^{\mu}  \Re \langle  \bar{\mathbf{X}}^k_n, \mathbf{L}_{j}  \left( T_n \right) \bar{\mathbf{X}}^k_n \rangle
$
for all 
$s \in \left[ T_n , T_{n+1} \right] $.

Next,
we simulate numerically the right-hand side of \eqref{eq:2.3} and \eqref{eq:2.10}.
In order to approximate the increments 
$ \left( W^j_{T_{n+1}} - W^j_{T_{n}} \right) / \sqrt{ T_{n+1} - T_{n} }$
and 
$N^m_{T_{n+1}}  - N^m_{T_n} $,
we consider the random variables 
$\overline{\delta W}^1_{n+1}, \ldots$, $\overline{\delta W}^J_{n+1}$,
$\overline{\delta N}^1_{n+1}, \ldots $, $\overline{\delta N}^{M}_{n+1} $,
which are  conditionally independent relative to $\mathfrak{F}_{T_n}$
(see, e.g., \cite{Chow1997,Chung2001}),  such that:
(i) $
 \overline{\delta W}^1_{n+1}, \ldots, \overline{\delta W}^J_{n+1}$
are identically distributed symmetric random variables with variance $1$
that are independent of $\mathfrak{F}_{T_n}$; and 
(ii) 
the conditional distribution of $\overline{\delta N}^m_{n+1}$ given $\mathfrak{F}_{T_n}$ 
is the Poisson law with parameter 
$ \bar{r}_n^m  \left( T_{n+1} - T_n \right) $
for any $m=1,\ldots,M$.
Since
$$
 \int_{T_n}^{t} r_{m} \left( s - \right) ds
 \approx 
 r_{m} \left( T_n \right)  \left( t - T_n \right)
 \approx_w 
 \bar{r}_n^m  \left( t - T_n \right)
\hspace{1cm}
\forall  t \in \left[ T_n , T_{n+1} \right] ,
$$
$ \overline{\delta N}^m_{n+1} \approx_w N^m_{T_{n+1}}  - N^m_{T_n} $.
Hence, 
combining \eqref{eq:2.3} with \eqref{eq:2.10} gives
\begin{equation}
 \label{eq:2.9}
\begin{aligned}
&
\int_{T_n+}^{ T_{n+1} }  
\exp{\left( \mathbf{G} \left ( T_n \right)   \left( T_{n+1} - s  \right)  \right )}  \left(  \dfrac{ \mathbf{R}_{m}  \left( s \right) \mathbf{X}^k_{s-}}{ \sqrt{r_{m} \left( s - \right)} } - \mathbf{X}^k_{s-} \right)
dN^m_s 
\\
& \quad \approx_w
\exp{\left( \mathbf{G} \left ( T_n \right)   \left( T_{n+1} - T_n  \right) \right ) }  
\left(  \dfrac{ \mathbf{R}_{m}  \left( T_n \right) \bar{\mathbf{X}}^k_n }{ \sqrt{\bar{r}_n^m }} -  \bar{\mathbf{X}}^k_n \right)
\overline{\delta N}^m_{n+1} ,
\end{aligned}
\end{equation}
where, for simplicity of notation,
we define 
$
\left(  { \mathbf{R}_{m}  \left( T_n \right) \bar{\mathbf{X}}^k_n }/{ \sqrt{\bar{r}_n^m }} -  \bar{\mathbf{X}}^k_n \right)
\overline{\delta N}^m_{n+1}
$
to be $0$ provided that $\bar{r}_n^m = 0$.

Using  \eqref{eq:2.5}, \eqref{eq:2.2}, \eqref{eq:2.1}, \eqref{eq:2.9},
together with 
$$
\ell_j \left( T_n \right) \approx_w \bar{\ell}^n_j 
:= \sum_{k=1}^{\mu}  \Re \langle  \bar{\mathbf{X}}^k_n, \mathbf{L}_{j}  \left( T_n \right) \bar{\mathbf{X}}^k_n \rangle ,
$$
we deduce that 
$\mathbf{X}^k_{T_{n+1}} \approx_w \hat{\mathbf{X}}^k_{n+1}$ for all  $k=1,\ldots, \mu$,
where 
\begin{align*}
 \hat{\mathbf{X}}^k_{n+1}
& =
\exp{ \left( \mathbf{G} \left ( T_n \right)   \left( T_{n+1} - T_n  \right) \right ) } \bar{\mathbf{X}}^k_n
+
\left( T_{n+1} - T_n  \right) \exp{\left ( \mathbf{G} \left ( T_n \right)   \left( T_{n+1} - T_n  \right) \right ) }  \bar{\mathbf{g}}^k_n   
 \\
 & \quad 
 + \sum_{j=1}^J
 \exp{\left( \mathbf{G} \left ( T_n \right)   \left( T_{n+1} - T_n  \right) \right ) } \left( \mathbf{L}_{j}  \left( T_n \right) \bar{\mathbf{X}}^k_n  - \bar{\ell}_n^j  \, \bar{\mathbf{X}}^k_n \right) 
\sqrt{T_{n+1} - T_n} \,  \overline{\delta W}^j_{n+1} 
\\
& \quad   
+  \sum_{m=1}^M  
\exp{\left( \mathbf{G} \left ( T_n \right)   \left( T_{n+1} - T_n  \right) \right ) }  \left(  \dfrac{ \mathbf{R}_{m}  \left( T_n \right) \bar{\mathbf{X}}^k_n }{ \sqrt{\bar{r}_n^m }} -  \bar{\mathbf{X}}^k_n \right)
\overline{\delta N}^m_{n+1} 
\end{align*}
with
$
\bar{\mathbf{g}}^k_n 
= 
\sum_{j=1}^J \left( \bar{\ell}_n^j  \, \mathbf{L}_j  \left( T_n \right) \bar{\mathbf{X}}^k_n  - {1}/{2} \left(\bar{\ell}_n^j\right)  ^2  \bar{\mathbf{X}}^k_n  \right)
 + {1}/{2}\, \sum_{m=1}^M  \bar{r}_n^m \, \bar{\mathbf{X}}^k_n 
$.
Therefore,
$$
\hat{\boldsymbol{\rho}}_{n+1}
:=
\sum_{k=1}^{\mu}   \ketbra{ \hat{\mathbf{X}}^k_{n+1} }{\hat{\mathbf{X}}^k_{n+1}}
\approx_w
\sum_{k=1}^{\mu}   \ketbra{ \mathbf{X}^k_{T_{n+1}} }{ \mathbf{X}^k_{T_{n+1}} } ,
$$
and so Theorem \ref{th:Representation} yields 
$
\hat{\boldsymbol{\rho}}_{n+1}
\approx_w
\boldsymbol{\rho}_{T_{n+1}}$.

In order to preserve the important physical property $\tr \left( \boldsymbol{\rho}_t \right) =1 $,
we normalize  $\hat{\boldsymbol{\rho}}_{n+1} $ by $\tr \left( \hat{\boldsymbol{\rho}}_{n+1} \right)$.
Thus, using $\tr \left( \hat{\boldsymbol{\rho}}_{n+1} \right) = \sum_{k=1}^{\mu}  \left\Vert \hat{\mathbf{X}}^k_{n+1} \right\Vert^2 $ 
we deduce that 
$$
\boldsymbol{\rho}_{T_{n+1}} 
\approx_w
\sum_{k=1}^{\mu}   \ketbra{ \frac{ \hat{\mathbf{X}}^k_{n+1}}{\sqrt{\sum_{k=1}^{\mu}  \left\Vert \hat{\mathbf{X}}^k_{n+1} \right\Vert^2}} }{\frac{ \hat{\mathbf{X}}^k_{n+1}}{\sqrt{\sum_{k=1}^{\mu}  \left\Vert \hat{\mathbf{X}}^k_{n+1} \right\Vert^2}}} .
$$
Summarizing, we have obtained the following numerical scheme of exponential type.
\begin{scheme}
 \label{scheme:SME_Euler_Exponential}
Suppose that the random variables $\bar{\mathbf{X}}^1_{0}, \ldots, \bar{\mathbf{X}}^{\mu}_{0}$
with values in $\boldsymbol{\mathfrak{h}}$ satisfy 
$
 \sum_{k=1}^{\mu} \left\Vert \bar{\mathbf{X}}_0^k \right\Vert ^2 = 1
$. 
Let 
$\overline{\delta W}^1_1, \overline{\delta W}^2_1,  \ldots,$ $\overline{\delta W}^J_1, \overline{\delta W}^1_2, \ldots$ 
be i.i.d. symmetric real random variables with variance $1$
that are independent of $\bar{\mathbf{X}}^1_{0}, \ldots, \bar{\mathbf{X}}^{\mu}_{0}$.
For any $n \geq 0$, we approximate $\boldsymbol{\rho}_{T_n} $ by 
$
\bar{\boldsymbol{\rho}}_{n} := 
\sum_{k=1}^{\mu}   \ketbra{ \bar{\mathbf{X}}^k_{n} }{\bar{\mathbf{X}}^k_{n}  } ,
$
where 
the $\boldsymbol{\mathfrak{h}}$-valued random vectors $\bar{\mathbf{X}}^1_{n} , \ldots, \bar{\mathbf{X}}^{\mu}_{n}$,
together with the real random variables 
$\overline{\delta N}^1_{1}, \ldots, \overline{\delta N}^{M}_{1}$, $\overline{\delta N}^{1}_{2} , \ldots $,
are defined recursively as follows:
\begin{itemize}

\item
Generate the  random variables $\overline{\delta N}^1_{n+1}, \ldots, \overline{\delta N}^{M}_{n+1} $
such that: 
\begin{description}
 \item[(i)] $\overline{\delta N}^1_{n+1}, \ldots$, $\overline{\delta N}^{M}_{n+1} $
are conditionally  independent relative to the $\sigma$-algebra  $\mathfrak{G}_{n}$ generated by 
$\bar{\mathbf{X}}^1_{0}, \ldots$, $\bar{\mathbf{X}}^{\mu}_{0}$, $\overline{\delta W}^1_1, \ldots$, $\overline{\delta W}^J_n$, 
$\overline{\delta N}^1_{1}, \ldots, \overline{\delta N}^{M}_{n}$.

\item[(ii)] $\overline{\delta N}^1_{n+1}, \ldots, \overline{\delta N}^{M}_{n+1} $ are independent of 
the  $\sigma$-algebra  generated by $\overline{\delta W}^1_{n+1}, \ldots$, $\overline{\delta W}^J_{n+1}$.

\item[(iii)] For any $m=1,\ldots,M$,
the conditional distribution of $\overline{\delta N}^m_{n+1}$ given $\mathfrak{G}_{n}$
is the Poisson law with parameter 
$ \left( T_{n+1} - T_n \right) \bar{r}_n^m $, 
where
$
\bar{r}_n^m = \sum_{k=1}^{\mu}   \left\Vert  \mathbf{R}_{m}  \left( T_n \right) \bar{\mathbf{X}}^k_n \right\Vert^2 
$.
\end{description}

\item  For any $k=1,\ldots, \mu$ we choose 
$$
\bar{\mathbf{X}}^k_{n+1} 
 = 
\begin{cases}
\bar{\mathbf{X}}^k_{n} 
&
\text{if } 
\hat{\mathbf{X}}^1_{n+1}= \ldots = \hat{\mathbf{X}}^{\mu}_{n+1} = 0 ,
\\
\hat{\mathbf{X}}^k_{n+1} / \sqrt{\sum_{k=1}^{\mu}  \left\Vert \hat{\mathbf{X}}^k_{n+1} \right\Vert^2} 
&
\text{otherwise},
\end{cases} 
$$
with
$
\hat{\mathbf{X}}^k_{n+1} = \exp \left( \mathbf{G} \left ( T_n \right)   \left( T_{n+1} - T_n  \right)  \right) \bar{\mathbf{Z}}^k_{n+1} ,
$
where
\begin{align*}
\bar{\mathbf{Z}}^k_{n+1}
& =
\bar{\mathbf{X}}^k_n 
+ \left(  T_{n+1} - T_n  \right) \bar{\mathbf{g}}^k_n
+ \sum_{j=1}^J \left( \mathbf{L}_{j}  \left( T_n \right) \bar{\mathbf{X}}^k_n  - \bar{\ell}_n^j  \, \bar{\mathbf{X}}^k_n \right) 
\sqrt{T_{n+1} - T_n} \,  \overline{\delta W}^j_{n+1} 
\\
& \quad 
+  \sum_{m=1}^M    \left(  { \mathbf{R}_{m}  \left( T_n \right) \bar{\mathbf{X}}^k_n }/{ \sqrt{\bar{r}_n^m }} -  \bar{\mathbf{X}}^k_n \right)
\overline{\delta N}^m_{n+1} ,
\end{align*}
$
\bar{\mathbf{g}}^k_n 
= 
\sum_{j=1}^J \left( \bar{\ell}_n^j \mathbf{L}_j  \left( T_n \right) \bar{\mathbf{X}}^k_n  - \left(\bar{\ell}_n^j\right)  ^2  \bar{\mathbf{X}}^k_n \,/2 \right)
 + \sum_{m=1}^M  \bar{r}_n^m  \bar{\mathbf{X}}^k_n \,/2
$
and
$
\bar{\ell}_j  = \sum_{k=1}^{\mu}  \Re \langle  \bar{\mathbf{X}}^k_n, \mathbf{L}_{j}  \left( T_n \right) \bar{\mathbf{X}}^k_n \rangle 
$.
\end{itemize}
\end{scheme}

\begin{remark}
\label{rem:Generalization}
Let $\boldsymbol{\rho}_{0}$ be a random pure state.
Then $\mu = 1$, 
and so $\bar{\mathbf{X}}^1_{n}$ described recursively by Scheme \ref{scheme:SME_Euler_Exponential}
is a new numerical  integrator for \eqref{eq:SSE},
which generalizes Scheme 2 of \cite{Mora2005} constructed for solving efficiently 
\eqref{eq:SSE} in the autonomous diffusive case,
that is, 
when $M=0$ and the coefficient operators 
$\mathbf{G} \left( t \right)$ and $ \mathbf{L}_{j} \left( t \right)$ are constant.
\end{remark}

\begin{remark}
An analysis of the derivation of Scheme \ref{scheme:SME_Euler_Exponential} 
suggests us that  the rate of  weak convergence of Scheme \ref{scheme:SME_Euler_Exponential} 
is  equal to $1$,
that is, 
for any regular function $\mathbf{f}: \boldsymbol{\mathfrak{h}} \rightarrow \mathbb{C}$ 
we expect that 
$
\mathbb{E} \mathbf{f} \left( \bar{\boldsymbol{\rho}}_{N}   \right) - \mathbb{E} \mathbf{f}  \left( \boldsymbol{\rho}_{T} \right)
=
o \left( \Delta \right)
$
whenever $\Delta = T / N$ and $ T_n = n  \Delta $.
The rigorous proof of this convergence property is in progress;
by combining techniques from \cite{Mora2005} and \cite{Mora2017}
we have actually proved that 
$
\mathbb{E} \tr \left( \mathbf{A} \, \bar{\boldsymbol{\rho}}_{N}   \right) 
- 
\mathbb{E} \tr \left(  \mathbf{A} \, \boldsymbol{\rho}_{T} \right)
=
o \left( \Delta \right)
$
for any linear operator  $\mathbf{A} : \boldsymbol{\mathfrak{h}} \rightarrow  \boldsymbol{\mathfrak{h}}$,
which is consistent with Table \ref{Tab:Ej2OpA1} given in Section \ref{subsec:DirectDetection}.
\end{remark}

\begin{remark}
In general,
it cannot be guaranteed  that
$\tr \left( \boldsymbol{\rho}_{t} \right) = \tr \left( \boldsymbol{\rho}_{0} \right)$
unless $\tr \left( \boldsymbol{\rho}_{0} \right) = 1$,
and so the quantity $
\sum_{k=1}^{\mu}  \left\Vert  \mathbf{X}^k_{t} \right\Vert^2
$
may not be preserved. 
If $M = 0$ and 
$ -  \mathrm{i} \, \mathbf{L}_{1} \left(  t \right), \ldots, -  \mathrm{i} \, \mathbf{L}_{J} \left(  t \right)$ are symmetric operators,
then 
$
\sum_{k=1}^{\mu}  \left\Vert  \mathbf{X}^k_{t} \right\Vert^2 = \sum_{k=1}^{\mu}  \left\Vert  \mathbf{X}^k_{0} \right\Vert^2
$
for all $ t \geq 0$.
In this interesting particular case,
which is not the focus of the paper,
\eqref{eq:StochasticQME} comes down to the linear stochastic master equation
$
 d\mathbf{\boldsymbol{\rho}}_{t}	
 =	
\boldsymbol{\mathcal{L}} \left( t \right)  \boldsymbol{\rho}_{t} dt 
+ \sum_{j =1}^J \left( \mathbf{L}_{j} \left( t \right) \boldsymbol{\rho}_{t} + \boldsymbol{\rho}_{t} \mathbf{L}_{j} \left( t \right)^* \right) dW^{j}_{t} 
$,
and \eqref{eq:SistemSSE} becomes the system of uncoupled linear  stochastic Schr\"odinger equations
\begin{equation}  
\label{eq:LinearSSE}
 \mathbf{X}^k _t   
  =   
\mathbf{X}^k _0
+ \int_{0}^t \mathbf{G} \left ( s \right) \mathbf{X}^k_{s}  \,  ds
 + \sum_{j=1}^J \int_{0}^t \mathbf{L}_{j}  \left( s \right) \mathbf{X}^k_{s} \, dW^j_s ,
 \end{equation}
which satisfy
$\left\Vert  \mathbf{X}^k_{t} \right\Vert = \left\Vert  \mathbf{X}^k_{0} \right\Vert$.
As an alternative to Scheme \ref{scheme:SME_Euler_Exponential},
we can solve \eqref{eq:LinearSSE} by schemes preserving quadratic invariants
(see, e.g.,  \cite{Abdulle2012,Anmarkrud2017,Hong2015} and  \cite{ChenHong2016,Cui2017,jiang_wang_hong_2013}).
Since 
$ -  \mathrm{i} \, \mathbf{L}_{j} \left(  t \right)$ is a symmetric operator,
according to \eqref{eq:LinearSSE} we have 
\begin{equation}  
\label{eq:LinearSSE_S}
 \mathbf{X}^k _t   
  =   
\mathbf{X}^k _0
- \int_{0}^t  \mathrm{i}  \, \mathbf{H}  \left ( s \right) \mathbf{X}^k_{s}  \,  ds
 + \sum_{j=1}^J \int_{0}^t \mathbf{L}_{j}  \left( s \right) \mathbf{X}^k_{s} \, \circ dW^j_s ,
 \end{equation}
where $\circ $ denotes the Stratonovich integral.
Applying the midpoint rule to  \ref{eq:LinearSSE_S} yields the recursive algorithm 
$$
\bar{\mathbf{V}}^k_{n+1}
=
\bar{\mathbf{V}}^k_{n}
+
\frac{1}{2} \bar{\boldsymbol{\Phi}}_{n} \left( \bar{\mathbf{V}}^k_{n} + \bar{\mathbf{V}}^k_{n+1} \right) ,
$$
with
$
\bar{\boldsymbol{\Phi}}_{n}
=
-  \mathrm{i}  \, \mathbf{H} \left (  \frac{T_n + T_{n+1} }{2} \right) \left( T_{n+1} - T_n  \right)
+
\sqrt{T_{n+1} - T_n} \sum_{j=1}^J \mathbf{L}_{j}  \left(  \frac{T_n + T_{n+1} }{2} \right)   \overline{\delta W}^j_{n+1}
$.
Here,
$\overline{\delta W}^1_1, \overline{\delta W}^2_1,  \ldots,$ $\overline{\delta W}^J_1, \overline{\delta W}^1_2, \ldots$ 
are i.i.d. symmetric real random variables with variance $1$
that are independent of $\bar{\mathbf{V}}^1_{0}, \ldots, \bar{\mathbf{V}}^{\mu}_{0}$.
From the symmetry of $ -  \mathrm{i} \, \mathbf{L}_{j} \left(  t \right)$ it follows that
$  \Re \langle v, \bar{\boldsymbol{\Phi}}_{n} v \rangle = 0$ for all $v \in \boldsymbol{\mathfrak{h}}$.
Hence,
$$
\left\Vert \bar{\mathbf{V}}^k_{n+1} \right\Vert^2  - \left\Vert  \bar{\mathbf{V}}^k_{n} \right\Vert^2
=
\frac{1}{2} \Re \langle  \bar{\mathbf{V}}^k_{n} + \bar{\mathbf{V}}^k_{n+1} , 
\bar{\boldsymbol{\Phi}}_{n} \left( \bar{\mathbf{V}}^k_{n} + \bar{\mathbf{V}}^k_{n+1} \right) \rangle
=
0
$$
and 
$\left( I - \bar{\boldsymbol{\Phi}}_{n}/2 \right)$ is invertible.
Therefore,
$\bar{\mathbf{V}}^k_{n}$ is well-defined and conserves the norm of $\bar{\mathbf{V}}^k_{0}$.
\end{remark}

\subsection{Variants of the Euler-exponential scheme}

In this subsection we adopt 
the notation and approximations of Section \ref{subsec:EulerExponential},
except the relation \eqref{eq:2.2}.
We will design exponential schemes for  \eqref{eq:SistemSSE}
by computing 
$$
\int_{T_n }^{T_{n+1}}  
\exp{\left( \mathbf{G} \left ( T_n \right)   \left( T_{n+1} - s  \right)\right )  }
\left(  \left( \mathbf{G} \left ( s \right)   - \mathbf{G} \left( T_{n} \right)  \right) \mathbf{X}^k_{s-}  +  \widetilde{\mathbf{g}}  \left( s -,  \mathbf{X}^k_{s-} \right)   \right) ds 
$$
in two different ways.

First,
$
 \left( \mathbf{G} \left ( s \right)   - \mathbf{G} \left( T_{n} \right)  \right) \mathbf{X}^k_{s-}  +  \widetilde{\mathbf{g}}  \left( s -,  \mathbf{X}^k_{s-} \right)
\approx
 \widetilde{\mathbf{g}}   \left( T_n ,  \mathbf{X}^k_{T_n} \right) 
$
for all $s \in \left] T_n , T_{n+1} \right]$.
Hence
\begin{equation}
\label{eq:2.11} 
 \begin{aligned}
&  \int_{T_n }^{T_{n+1}} 
\exp{\left( \mathbf{G} \left ( T_n \right)   \left( T_{n+1} - s  \right) \right ) }
\left(  \left( \mathbf{G} \left ( s \right)   - \mathbf{G} \left( T_{n} \right)  \right) \mathbf{X}^k_{s-}  +  \widetilde{\mathbf{g}}  \left( s -,  \mathbf{X}^k_{s-} \right)   \right) ds
\\
&
\quad  \approx
\int_{T_n }^{T_{n+1}}  \exp{\left( \mathbf{G} \left ( T_n \right)   \left( T_{n+1} - s  \right)\right ) }  
\, \widetilde{\mathbf{g}}  \left( T_n ,  \mathbf{X}^k_{T_n} \right)  ds
\approx_w
\int_{T_n }^{T_{n+1}}  \exp{\left( \mathbf{G} \left ( T_n \right)   \left( T_{n+1} - s  \right) \right )}  
\, \widetilde{\mathbf{g}}  \left( T_n ,  \bar{\mathbf{X}}^k_n \right)  ds .
\end{aligned}
\end{equation}
This implies
\begin{align*}
\mathbf{X}^k_{T_{n+1}}
& \approx_w
\exp{\left( \mathbf{G} \left ( T_n \right)   \left( T_{n+1} - T_n  \right)  \right )} \bar{\mathbf{X}}^k_n
+
\int_{T_n }^{ T_{n+1} }  
\exp{\left( \mathbf{G} \left ( T_n \right)   \left( T_{n+1} - s  \right) \right ) } \, \bar{\mathbf{g}}^k_n  \, ds
 \\
 & \quad 
 + \sum_{j=1}^J
 \exp{\left( \mathbf{G} \left ( T_n \right)   \left( T_{n+1} - T_n  \right)  \right )} \left( \mathbf{L}_{j}  \left( T_n \right) \bar{\mathbf{X}}^k_n  - \bar{\ell}_n^j  \, \bar{\mathbf{X}}^k_n \right) 
\sqrt{T_{n+1} - T_n} \,  \overline{\delta W}^j_{n+1} 
\\
& \quad   
+  \sum_{m=1}^M  
\exp{\left( \mathbf{G} \left ( T_n \right)   \left( T_{n+1} - T_n  \right)\right )  }  \left(  \dfrac{ \mathbf{R}_{m}  \left( T_n \right) \bar{\mathbf{X}}^k_n }{ \sqrt{\bar{r}_n^m }} -  \bar{\mathbf{X}}^k_n \right)
\overline{\delta N}^m_{n+1} .
\end{align*}
Using 
\begin{align*}
 & \exp \left(  
\begin{pmatrix}
  \mathbf{G} \left ( T_n \right) &  \bar{\mathbf{g}}^k_n  \\
  \mathbf{0}_{1 \times d}  & 0
\end{pmatrix} 
 \left( T_{n+1} - T_n  \right)
\right)
=
\begin{pmatrix}
\exp{\left( \mathbf{G} \left ( T_n \right)   \left( T_{n+1} - T_n  \right) \right ) }  &
\int_{T_n }^{ T_{n+1} }  \exp{\left( \mathbf{G} \left ( T_n \right)   \left( T_{n+1} - s  \right)  \right )} \, \bar{\mathbf{g}}^k_n \, ds
\\
0_{1 \times d} & 1  
\end{pmatrix}
\end{align*}
(see, e.g., \cite{VanLoan1978}),
we get
$\exp{\left( \mathbf{G} \left ( T_n \right)   \left( T_{n+1} - T_n  \right) \right ) } \bar{\mathbf{X}}^k_n$
 and 
$\int_{T_n }^{ T_{n+1} }  \exp{\left( \mathbf{G} \left ( T_n \right)   \left( T_{n+1} - s  \right)  \right )} \, \bar{\mathbf{g}}^k_n \,   ds$
by evaluating just one matrix exponential.
As in Scheme  \ref{scheme:SME_Euler_Exponential},
normalizing by the trace of the approximation of $ \boldsymbol{\rho}_{T_{n+1}}$ 
we deduce the following numerical integrator.

\begin{scheme}
\label{scheme:SME_Exponential_Integral}
Define $\bar{\boldsymbol{\rho}}_{n}, \bar{\mathbf{X}}^1_{n}, \ldots, \bar{\mathbf{X}}^{\mu}_{n}$, together with $\overline{\delta W}^j_{n}$ and $\overline{\delta N}^m_{n}$,
as in Scheme  \ref{scheme:SME_Euler_Exponential}
with $\hat{\mathbf{X}}^k_{n+1}$ replaced by the first $d$ components of
$$
\exp \left(  
\begin{pmatrix}
  \mathbf{G} \left ( T_n \right) & \bar{\mathbf{g}}^k_n   \\
  \mathbf{0}_{1 \times d}  & 0
\end{pmatrix} 
 \left( T_{n+1} - T_n  \right)
\right)
\begin{pmatrix}
\bar{\mathbf{Y}}^k_{n+1} 
\\
1
\end{pmatrix}
$$
with
$$
\bar{\mathbf{Y}}^k_{n+1}
 =
\bar{\mathbf{X}}^k_n
 + \sum_{j=1}^J \left( \mathbf{L}_{j}  \left( T_n \right) \bar{\mathbf{X}}^k_n  - \bar{\ell}_n^j  \, \bar{\mathbf{X}}^k_n \right) 
\sqrt{T_{n+1} - T_n} \,  \overline{\delta W}^j_{n+1}  
+  \sum_{m=1}^M    \left(  \dfrac{ \mathbf{R}_{m}  \left( T_n \right) \bar{\mathbf{X}}^k_n }{ \sqrt{\bar{r}_n^m }} -  \bar{\mathbf{X}}^k_n \right)
\overline{\delta N}^m_{n+1} .
$$
\end{scheme}

Second,
we suppose that $\mathbf{G}$ is continuously differentiable.
Then
$
\mathbf{G} \left ( s \right)   - \mathbf{G} \left( T_{n} \right)  \approx \mathbf{G}^{\prime} \left( T_n \right) \left( s - T_n  \right) 
$
for all
$
 s \in \left[ T_n , T_{n+1} \right] .
$
Therefore
\begin{equation}
\label{eq:2.12}
 \begin{aligned}
&
\int_{T_n }^{T_{n+1}}  
\exp{\left( \mathbf{G} \left ( T_n \right)   \left( T_{n+1} - s  \right)  \right )}
\left(  \left( \mathbf{G} \left ( s \right)   - \mathbf{G} \left( T_{n} \right)  \right) \mathbf{X}^k_{s-}  +  \widetilde{\mathbf{g}}  \left( s -,  \mathbf{X}^k_{s-} \right)   \right) ds
\\
& 
\quad  \approx
\int_{T_n }^{T_{n+1}}  \exp{\left( \mathbf{G} \left ( T_n \right)   \left( T_{n+1} - s  \right) \right )} 
\left(  \widetilde{\mathbf{g}}  \left( T_n ,  \mathbf{X}^k_{T_n} \right) 
+ \left( s - T_n  \right) \mathbf{G}^{\prime} \left( T_n \right)  \mathbf{X}^k_{T_n}   \right) ds
\\
& 
\quad  \approx_w
\mathbf{I}_{n+1}:= 
\int_{T_n }^{ T_{n+1} }  
\exp{\left( \mathbf{G} \left ( T_n \right)   \left( T_{n+1} - s  \right) \right ) }
\left(  \bar{\mathbf{g}}^k_n +  \left( s - T_n  \right)  \mathbf{G}^{\prime} \left( T_n \right) \bar{\mathbf{X}}^k_n \right) ds ,
\end{aligned}
\end{equation}
which becomes an alternative to the approximations \eqref{eq:2.2} and \eqref{eq:2.11}.
It is worth pointing out that \eqref{eq:2.12} becomes \eqref{eq:2.11} whenever $ \mathbf{G}^{\prime} \left( T_n \right) = 0$.
Since 
\begin{align*}
 & \exp \left(  
\begin{pmatrix}
  \mathbf{G} \left ( T_n \right) & \mathbf{G}^{\prime} \left( T_n \right)  \bar{\mathbf{X}}^k_n & \bar{\mathbf{g}}^k_n   
  \\
  \mathbf{0}_{1 \times d}  & 0 & 1
  \\
  \mathbf{0}_{1 \times d}  & 0 & 0
\end{pmatrix} 
 \left( T_{n+1} - T_n  \right)
\right)
\\
& \quad  =
\begin{pmatrix}
\exp{\left( \mathbf{G} \left ( T_n \right)   \left( T_{n+1} - T_n  \right)\right )  }  &
\int_{T_n }^{ T_{n+1} }  \exp{\left( \mathbf{G} \left ( T_n \right)   \left( T_{n+1} - s  \right) \right ) } \mathbf{G}^{\prime} \left( T_n \right)  \bar{\mathbf{X}}^k_n  ds
&
\mathbf{I}_{n+1}
\\
\mathbf{0}_{1 \times d} & 1  &  T_{n+1} - T_n   
\\
\mathbf{0}_{1 \times d} &  0 &  1
\end{pmatrix} 
\end{align*}
(see, e.g., \cite{VanLoan1978}),
using \eqref{eq:2.5}, \eqref{eq:2.1} and \eqref{eq:2.9} we construct the next numerical method for \eqref{eq:StochasticQME}.

\begin{scheme}
 \label{scheme:SME_ExponentialExponential_Integral_GD}
Let $\bar{\boldsymbol{\rho}}_{n}, \bar{\mathbf{X}}^1_{n}, \ldots, \bar{\mathbf{X}}^{\mu}_{n}$, $\overline{\delta W}^j_{n}$ and $\overline{\delta N}^m_{n}$
be as in Scheme  \ref{scheme:SME_Euler_Exponential} 
but with $\hat{\mathbf{X}}^k_{n+1}$ given by  the first $d$ components of
$$
\exp \left(  
\begin{pmatrix}
  \mathbf{G} \left ( T_n \right) & \mathbf{G}^{\prime} \left( T_n \right)  \bar{\mathbf{X}}^k_n &  \bar{\mathbf{g}}^k_n   
  \\
  \mathbf{0}_{1 \times d}  & 0 & 1
  \\
  \mathbf{0}_{1 \times d}  & 0 & 0
\end{pmatrix} 
 \left( T_{n+1} - T_n  \right)
\right)
\begin{pmatrix}
\bar{\mathbf{Y}}^k_{n+1} 
\\
0
\\
1
\end{pmatrix} 
$$ 
with
$$
\bar{\mathbf{Y}}^k_{n+1}
 =
\bar{\mathbf{X}}^k_n
 + \sum_{j=1}^J \left( \mathbf{L}_{j}  \left( T_n \right) \bar{\mathbf{X}}^k_n  - \bar{\ell}_n^j  \, \bar{\mathbf{X}}^k_n \right) 
\sqrt{T_{n+1} - T_n} \,  \overline{\delta W}^j_{n+1}  
+  \sum_{m=1}^M    \left(  \dfrac{ \mathbf{R}_{m}  \left( T_n \right) \bar{\mathbf{X}}^k_n }{ \sqrt{\bar{r}_n^m }} -  \bar{\mathbf{X}}^k_n \right)
\overline{\delta N}^m_{n+1} .
$$
\end{scheme}

\begin{remark}
If  $t \mapsto \mathbf{G} \left( t \right)$ is a constant function,
then 
Schemes \ref{scheme:SME_Exponential_Integral} and \ref{scheme:SME_ExponentialExponential_Integral_GD}
provide the same approximation of $\boldsymbol{\rho}_{T_n} $.
Nevertheless, 
Scheme \ref{scheme:SME_ExponentialExponential_Integral_GD}  has higher computational cost.
\end{remark}

\begin{remark}
Similar to Remark \ref{rem:Generalization},
$\bar{\mathbf{X}}^1_{n}$ defined recursively by Schemes  \ref{scheme:SME_Exponential_Integral}
and \ref{scheme:SME_ExponentialExponential_Integral_GD}
are new numerical methods for \eqref{eq:SSE}
in case $\boldsymbol{\rho}_{0}$ is a random pure state.
\end{remark}

On the other hand,
$
\dfrac{1}{2}\, \sum_{j = 1}^{J} \mathbf{L}_{j}  \left( t \right)^*\mathbf{L}_{j}  \left( t \right) 
+ \dfrac{1}{2}\, \sum_{m = 1}^{M} \mathbf{R}_{m}  \left( t \right)^*\mathbf{R}_{m}  \left( t \right)
$
is a  positive-definite matrix,
and so we deduce from \eqref{eq:1.1} that 
$\mathbf{I} - \mathbf{G} \left( T_n \right)   \left( T_{n+1} - T_n  \right)$ is a nonsingular matrix,
where $\mathbf{I}$ stands for the identity matrix.
Using
$
\exp \bigl( \mathbf{G} \left( T_n \right)   \left( T_{n+1} - T_n  \right)  \bigr) 
\approx
\bigl( \mathbf{I} - \mathbf{G} \left( T_n \right)   \left( T_{n+1} - T_n  \right) \bigr)^{-1} 
$
we transform Scheme  \ref{scheme:SME_Euler_Exponential} into 
the following semi-implicit Euler scheme, 
which avoids the solution of nonlinear equations.

\begin{scheme}
 \label{scheme:Euler_Implicit}
Adopt the setup of Scheme \ref{scheme:SME_Euler_Exponential}
with $\hat{\mathbf{X}}^k_{n+1}$ substituted by 
$$
\hat{\mathbf{X}}^k_{n+1} 
=
\bigl(  \mathbf{I} - \mathbf{G} \left ( T_n \right)   \left( T_{n+1} - T_n  \right)  \bigr)^{-1}  \bar{\mathbf{Z}}^k_{n+1} .
$$
\end{scheme}


\subsection{Implementation issues}

Schemes \ref{scheme:SME_Euler_Exponential}, \ref{scheme:SME_Exponential_Integral} and \ref{scheme:SME_ExponentialExponential_Integral_GD}
involve the computation of matrix exponentials times vectors,
which can be accomplished in many ways (see, e.g.,  \cite{MolerVanLoan2003}).
If $\mathbf{A}$ is a large and sparse matrix and $\mathbf{v}$ is a vector,
then we can get $\exp \left( \mathbf{A} \right) \mathbf{v}$ by Krylov subspace iterative methods 
(see, e.g., \cite{GallopoulosSaad1992,HochbruckLubich1997,Wang2016}).
In case the dimension of $\mathbf{v}$ is less than  a few thousand,
we use the standard scaling and squaring method with Pad\'e approximants
for computing the exponential of $\mathbf{A}$   on a current computer
(see, e.g.,  \cite{DeLaCruz2013,Higham2005,Higham2009,MolerVanLoan2003}).
In the latter method,
$m \in \mathbb{Z}_+$ is chosen such that $\left\Vert \mathbf{A} / 2^m \right\Vert$ is of order of magnitude $1$
and  $\exp \left( \mathbf{A}/2^m \right)$ is approximated by the rational function 
$\mathbf{D}_{p q} \left( \mathbf{A} /{2^m} \right)^{-1} \mathbf{N}_{p q}\left( \mathbf{A}/{2^m} \right)$,
where $p,q \in  \mathbb{N}$ and for any matrix $\mathbf{X}$ we define
\begin{equation}
\label{eq:2.7}
\mathbf{N}_{p q} \left( \mathbf{X} \right) = \sum_{j=0}^p \frac{\left(p+q-j \right)! \, p!}{\left(p+q \right)! \left(p-j \right)!}\frac{\mathbf{X}^j}{j!}
\text{ and }
\mathbf{D}_{p q}(\mathbf{X}) =  \sum_{j=0}^q \frac{\left(p+q-j \right)! \, q!}{\left(p+q \right)! \left(q-j \right)!}\frac{\left(-1 \right)^j \mathbf{X}^j}{j!} .
\end{equation}
Thus,
$
\exp \left( \mathbf{A} \right) 
=
\exp \left( \mathbf{A}/2^m \right)^{2^m}
\approx
\left(\mathbf{D}_{p q} \left( \mathbf{A}/{2^m} \right)^{-1} \mathbf{N}_{p q}\left( \mathbf{A}/{2^m} \right) \right)^{2^m} 
$.

We next obtain simple algebraic expressions for 
the Pad\'e approximants to the matrix exponentials appearing in
Schemes \ref{scheme:SME_Exponential_Integral}
and \ref{scheme:SME_ExponentialExponential_Integral_GD}.
Hence, 
Schemes \ref{scheme:SME_Exponential_Integral}
and \ref{scheme:SME_ExponentialExponential_Integral_GD}
increasing slightly the computational cost of Scheme \ref{scheme:SME_Euler_Exponential}.
At each recursive step,
Scheme \ref{scheme:SME_Exponential_Integral} computes 
\begin{equation*}
\exp \left(  
\begin{pmatrix}
  \mathbf{G} \left ( T_n \right) & \mathbf{g} \left ( T_n ,  \bar{\mathbf{X}}_n \right)  \\
  \mathbf{0}_{1 \times d}  & 0
\end{pmatrix} 
 \left( T_{n+1} - T_n  \right)
\right)
\begin{pmatrix}
\bar{\mathbf{Z}}_{n+1} 
\\
1
\end{pmatrix} ,
\end{equation*}
where for any $n \in \mathbb{Z}_+$,  $\mathbf{G} \left( T_n \right)$ is a fixed matrix and 
$\mathbf{g} \left ( T_n ,  \bar{\mathbf{X}}_n \right)$, $\bar{\mathbf{Z}}_{n+1}$ depend on the statistical sample.
According to Theorem \ref{th:Pade1} given below,
we can reduce $M$ evaluations of 
$$
\exp \left(  
\begin{pmatrix}
  \mathbf{G} \left ( T_n \right) & \mathbf{g} \left ( T_n ,  \bar{\mathbf{X}}_n \right)  \\
  \mathbf{0}_{1 \times d}  & 0
\end{pmatrix} 
 \left( T_{n+1} - T_n  \right)
\right)
$$
to the operations required essentially to evaluate $\exp \left( \left( T_{n+1} - T_n  \right)  \mathbf{G} \left ( T_n \right) \right)$,
together with the multiplication of a deterministic $d\times d$-matrix by 
$M$  sample vectors $\mathbf{g} \left ( T_n ,  \bar{\mathbf{X}}_n \right)$.
Thus,
Schemes \ref{scheme:SME_Euler_Exponential} and \ref{scheme:SME_Exponential_Integral} 
involve similar number of floating point multiplications.

\begin{theorem}
 \label{th:Pade1}
Suppose that 
$
\mathbf{A} =
\begin{pmatrix}
  \mathbf{G}  & \mathbf{g}  \\
  \mathbf{0}_{1 \times d}  & 0
\end{pmatrix} 
$
with $\mathbf{G} \in \mathbb{C}^{d \times d}$ and $\mathbf{g} \in \mathbb{C}^{d \times 1}$.  
Fix $p,q \in  \mathbb{N}$ and  $m \in  \mathbb{Z}_+$.
For any matrix $\mathbf{X}$ we set 
$
\mathbf{P}_{p q} \left( \mathbf{X} \right) 
=
\mathbf{D}_{p q} \left( \mathbf{X} \right)^{-1} \mathbf{N}_{p q}\left( \mathbf{X} \right) 
$,
where $\mathbf{D}_{p q}$ and $\mathbf{N}_{p q}$ are as in \eqref{eq:2.7}.
Then
$
\mathbf{P}_{p q} \left( \mathbf{A} / \left(2^m \right) \right)^{2^m}
$
is equal to
$$
\begin{pmatrix}
 \mathbf{P}_{p q} \left( \frac{\mathbf{G}}{2^m} \right) ^{2^m}  
  &  
 \prod^{m-1}_{j=0}\left( \mathbf{I} + \mathbf{P}_{p q} \left( \frac{\mathbf{G}}{2^m} \right)^{2^j} \right)
 \mathbf{D}_{p q} \left( \frac{\mathbf{G}}{2^m} \right)^{-1}
 \left(\widetilde{\mathbf{N}}_{p q}\left( \frac{\mathbf{G}}{2^m} \right) - \widetilde{\mathbf{D}}_{p q} \left( \frac{\mathbf{G}}{2^m} \right)\right)  \frac{\mathbf{g}}{2^m}
 \\
  \mathbf{0}_{1 \times d}  & 1
\end{pmatrix} ,
$$
where for any matrix $\mathbf{X}$,
$
\widetilde{\mathbf{N}}_{ p q} \left( \mathbf{X} \right) = \sum_{j=1}^p \frac{\left(p+q-j \right)! \, p!}{\left(p+q \right)! \left(p-j \right)!}\frac{\mathbf{X}^{j-1}}{j!}
$
and 
$
\widetilde{\mathbf{D}}_{p q }(\mathbf{X}) =  \sum_{j=1}^q \frac{\left(p+q-j \right)! \, q!}{\left(p+q \right)! \left(q-j \right)!}\frac{\left(-1 \right)^j \mathbf{X}^{j-1}}{j!} .
$
\end{theorem}

\begin{proof}
 Deferred to Subsection \ref{sec:Proof:Pade1}.
\end{proof}

Every integration step of Scheme \ref{scheme:SME_ExponentialExponential_Integral_GD} involves 
the computation of a matrix exponential times a random vector,
where the matrix exponential function is evaluated at 
$\begin{pmatrix}
   \mathbf{G} \left ( T_n \right) &  \mathbf{G}^{\prime} \left( T_n \right)  \bar{ \mathbf{X}}^k_n &  \bar{ \mathbf{g}}^k_n   
  \\
   \mathbf{0}_{1 \times d}  & 0 & 1
  \\
   \mathbf{0}_{1 \times d}  & 0 & 0
\end{pmatrix} 
 \left( T_{n+1} - T_n  \right)
$
with $ \mathbf{G} \left( T_n \right)$,  $ \mathbf{G}^{\prime} \left( T_n \right)$ deterministic matrices and 
 $\bar{ \mathbf{X}}^k_n$, $\bar{ \mathbf{g}}^k_n$ random vectors.
Similarly to Scheme \ref{scheme:SME_Exponential_Integral}, 
Theorem \ref{th:Pade2} allows us to implement efficiently Scheme \ref{scheme:SME_ExponentialExponential_Integral_GD},
at a computational cost not substantially greater than that of Scheme \ref{scheme:SME_Euler_Exponential}.

\begin{theorem}
\label{th:Pade2}
Set
$
\mathbf{A} =
\begin{pmatrix}
  \mathbf{G}  & \mathbf{a} & \mathbf{g} \\
  \mathbf{0}_{1 \times d}  & 0 & b\\
  \mathbf{0}_{1 \times d}  & 0 & 0 
\end{pmatrix} 
$
with $\mathbf{G} \in \mathbb{C}^{d \times d}$, $\mathbf{a},\mathbf{g} \in \mathbb{C}^{d \times 1}$ and $b \in \mathbb{C}$.
Let $p,q \in  \mathbb{N}$ and  $m \in  \mathbb{Z}_+$.
Suppose that $\mathbf{D}_{p q}$ and $\mathbf{N}_{p q}$ are given by \eqref{eq:2.7}
and that $\mathbf{P}_{p q} $, $\widetilde{\mathbf{N}}_{ p q}$, $\widetilde{\mathbf{D}}_{p q }$ are as in Theorem \ref{th:Pade1}.
Then
$$
\mathbf{P}_{p q} \left(  \frac{\mathbf{A}}{2^m} \right)^{2^m}
=
\begin{pmatrix}
 \mathbf{P}_{p q} \left( \frac{\mathbf{G}}{2^m} \right) ^{2^m}  
  &  
 \mathbf{R}_{m} \widetilde{\mathbf{r}}_{m} \frac{\mathbf{a}}{2^m}
 &
 \mathbf{R}_{m} \left( \widetilde{\mathbf{r}}_m \frac{\mathbf{g}}{2^m} + \widehat{\mathbf{r}}_m  \frac{b \, \mathbf{a}}{2^{2m}} \right)
 + 
 \widehat{\mathbf{R}}_m \widetilde{\mathbf{r}}_m  \frac{b \, \mathbf{a}}{2^{2m}}
 \\
  \mathbf{0}_{1 \times d}  & 1 & b
   \\
  \mathbf{0}_{1 \times d}  & 0 & 1
\end{pmatrix} ,
$$
where 
$
\widetilde{\mathbf{r}}_m
=
\mathbf{D}_{p q} \left( \mathbf{G}/{2^m} \right)^{-1}
 \left(\widetilde{\mathbf{N}}_{p q}\left( \mathbf{G}/{2^m} \right) - \widetilde{\mathbf{D}}_{p q} \left( \mathbf{G}/{2^m} \right)\right) 
$,
the matrices $\mathbf{R}_{m}$, $\widehat{\mathbf{R}}_m$ are defined recursively by
\begin{equation}
 \label{eq:2.21}
 \mathbf{R}_{k+1}
 = 
\begin{cases}
\left( \mathbf{I} + \mathbf{P}_{p q} \left( \mathbf{G}/{2^m} \right) ^{2^k} \right) \mathbf{R}_k
&
\text{if } 
k \in \mathbb{N} ,
\\
\mathbf{I}
&
\text{if } k=0 ,
\end{cases} 
\text{ and }
\widehat{\mathbf{R}}_{k+1}
 = 
\begin{cases}
\left( \mathbf{I} + \mathbf{P}_{p q} \left( \mathbf{G}/{2^m} \right) ^{2^k} \right) \widehat{\mathbf{R}}_{k}
+2^k \mathbf{R}_k
&
\text{if } 
k \in \mathbb{N} ,
\\
\mathbf{0}_{d\times d}
&
\text{if } k=0 ,
\end{cases} 
\end{equation}
and 
$
\widehat{\mathbf{r}}_m
=
\mathbf{D}_{p q} \left( \mathbf{G}/{2^m} \right)^{-1}
 \left(\widehat{\mathbf{N}}_{p q}\left( \mathbf{G}/{2^m} \right) 
 - \widehat{\mathbf{D}}_{p q} \left( \mathbf{G}/{2^m} \right)
 - \widetilde{\mathbf{D}}_{p q} \left( \mathbf{G}/{2^m} \right)\right)
$
with 
$$
\widehat{\mathbf{N}}_{ p q} \left( X \right) 
=
\begin{cases}
\mathbf{0}_{d\times d}
&
\text{if } 
p = 1,
\\
\sum_{j=2}^p \frac{\left(p+q-j \right)! \, p!}{\left(p+q \right)! \left(p-j \right)!}\frac{\mathbf{X}^{j-2}}{j!}
&
\text{if } p \geq 2 
\end{cases} 
\text{ and }
\widehat{\mathbf{D}}_{p q }(X) 
= 
\begin{cases}
\mathbf{0}_{d\times d}
&
\text{if } 
q = 1,
\\
\sum_{j=2}^q \frac{\left(p+q-j \right)! \, q!}{\left(p+q \right)! \left(q-j \right)!}\frac{\left(-1 \right)^j \mathbf{X}^{j-2}}{j!} 
&
\text{if } p \geq 2 .
\end{cases} 
$$
\end{theorem}

\begin{proof}
 Deferred to Subsection \ref{sec:Proof:Pade2}.
\end{proof}

In the next remark we show how to change  the time scale of \eqref{eq:StochasticQME}.

\begin{remark}
\label{rem:normalizacion}
We can consider  \eqref{eq:StochasticQME} with normalized physical constants.
For a given $c > 0$,
we define $\boldsymbol{\varrho}_t := \boldsymbol{\rho}_{c \, t}$, where $t \geq 0$.
Then $\boldsymbol{\varrho}_t$ satisfies \eqref{eq:StochasticQME} with $\mathbf{L}_j \left( t \right)$ and $\mathbf{R}_{m} \left( t \right)$ 
replaced by $ c \, \mathbf{L}_j \left( c \, t \right)$ and $\sqrt{c} \, \mathbf{R}_{m} \left( c \, t \right)$,
respectively.
Indeed, 
from \eqref{eq:StochasticQME} we obtain that 
\begin{equation}
\label{eq:TimeChanged}
\begin{aligned}
d \boldsymbol{\varrho}_t
&  =	
 \boldsymbol{\mathcal{L}} \left( c \, t \right)  \boldsymbol{\varrho}_{t-} c \, dt 
\\ 
& \quad
+ \sum_{j =1}^J \left( \mathbf{L}_{j} \left( c \, t \right) \boldsymbol{\varrho}_{t-} + \boldsymbol{\varrho}_{t-} \mathbf{L}_{j} \left( c \,t \right)^* 
 - 2 \Re  \left(  \tr \left( \mathbf{L}_j \left( c \, t \right) \boldsymbol{\varrho}_{t-} \right) \boldsymbol{\varrho}_{t-} \right) \right) \sqrt{c} \, d\widetilde{W}^{j}_{t}
\\ 
& \quad
+ 
\sum_{m=1}^M   \left( 
\frac{ \mathbf{R}_{m} \left( c \, t \right) \boldsymbol{\varrho}_{t-}  \mathbf{R}_{m} \left(  c \, t \right) ^*}
{ \tr \left( \mathbf{R}_{m} \left( c \, t \right) ^*\mathbf{R}_{m} \left( c \, t \right) \boldsymbol{\varrho}_{t-} \right) } 
- \boldsymbol{\varrho}_{t-} \right) \biggl( d\widetilde{N}^{m}_{t} -  \tr \left( \mathbf{R}_{m} \left( c \, t \right) ^*\mathbf{R}_{m} \left( c \, t \right) \boldsymbol{\varrho}_{t-} \right) c \, dt \biggr) ,
\end{aligned}
\end{equation}
where 
$\boldsymbol{\varrho}_t$, $ \widetilde{W}^{j}_{t} := W^{j}_{c \, t} / \sqrt{c}$ and $\widetilde{N}^{m}_{t} := N^{m}_{c \, t}$ 
are adapted stochastic process with respect to the new filtration $\left(\mathfrak{F}_{c \, t}\right) _{ t\geq 0}$
(see, e.g., Chapter X of \cite{Jacod1979} or \cite{Kobayashi2011}).
Since 
$\widetilde{W}^1, \ldots, \widetilde{W}^J$ are  independent  $\left(\mathfrak{F}_{c \, t}\right) _{ t\geq 0}$-Brownian motions
and 
the $\widetilde{N}^m$'s are $\left(\mathfrak{F}_{c \, t}\right) _{ t\geq 0}$-doubly stochastic Poisson processes with 
intensity
$
 \, \tr \left( \mathbf{R}_{m} \left( c \, t \right) ^*\mathbf{R}_{m} \left( c \, t \right) \boldsymbol{\varrho}_{t-} \right)
$,
\eqref{eq:TimeChanged} becomes \eqref{eq:StochasticQME} with 
$ c \, \mathbf{L}_j \left( c \, t \right)$ and $\sqrt{c} \, \mathbf{R}_{m} \left( c \, t \right)$
in place of 
$\mathbf{L}_j \left( t \right)$ and $\mathbf{R}_{m} \left( t \right)$,
respectively.
\end{remark}

\section{Simulation results}
\label{sec:Simulation results}

This section illustrates the performance of Schemes \ref{scheme:SME_Euler_Exponential} through \ref{scheme:Euler_Implicit}.
For this purpose,
as a ``small quantum system"
we select 
a single-mode quantized electromagnetic field interacting with a two-level system.
The internal dynamics of this coupled system is governed by the Rabi Hamiltonian 
\begin{equation}
\label{eq:RabiHamiltonian}
\mathbf{H}_{Rabi}
=
{\omega_1}\,\boldsymbol{\sigma}^z /2+ \omega_2 \, \mathbf{a}^{\dag} \mathbf{a} + g   \left( \mathbf{a}^{\dag} + \mathbf{a} \right )  \boldsymbol{\sigma}^x ,
\end{equation}
which acts upon the tensor product space $ \boldsymbol{\mathfrak{h}}:= \ell^2 \left(\mathbb{Z}_+ \right)\otimes \mathbb{C}^{2}$
(see, e.g., \cite{Agarwal2012,Braak2011,Niemczyk2010,Xiang2013}).
Here,
$\omega_1 > 0$ is the transition frequency between the ground state 
$
\begin{pmatrix}
 0 
 \\
 1 
\end{pmatrix}
$
of the two-level system and its upper level
$
\begin{pmatrix}
 1
 \\
 0 
\end{pmatrix}
$,
$\omega_2 > 0$ is the angular frequency of the electromagnetic field mode
and
$g \geq 0$ is the dipole coupling constant.
As usual,
the Pauli matrices $\boldsymbol{\sigma}^x$, $\boldsymbol{\sigma}^y$, $\boldsymbol{\sigma}^z$ are  defined by
\begin{equation*}
\boldsymbol{\sigma}^x = 
\begin{pmatrix}
 0 & 1
 \\
 1 & 0
\end{pmatrix},
\quad
\boldsymbol{\sigma}^y = 
\begin{pmatrix}
 0 &  - \mathrm{i} 
 \\
  \mathrm{i}  & 0
\end{pmatrix} ,
\quad
\boldsymbol{\sigma}^z = 
\begin{pmatrix}
 1 &  0
 \\
0  & -1
\end{pmatrix} ,
\end{equation*}
and
the creation and annihilation operators $\mathbf{a}$, $\mathbf{a}^{\dagger}$ are 
closed operators on $\ell^2 \left(\mathbb{Z}_+ \right)$ given by
$$
\mathbf{a}\mathbf{e}_{n} 
=
\begin{cases}
 \sqrt{n} \, \mathbf{e}_{n-1}  & \text{if } n \in  \mathbb{N} ,
 \\
 0 &  \text{if } n = 0 
\end{cases}
\text{ and }
\
\mathbf{a}^{\dagger}\mathbf{e}_{n}
=
\begin{cases} 
 \sqrt{ n+1} \, \mathbf{e}_{n+1}  & \text{if } n \in   \mathbb{Z}_{+} ,
\end{cases}
$$
where
$(\mathbf{e}_n)_{n\ge 0}$ denotes the canonical orthonormal basis of $\ell^2(\mathbb{Z_+})$.

\subsection{Autonomous stochastic quantum master equation of diffusive type}
\label{subsec:DiffusiveSQMEs}

In this subsection,
$M=0$ and 
$\mathbf{H} \left( t \right) $, $\mathbf{L}_1 \left( t \right) , \ldots, \mathbf{L}_J \left( t \right)$
are constant functions.
Therefore, 
we will compute the solution of the autonomous stochastic quantum master equation
\begin{equation}
\label{eq:DiffusiveAutSQME} 
 d\boldsymbol{\rho}_{t}	
  =	
 \boldsymbol{\mathcal{L}}  \, \boldsymbol{\rho}_{t} dt 
+ \sum_{j =1}^J \left( \mathbf{L}_{j} \boldsymbol{\rho}_{t} + \boldsymbol{\rho}_{t} \, \mathbf{L}_{j}^* 
 - 2 \Re  \left(  \tr \left( \mathbf{L}_j \boldsymbol{\rho}_{t} \right) \boldsymbol{\rho}_{t} \right) \right) dW^{j}_{t} ,
\end{equation}
where 
$$
\boldsymbol{\mathcal{L}} \, \boldsymbol{\varrho}
=
- \mathrm{i} \left( \mathbf{H} \, \boldsymbol{\varrho} - \boldsymbol{\varrho}  \, \mathbf{H} \right)
+ \sum_{j = 1}^{J} \left( \mathbf{L}_{j}  \boldsymbol{\varrho} \, \mathbf{L}_{j} ^* -  \mathbf{L}_{j}^* \mathbf{L}_{j} \, \boldsymbol{\varrho} /2-\boldsymbol{\varrho} \, \mathbf{L}_{j}^* \mathbf{L}_{j} /2
\right) 
$$
and 
$\mathbf{H},$ $ \mathbf{L}_1, \ldots, \mathbf{L}_J$ are linear operators on $\boldsymbol{\mathfrak{h}}$ with $\mathbf{H}^* = \mathbf{H}$.
We will compare the performance of  
Schemes \ref{scheme:SME_Euler_Exponential}, \ref{scheme:SME_Exponential_Integral}, \ref{scheme:Euler_Implicit} 
and the following numerical method designed by \cite{Amini2011} for solving \eqref{eq:DiffusiveAutSQME}.

\begin{scheme}
\label{scheme:AminiMirrahimiRouchon}
Let $\bar{\boldsymbol{\rho}}_0$ be a random variable with values in $\boldsymbol{\mathfrak{h}}$.
Suppose that 
$\overline{\delta W}^1_1, \overline{\delta W}^2_1$, $\ldots$, 
$\overline{\delta W}^J_1, \overline{\delta W}^1_2, \ldots$ 
are i.i.d. symmetric real random variables with variance $1$
that are independent of $\bar{\boldsymbol{\rho}}_0$.
Define recursively 
$$
\bar{\boldsymbol{\rho}}_{n+1} = \dfrac{\mathbf{M}_n \, \bar{\boldsymbol{\rho}}_{n} \mathbf{M}_n^*}{\tr\left (\mathbf{M}_n \, \bar{\boldsymbol{\rho}}_{n}  \mathbf{M}_n^*\right )} ,
$$
where $n \in \mathbb{Z}_+$ and
$$
\mathbf{M}_n 
=
\mathbf{I} - \left (i \mathbf{H}+{1}/{2}\,\sum^{J}_{j=1}\mathbf{L}_{j}^*\mathbf{L}_{j} \right)\Delta 
+ \sum_{j=1}^{J} \overline{\delta y}^j_n  \mathbf{L}_{j} 
$$
with
$
\overline{\delta y}^j_n  
= 
\tr \left (\mathbf{L}_{j} \bar{\boldsymbol{\rho}}_n + \bar{\boldsymbol{\rho}}_n \mathbf{L}_{j}^*\right )\Delta +\sqrt{\Delta }\,\,\overline{\delta W}_{n+1}^j
$.
\end{scheme}

We will simulate computationally the next physical system.

\begin{system}
\label{Ex:DiffusiveSQME} 
The internal dynamics of the small quantum system is determined by 
the Hamiltonian $\mathbf{H} = \mathbf{H}_{Rabi}$ acting upon the quantum state space $ \boldsymbol{\mathfrak{h}}= \ell^2 \left(\mathbb{Z}_+ \right)\otimes \mathbb{C}^{2}$,
where $\mathbf{H}_{Rabi}$ is as in \eqref{eq:RabiHamiltonian}.
The interaction of the small system with two independent non-zero temperature baths
is described by the Gorini-Kossakowski-Sudarshan-Lindblad operators  
$ \mathbf{L}_1 = \sqrt{\alpha_{1}}  \, \mathbf{a} $, 
$ \mathbf{L}_2 = \sqrt{\alpha_{2}} \, \mathbf{a}^{\dag} $, 
$ \mathbf{L}_3 = \sqrt{\beta_1} \boldsymbol{\sigma}^- $, 
$ \mathbf{L}_4 = \sqrt{\beta_2} \boldsymbol{\sigma}^+  $
and 
$ \mathbf{L}_5 =  \sqrt{\beta_3} \boldsymbol{\sigma}^z $,
where  $\alpha_{1}, \alpha_{2}, \beta_1, \beta_2,  \beta_3 \geq 0$
and $\boldsymbol{\sigma}^-$, $\boldsymbol{\sigma}^+$ are defined by 
\begin{equation*}
\boldsymbol{\sigma}^+ = 
\begin{pmatrix}
 0 & 1
 \\
 0 & 0
\end{pmatrix}
\text{ and }
\boldsymbol{\sigma}^- = 
\begin{pmatrix}
 0 & 0
 \\
 1 & 0
\end{pmatrix} 
\end{equation*}
(see, e.g., \cite{HarocheRaimond2006}).
The continuous monitoring of the radiation field is characterized  by 
$$
 \mathbf{L}_6 
=   
\sqrt{\alpha_{3} \gamma/2}  \left( \exp{\left (\mathrm{i} \psi\right )} \mathbf{a}^{\dag} 
+ \exp{\left (- \mathrm{i}  \psi \right )} \mathbf{a} \right),
$$
 and 
$
 \mathbf{L}_7 
=  
 \sqrt{\alpha_{3} \left( 1 - \gamma \right)/2}  \left( \exp{\left (\mathrm{i} \psi\right )} \mathbf{a}^{\dag} 
+ \exp{\left (- \mathrm{i}  \psi \right )} \mathbf{a} \right) 
$,
with $ \alpha_{3}, \psi \geq 0$ and $\gamma \in \left] 0 , 1 \right]$ 
(see, e.g., \cite{BarchielliBelavkin1991,BarchielliGregoratti2009,Shabani2014,WisemanMilburn2009}).
\end{system}

The above generic model describes important physical phenomena.
In circuit quantum electrodynamics,
for instance,
the two-level system is realized by an artificial atom (superconducting qubit)
and the field represents a microwave resonator.
Moreover,
the channel $\mathbf{L}_6$ characterizes 
a quantum non-demolition measurement of the squeezed quadrature observable of 
the cavity microwave mode with detection efficiency $\gamma$
(see, e.g., \cite{Agarwal2012,Gambetta2008,Niemczyk2010,Shabani2014,Xiang2013}),
the channel $\mathbf{L}_7$ represents the loss light in the detection of the field quadrature.

\begin{figure}[tb]
\centering
\subfigure{
\includegraphics[height= 0.35in,width=3.17in]{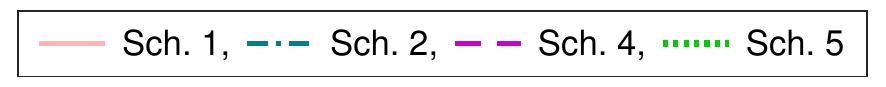}}\\
\addtocounter{subfigure}{-1}
\subfigure[$\mathbb{E}\left( \tr \left( \mathbf{L}_6 \, \boldsymbol{\rho}_t \right) \right)$ with $\Delta = 2^{-9}$]{
\includegraphics[height= 1.0in,width=3.17in]{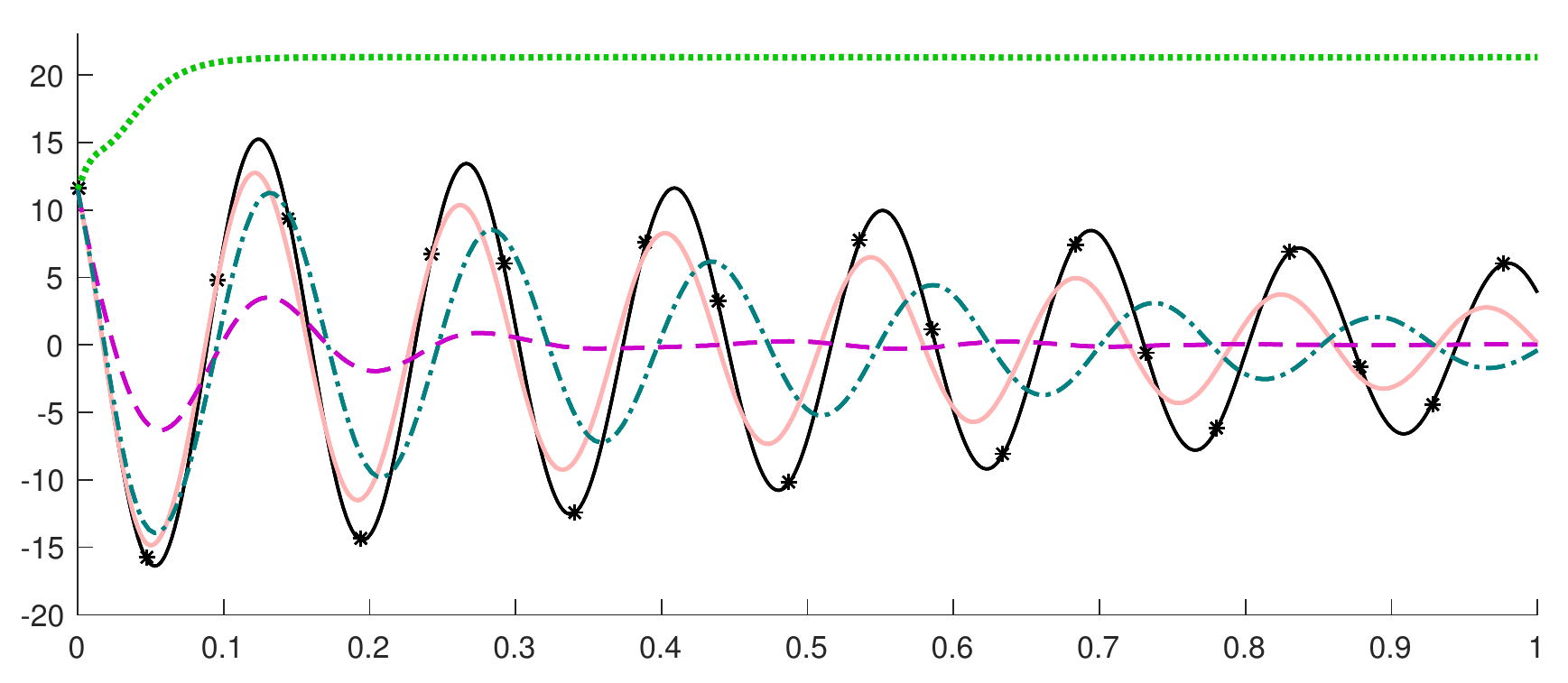}} 
\subfigure[$V\left (\tr \left( \mathbf{L}_6 \boldsymbol{\rho}_t \right)\right )$ with $\Delta=  2^{-9} $]{
\includegraphics[height= 1.0in,width=3.17in]{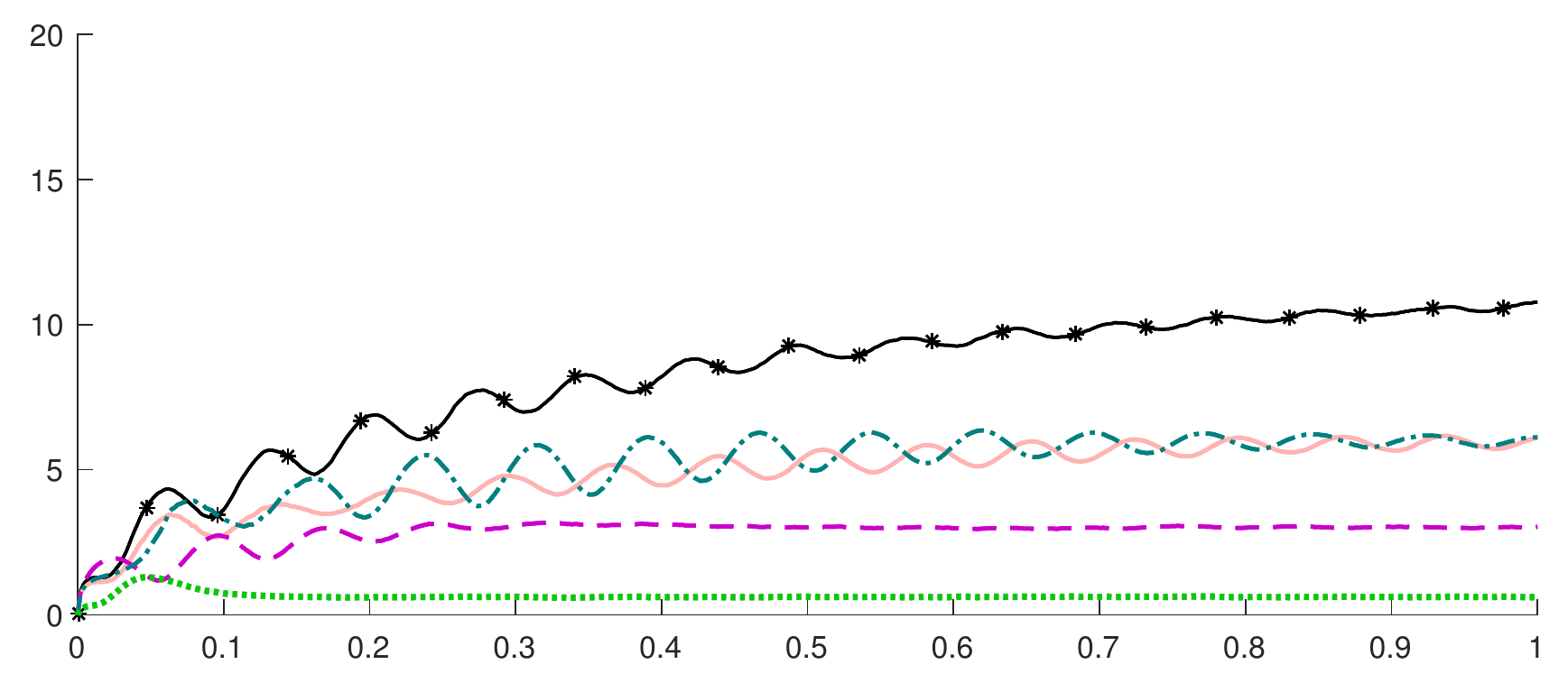}} \\
\subfigure[$\mathbb{E}\left( \tr \left( \mathbf{L}_6 \, \boldsymbol{\rho}_t \right) \right)$ with $\Delta=2^{-11}$]{
\includegraphics[height= 1.0in,width=3.17in]{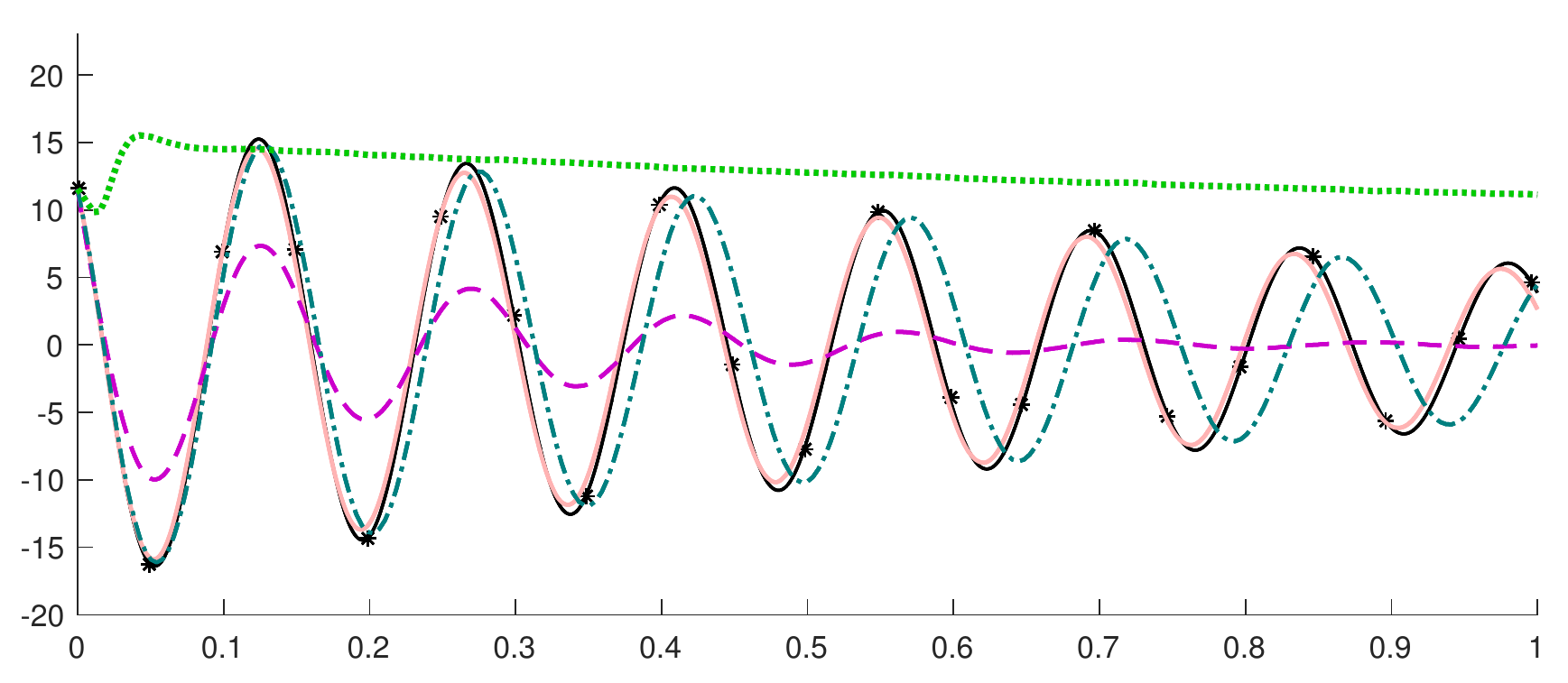}} 
\subfigure[$V\left (\tr \left( \mathbf{L}_6 \boldsymbol{\rho}_t \right)\right )$ with $\Delta=2^{-11}$]{
\includegraphics[height= 1.0in,width=3.17in]{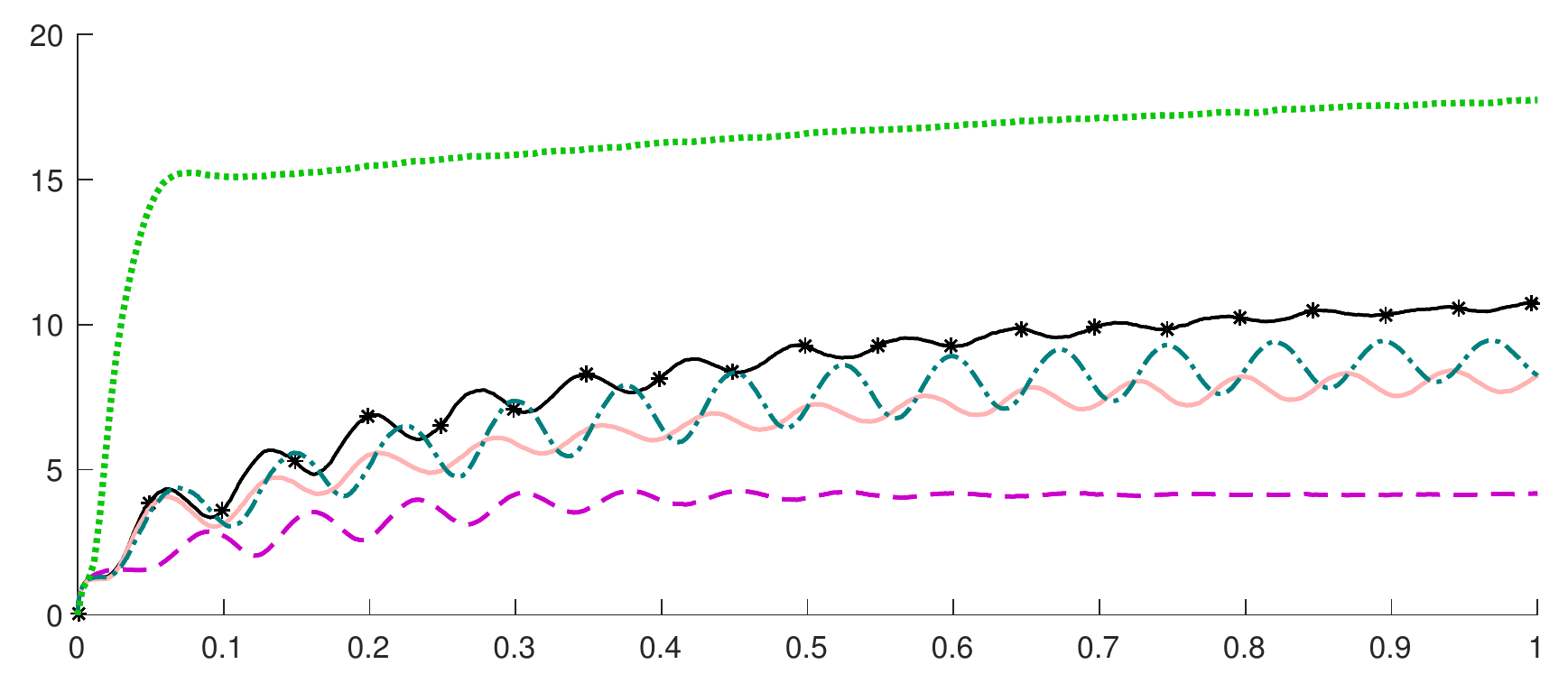}} \\
\subfigure[$\mathbb{E}\left( \tr \left( \mathbf{L}_6 \, \boldsymbol{\rho}_t \right) \right)$ with $\Delta=2^{-14}$]{
\includegraphics[height= 1.0in,width=3.17in]{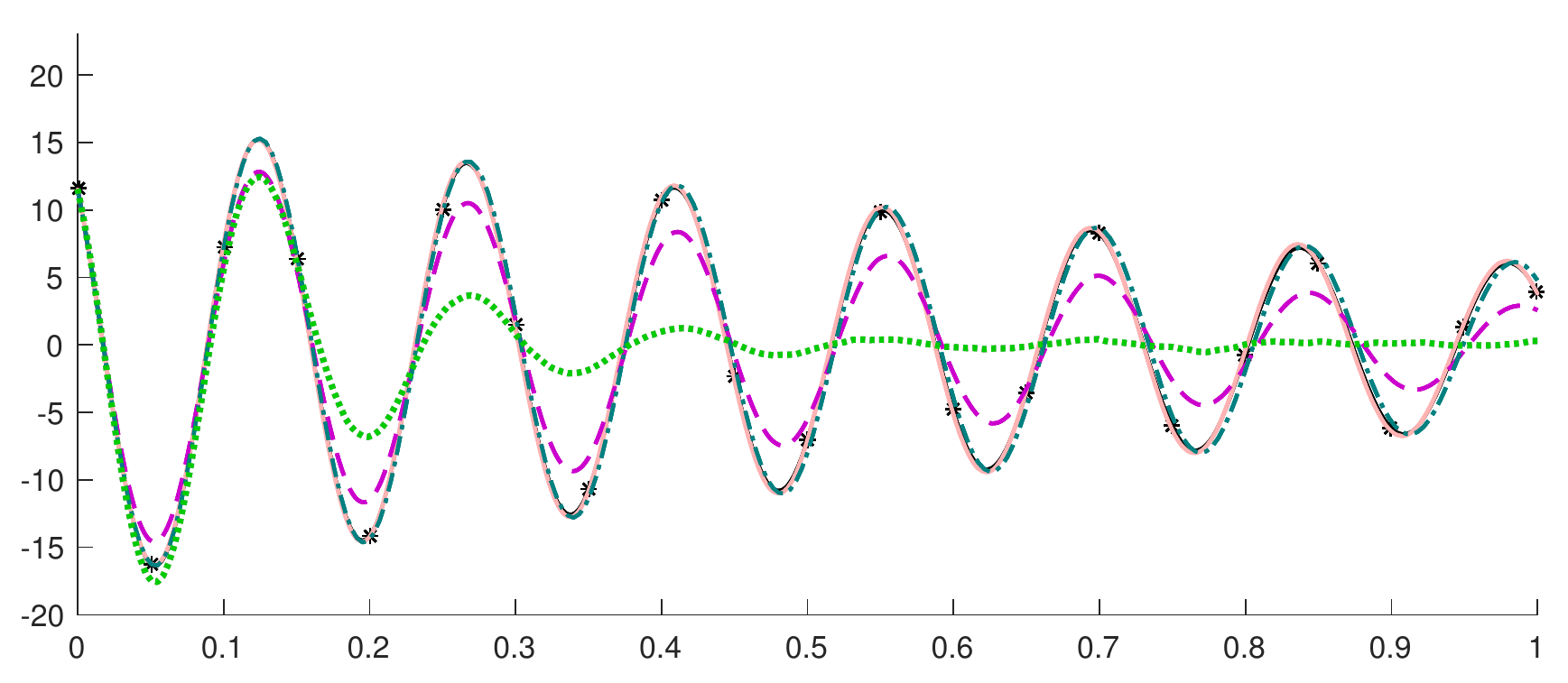}} 
\subfigure[$V\left (\tr \left( \mathbf{L}_6 \boldsymbol{\rho}_t \right)\right )$ with $\Delta=2^{-14}$]{
\includegraphics[height= 1.0in,width=3.17in]{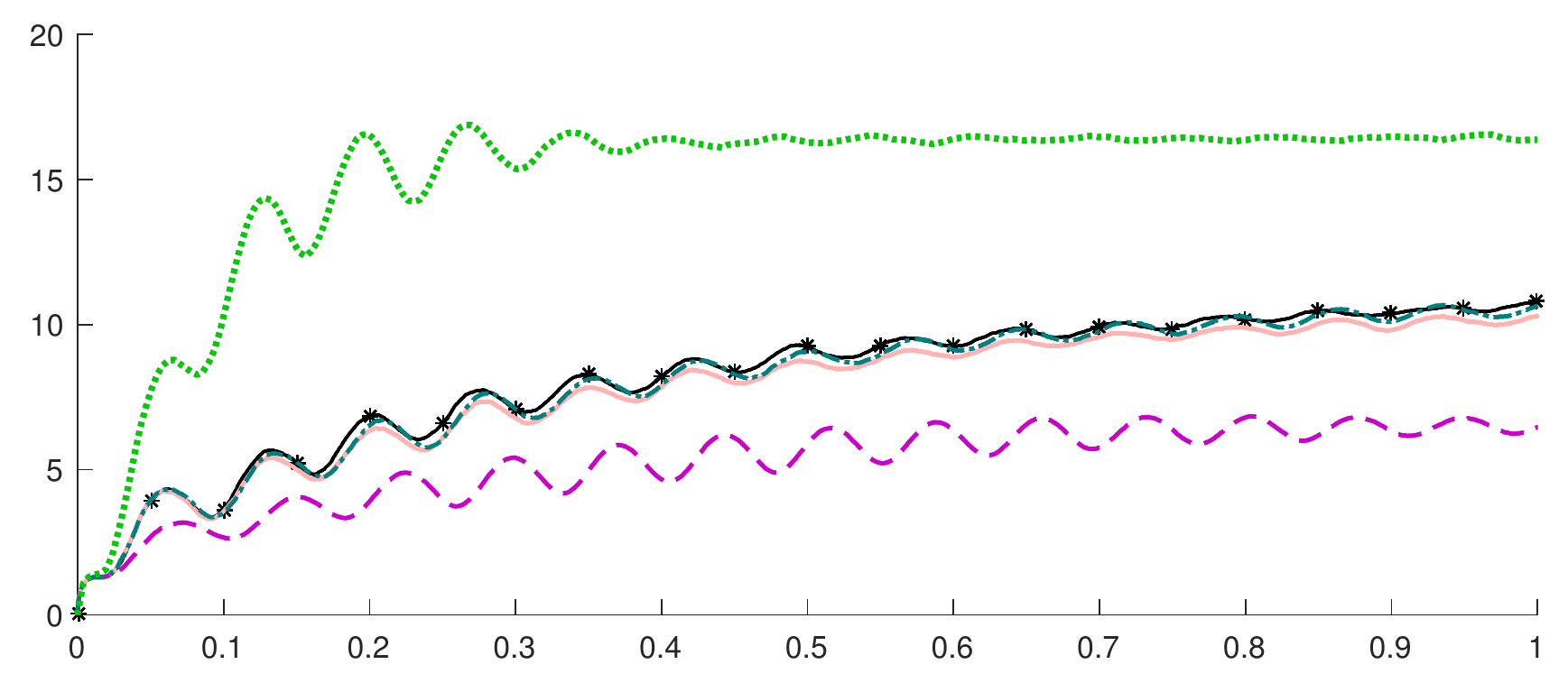}} \\
\caption{Computation of $\mathbb{E}\left( \tr \left( \mathbf{L}_6 \, \boldsymbol{\rho}_t \right) \right)$ and $V\left( \tr \left( \mathbf{L}_6 \, \boldsymbol{\rho}_t \right) \right)$,
where $ \mathbf{L}_6$ and $\boldsymbol{\rho}_t$ are given by System \ref{Ex:DiffusiveSQMEapp} with  $d=30$.
The reference values are represented by a solid line with stars, 
and the time $t$ runs from $0$ to to $1$  nanosecond.
}

\label{Fig:MeanValues}
\end{figure}

\renewcommand{\arraystretch}{1.3}
\begin{table}[tb]
\centering
 \caption{
Estimated values of the error $\epsilon\left (\Delta\right )$ and the relative CPU time $\tau\left (\Delta\right )$ for Schemes \ref{scheme:SME_Euler_Exponential}, \ref{scheme:SME_Exponential_Integral}, \ref{scheme:Euler_Implicit} 
and \ref{scheme:AminiMirrahimiRouchon}. }
 \begin{tabular}{ccccccccc}
\hline \hline
	& $\Delta$      & $2^{-8}$  & $2^{-9}$ & $2^{-10}$ & $2^{-11}$  & $2^{-12}$ & $2^{-13}$ & $2^{-14}$  \\
\hline \hline 
\multirow{4}*{\rotatebox[origin=c]{90}{$\epsilon\left (\Delta\right )$}}
 	& Scheme \ref{scheme:SME_Euler_Exponential} &5.8448 & 4.8383 & 4.0076 & 2.9131 & 1.8881 & 0.9954 & 0.6367 \\
 	& Scheme \ref{scheme:SME_Exponential_Integral} &6.0139 & 4.6793 & 3.9311 & 2.7937 & 1.6056 & 0.8887 & 0.6130 \\
 	& Scheme \ref{scheme:Euler_Implicit} &8.1782 & 7.7518 & 7.2824 & 6.6006 & 6.0178 & 5.1713 & 4.4519 \\
	& Scheme \ref{scheme:AminiMirrahimiRouchon} &7.4041 & 10.174 & 8.6682 & 11.792 & 15.868 & 12.446 & 9.8675 \\ 
 \hline
\multirow{4}*{\rotatebox[origin=c]{90}{$\tau\left (\Delta\right )$}}
	& Scheme \ref{scheme:SME_Euler_Exponential} 
         & $    1.01$ & $    1.97$ & $    3.95$ & $    7.89$ & $   15.82$ & $   31.69$ & $   63.41$ \\
	& Scheme \ref{scheme:SME_Exponential_Integral} 
         & $    1.06$ & $    2.10$ & $    4.18$ & $    8.34$ & $   16.71$ & $   33.39$ & $   66.96$ \\
	& Scheme \ref{scheme:Euler_Implicit} 
         & $    1.00$ & $    1.99$ & $    3.95$ & $    7.90$ & $   15.91$ & $   31.82$ & $   63.37$ \\
	& Scheme \ref{scheme:AminiMirrahimiRouchon} 
         & $  248.28$ & $  497.41$ & $  997.33$ & $ 1991.83$ & $ 3990.66$ & $ 7974.96$ & $16423.67$ \\
\hline \hline
 \end{tabular} 
 \label{Tab:Error}
 \end{table}

The time $t$ is given in nanoseconds.
The resonator frequency of the field is set at $2\pi*7$ $GHz$
(see, e.g., \cite{Niemczyk2010}),
and so  $\omega_2=14\pi$
(see Remark \ref{rem:normalizacion}).
We choose the qubit transition frequency such that $\omega_1= \omega_2 $.
Moreover, we select 
the qubit-mode coupling strength $g$ equal to $0.15\,\omega_2 $, 
and hence the small quantum system reach the ultrastrong-coupling regime
(see, e.g., \cite{Niemczyk2010}).
The  measurement strength $\alpha_{3}$ is set to $ 0.3 \, \omega_2$,
the quadrature angle is $\psi =  \pi/4 $ and the detection efficiency $\gamma$ has been assumed to be $0.9$.
Furthermore, we take 
$ \alpha_1 =2 \, \alpha_2 = \beta_1 = 2 \, \beta_2 = \beta_3 = 0.02 \, \omega_2 . $
The initial density operator is the mixed state described by 
$$
 \boldsymbol{\rho}_{0} 
 = 
 \frac{1}{2} \, \ketbra{3}{3} \otimes 
 \ketbra{ \begin{pmatrix} 1/\sqrt{2} \\ 1 /\sqrt{2} \end{pmatrix}}{\begin{pmatrix} 1/\sqrt{2} \\ 1 /\sqrt{2} \end{pmatrix}}
 +
 \frac{1}{2}  \, \ketbra{4}{4} \otimes
 \ketbra{  \begin{pmatrix} 1 \\ 0 \end{pmatrix}}{ \begin{pmatrix} 1 \\ 0 \end{pmatrix}} ,
 $$
 that is, 
 $\boldsymbol{\rho}_{0} $ is given by \eqref{eq:3.1} with $\mu=2$,
$\mathbf{X}^{1}_{0} = \dfrac{1}{2}   \, \left|  3 \right \rangle  \otimes \left| 
 \begin{pmatrix} 1 \\ 1  \end{pmatrix} \right \rangle $
and 
$\mathbf{X}^{2}_{0} = \dfrac{ \sqrt{2}}{2}  \,  \left|  4 \right \rangle  \otimes \left|  \begin{pmatrix} 1 \\ 0 \end{pmatrix} \right \rangle $.
As usual,  the coherent state $\left|  \alpha \right. \rangle$ is defined by 
$
\left|  \alpha \right \rangle
=  \exp\left(- { \left\vert \alpha \right\vert^2 }/{ 2 } \right) 
\sum_{k = 0}^{\infty}  {\alpha^k} \left| e_k \right \rangle / {\sqrt{ k !}} 
$.

Following Remark \ref{rem:aproximacion} we approximate System \ref{Ex:DiffusiveSQME} 
by the finite-dimensional stochastic quantum master equation described in System \ref{Ex:DiffusiveSQMEapp}.

\begin{system}
\label{Ex:DiffusiveSQMEapp} 
Consider the stochastic quantum master equation \eqref{eq:StochasticQME} with 
$\boldsymbol{\mathfrak{h}} =  \ell^2_d \otimes \mathbb{C}^{2} $,
where $\ell^2_d$ stands for the linear span of $\mathbf{e}_0, \ldots, \mathbf{e}_d$ for any $d \in \mathbb{N}$.
Take 
$$
\mathbf{H}
=
{\omega_1}\,\boldsymbol{\sigma}^z/2 + \omega_2 \, \mathbf{P}_d  \, \mathbf{a}^{\dag} \mathbf{a}  \, \mathbf{P}_d + g   \, \mathbf{P}_d \left( \mathbf{a}^{\dag} + \mathbf{a} \right ) \, \mathbf{P}_d  \, \boldsymbol{\sigma}^x ,
$$
where $\mathbf{P}_d$ is the orthogonal projection of  $\ell^2 \left(\mathbb{Z}_+ \right)$ onto $\ell^2_d$.
Moreover,  set 
$ \mathbf{L}_1  = \sqrt{\alpha_{1}}  \, \mathbf{P}_{d} \mathbf{a} \mathbf{P}_{d}$, 
$ \mathbf{L}_2 = \sqrt{\alpha_{2}} \, \mathbf{P}_d \mathbf{a}^{\dag} \mathbf{P}_d $, 
$ \mathbf{L}_3 = \sqrt{\beta_1} \boldsymbol{\sigma}^- $, 
$ \mathbf{L}_4 = \sqrt{\beta_2} \boldsymbol{\sigma}^+  $,
$ \mathbf{L}_5 =  \sqrt{\beta_3} \boldsymbol{\sigma}^z $,
$$
 \mathbf{L}_6 
=   
\sqrt{\alpha_{3} \gamma / 2} \, \mathbf{P}_d \left( \exp{\left (\mathrm{i} \psi\right )}  \mathbf{a}^{\dag} 
+ \exp{\left (- \mathrm{i}  \psi \right )} \mathbf{a} \right) \mathbf{P}_d 
$$
and
$
 \mathbf{L}_7 
=  
 \sqrt{\alpha_{3}  \left( 1 - \gamma \right) / 2}  \, \mathbf{P}_d \left( \exp{\left (\mathrm{i} \psi\right )} \mathbf{a}^{\dag} 
+ \exp{- \mathrm{i}  \psi } \mathbf{a} \right) \mathbf{P}_d 
$.
\end{system}

Adopt the framework of System \ref{Ex:DiffusiveSQMEapp}.
By pondering the computational time for running Scheme \ref{scheme:AminiMirrahimiRouchon}
and 
the match between System \ref{Ex:DiffusiveSQME} and \ref{Ex:DiffusiveSQMEapp},
we take $d=30$.
First,
we test Schemes \ref{scheme:SME_Euler_Exponential}, \ref{scheme:SME_Exponential_Integral}, \ref{scheme:Euler_Implicit} and  \ref{scheme:AminiMirrahimiRouchon}
by computing  
$
\mathbb{E} \,  \tr \left( \mathbf{L}_6  \, \boldsymbol{\rho}_t  \right)
$
with $0 \leq t \leq 1$.
For any $\mathbf{A} \in \boldsymbol{\mathfrak{L}} \left( \ell^2_d \otimes \mathbb{C}^{2} \right)$,
$
\mathbb{E} \,  \tr \left( \mathbf{A} \boldsymbol{\rho}_t  \right) =  \tr \left( \mathbf{A} \, \mathbb{E} \,  \boldsymbol{\rho}_t  \right)
$.
Since 
$\mathbb{E} \,  \boldsymbol{\rho}_t $ satisfies the quantum master equation \eqref{eq:1.3}
and  $\dim \left(  \boldsymbol{\mathfrak{h}} \right) = 2d$ is in the range of a few ten,
we get the reference values of 
$\mathbb{E} \,  \tr \left( \mathbf{L}_6  \, \boldsymbol{\rho}_t  \right)$
by calculating the explicit solution of \eqref{eq:1.3}.
Second, we compute  $V \left( \tr \left( \mathbf{L}_6 \, \boldsymbol{\rho}_t  \right) \right)$,
the variance of $\tr \left( \mathbf{L}_6  \boldsymbol{\rho}_t  \right)$. The reference values for the variance have been calculated by 
sampling $10^4$ times Scheme \ref{scheme:SME_Exponential_Integral} with $\Delta=2^{-22}$.
It is worth pointing out that the maximal gap between the  variances of $\tr \left( \mathbf{L}_6  \, \bar{\boldsymbol{\rho}}_{n}  \right)$
produced by Schemes \ref{scheme:SME_Euler_Exponential} and \ref{scheme:SME_Exponential_Integral}
with $\Delta=2^{-22}$ is of order of $0.008$, while the maximal difference for Schemes \ref{scheme:SME_Euler_Exponential} and \ref{scheme:Euler_Implicit} is $0.2296$. 

Figure \ref{Fig:MeanValues} displays estimations of 
$\mathbb{E} \,  \tr \left(  \mathbf{L}_6  \, \boldsymbol{\rho}_t  \right)$ and $V \left( \tr \left( \mathbf{L}_6 \, \boldsymbol{\rho}_t  \right) \right)$
obtained from sampling $10^4$ times
Schemes \ref{scheme:SME_Euler_Exponential}, \ref{scheme:SME_Exponential_Integral}, 
\ref{scheme:Euler_Implicit} and \ref{scheme:AminiMirrahimiRouchon}
with step sizes $\Delta$ equal to $2^{-9}$, $2^{-11}$ and $2^{-14}$.
Figure \ref{Fig:MeanValues} shows that 
Schemes \ref{scheme:SME_Euler_Exponential}  and \ref{scheme:SME_Exponential_Integral}
reproduce very well the oscillatory behavior of $\mathbb{E} \,  \tr \left( \mathbf{L}_6  \boldsymbol{\rho}_t  \right)$.
The first part of Table \ref{Tab:Error} provides the errors
$$
 \epsilon \left( \Delta \right)
=
\max \left\{ 
\left\vert V \left( \tr \left( \mathbf{L}_6  \bar{\boldsymbol{\rho}}_{n}  \right) \right) - V \left( \tr \left( \mathbf{L}_6  \boldsymbol{\rho}_{n \Delta}  \right)  \right) \right\vert
:
n = 0 ,  \ldots,  \lfloor 1 / \Delta \rfloor
\right\} ,
$$
where
each numerical method $\bar{\boldsymbol{\rho}}_{n}$ is sampled $10^4$ times.
Figure \ref{Fig:MeanValues} 
together with Table \ref{Tab:Error},
point out the superior accuracy of the exponential numerical methods 
Schemes \ref{scheme:SME_Euler_Exponential}  and \ref{scheme:SME_Exponential_Integral}
over 
Schemes \ref{scheme:Euler_Implicit} and \ref{scheme:AminiMirrahimiRouchon}.
Furthermore,
Scheme \ref{scheme:AminiMirrahimiRouchon} has experienced serious difficulties
in approximating System \ref{Ex:DiffusiveSQMEapp}.

The second part of Table \ref{Tab:Error} presents the relative mean CPU time $\tau\left (\Delta\right )$ used for processing 
 $100$ realizations of Schemes
\ref{scheme:SME_Euler_Exponential}, \ref{scheme:SME_Exponential_Integral}, \ref{scheme:Euler_Implicit}
and \ref{scheme:AminiMirrahimiRouchon} with  step size $\Delta$.
Each $\tau\left (\Delta\right )$ has been estimated by averaging $20$ batches of $100$ numerical approximations.
For a given $\Delta$,
we report the ratio of the CPU time for each numerical method to the minimum of all CPU times,
which corresponds to Scheme  \ref{scheme:Euler_Implicit}   with $\Delta = 2^{-8}$.
Table \ref{Tab:Error} shows the advantages of solving \eqref{eq:StochasticQME} by means of 
the representation \eqref{eq:Representation}.
In particular, the computational costs of
Schemes \ref{scheme:SME_Euler_Exponential}, \ref{scheme:SME_Exponential_Integral} and \ref{scheme:Euler_Implicit}
are significantly lower  than the one of Scheme \ref{scheme:AminiMirrahimiRouchon}.
Moreover, the CPU time used by Scheme \ref{scheme:SME_Exponential_Integral}
is close to that of Scheme \ref{scheme:SME_Euler_Exponential}, thanks to Theorem \ref{th:Pade1}.

\begin{remark}
\label{rem:dimension}
We examined the dimensions for which 
System \ref{Ex:DiffusiveSQMEapp} can be simulated running Scheme \ref{scheme:SME_Euler_Exponential}
on a current basic computer.
To this end,
we ran Scheme \ref{scheme:SME_Euler_Exponential} with $\Delta = 2^{-14} $ on a  
3,3 GHz Intel Core i5 with 32 GB RAM.
Recall that the dimension of the state space of System \ref{Ex:DiffusiveSQMEapp} is $2d$,
and so here the density operators have $4 \cdot d^2$ elements.
In the case where the matrix exponentials are calculated applying 
Pad\'e approximants together with the scaling and squaring method,
we could run Scheme \ref{scheme:SME_Euler_Exponential} up to $ d = 23000 $,
and we were only  able to compute the explicit solution of \eqref{eq:1.3} if $d \leq 107 $.
Using Krylov subspace methods for the computation of matrix exponentials 
we ran Scheme \ref{scheme:SME_Euler_Exponential} even with $d = 5 \cdot10^7$.
Table \ref{Tab:TiemposCorrida} shows 
the dependence between $d$ and 
the time spent for ten integration steps of Scheme \ref{scheme:SME_Euler_Exponential} with $\Delta = 2^{-14} $.
\end{remark}

\renewcommand{\arraystretch}{1.3}
\begin{table}[tb]
\centering
 \caption{
Estimated values of the time (given in seconds) 
spent for a simple implementation of ten integration steps of Scheme \ref{scheme:SME_Euler_Exponential} 
based on Krylov subspace methods.}
 \begin{tabular}{ccccccccc}
 \hline
	 $d$ &  $3 \cdot 10^{3}$ & $ 10^{4}$ & $2 \cdot 10^{4}$ & $ 10^{5}$  &  $5 \cdot 10^{5}$  & $ 10^{6}$  \\
 \hline 
 Desktop 3,3 GHz Intel Core i5 with 32 GB RAM   & 1.8 & 3.1 & 4.9  & 24.9 & 162.1 & 355.5 \\ 
 Notebook 2,5 GHz Intel Core i7 with 16 GB RAM & 1.2  & 3.0  & 6.5  & 25.3 &	157.6 & 301.2  \\
\hline
 \end{tabular} 
 \label{Tab:TiemposCorrida}
 \end{table}

\subsection{Direct detection}
\label{subsec:DirectDetection}
 
This subsection concerns with the numerical simulation of 
non-autonomous stochastic quantum master equations involving diffusion and jump terms.
We study the following driven open quantum system,
where a monochromatic laser light is applied to the field mode of the small quantum system under consideration.

\begin{system}
\label{Ex:DrivenSQME}
We take 
$$
\mathbf{H}
= 
\mathbf{H}_{Rabi} 
+
\epsilon \left( \exp{ \left (- \mathrm{i} \left( \omega_3 t + \phi \right) \right )} \mathbf{a}^{\dag}  
+ \exp{\left ( \mathrm{i} \left( \omega_3 t + \phi \right) \right )} \mathbf{a} \right)  ,
$$
acting upon $ \boldsymbol{\mathfrak{h}}= \ell^2 \left(\mathbb{Z}_+ \right)\otimes \mathbb{C}^{2}$,
and so 
the small quantum system is driven by a coherent field of amplitude $\epsilon > 0$, 
frequency $\omega_3 \in \mathbb{R}$ 
and phase $ \phi \in \mathbb{R} $
(see, e.g., \cite{Gambetta2008}).
A noiseless counter is determined by 
$\mathbf{R}_1 = \sqrt{\beta_1  \gamma} \boldsymbol{\sigma}^-$,
and 
the channel $\mathbf{R}_2 = \sqrt{\beta_1 \left( 1 -  \gamma \right)} \boldsymbol{\sigma}^-$ 
represents the loss light in the direct detection,
where 
$\gamma \in \left] 0 , 1 \right]$ is the fraction of the detected fluorescent light
and $\beta_1 > 0$
(see, e.g., \cite{BarchielliGregoratti2012}).
The interaction between the two level system and a thermal bath is described by 
$
\mathbf{R}_3 = \sqrt{\beta_2 } \boldsymbol{\sigma}^-, 
\mathbf{R}_4 = \sqrt{\beta_3} \boldsymbol{\sigma}^+,
\mathbf{R}_5 = \sqrt{\beta_4} \boldsymbol{\sigma}^z  
$,
as well as the interaction of the field mode with an independent reservoir is modeled by 
$\mathbf{L}_1 = \sqrt{\alpha_{1}} \, \mathbf{a} $
and 
$\mathbf{L}_2 = \sqrt{\alpha_{2}} \, \mathbf{a}^{\dag} $.
Here, $\alpha_{1}, \alpha_{2}, \beta_2, \beta_3, \beta_4 \geq 0$
\end{system}

As in Section \ref{subsec:DiffusiveSQMEs},
we choose $\omega_1= \omega_2 = 14\pi$ and   $g = 0.15 \,\omega_2 $.
Moreover, 
we consider a coherent field with amplitude $\epsilon = 0.15 \, \omega_{2}$,
frequency $\omega_3 = 1.1\, \omega_2$ and phase $\phi = \pi/4$.
The  measurement strength $\beta_{1}$ is set to $ 0.3 \, \omega_1$ 
and the fraction of the detected fluorescent light $\gamma$ has been assumed  $0.9$.
Finally, we select 
$ \alpha_1 =2 \, \alpha_2 = 0.02 \, \omega_2 $
and
$  \beta_2 = 2\,\beta_3 = 2\,\beta_4 = 0.01\omega_{1} $.
The initial density operator is the mixed state described by 
\eqref{eq:3.1} with $\mu=2$,
$\mathbf{X}^{1}_{0} =  \dfrac{1}{2} \, \left|  2 \right \rangle  \otimes \left| 
 \begin{pmatrix} 1 \\ 1 \end{pmatrix} \right \rangle $
and 
$\mathbf{X}^{2}_{0} =  \dfrac{ \sqrt{2}}{2} \, \left|  1 \right \rangle  \otimes \left|  \begin{pmatrix} 1 \\ 0 \end{pmatrix} \right \rangle $.

\begin{figure}[H]
\centering

\subfigure{
\includegraphics[height= 0.25in,width=3.17in]{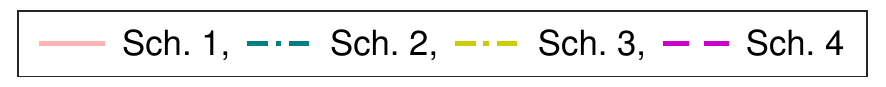}}
	\\
\addtocounter{subfigure}{-1}
\subfigure[$\mathbb{E}\,\left (\tr \left( \mathbf{L}_1 \boldsymbol{\rho}_t \right)\right )$ with $\Delta=2^{-6}$]{
\includegraphics[height= 1in,width=3.17in]{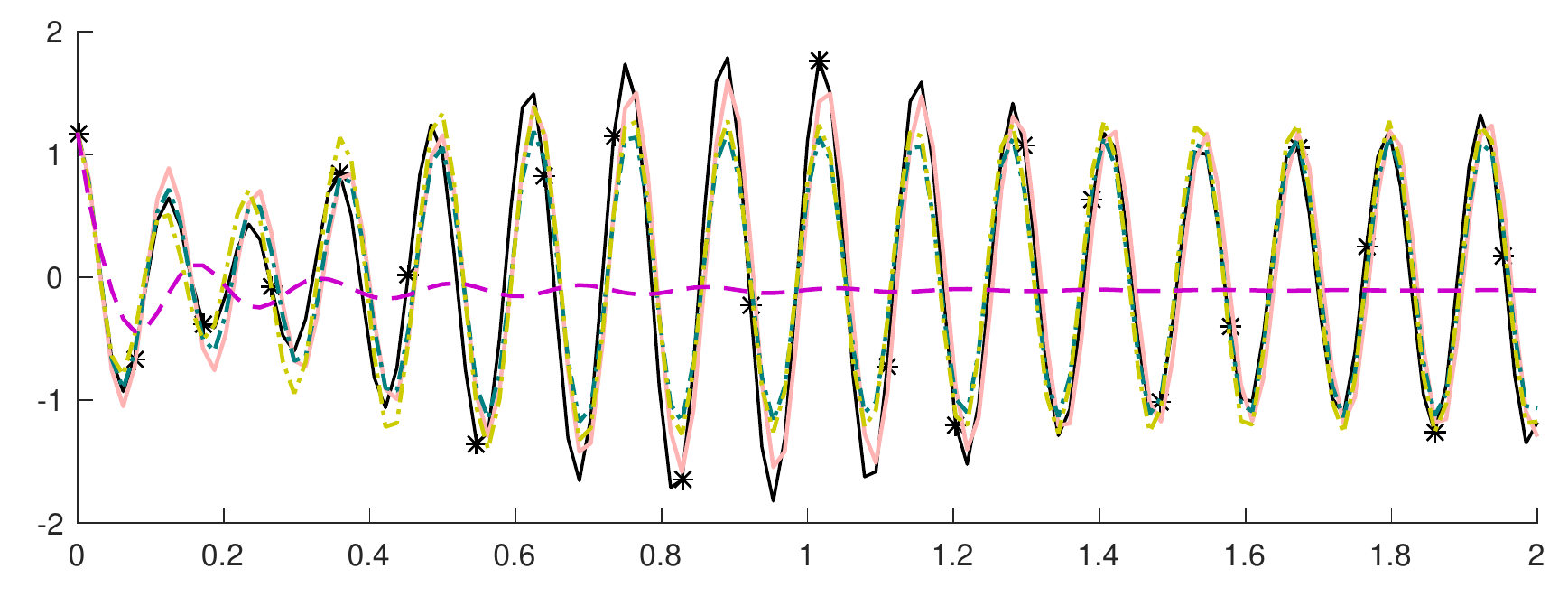}} 
\subfigure[$\mathbb{E}\,\left (\tr \left( \mathbf{R}_1 \boldsymbol{\rho}_t \right)\right )$ with $\Delta=2^{-6}$]{
\includegraphics[height= 1in,width=3.17in]{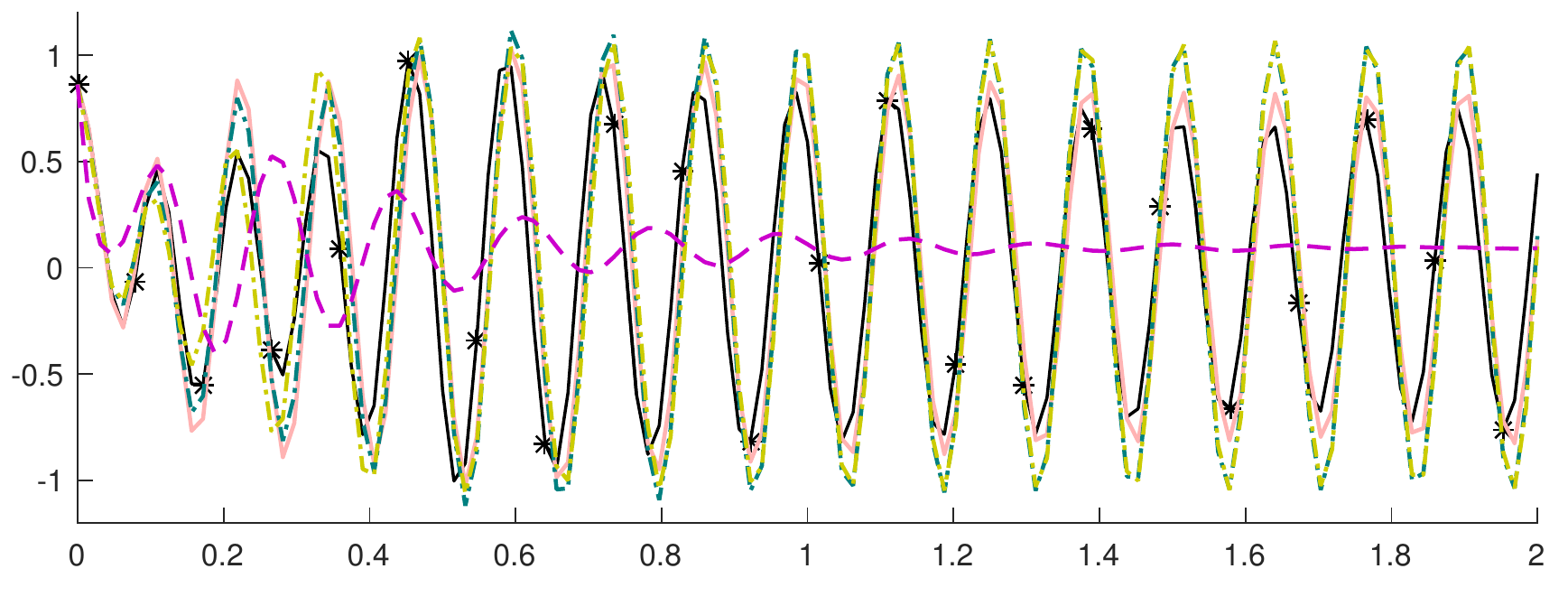}}
\\
\subfigure[$\mathbb{E}\,\left (\tr \left( \mathbf{L}_1 \boldsymbol{\rho}_t \right)\right )$ with $\Delta=  2^{-9} $]{
\includegraphics[height= 1in,width=3.17in]{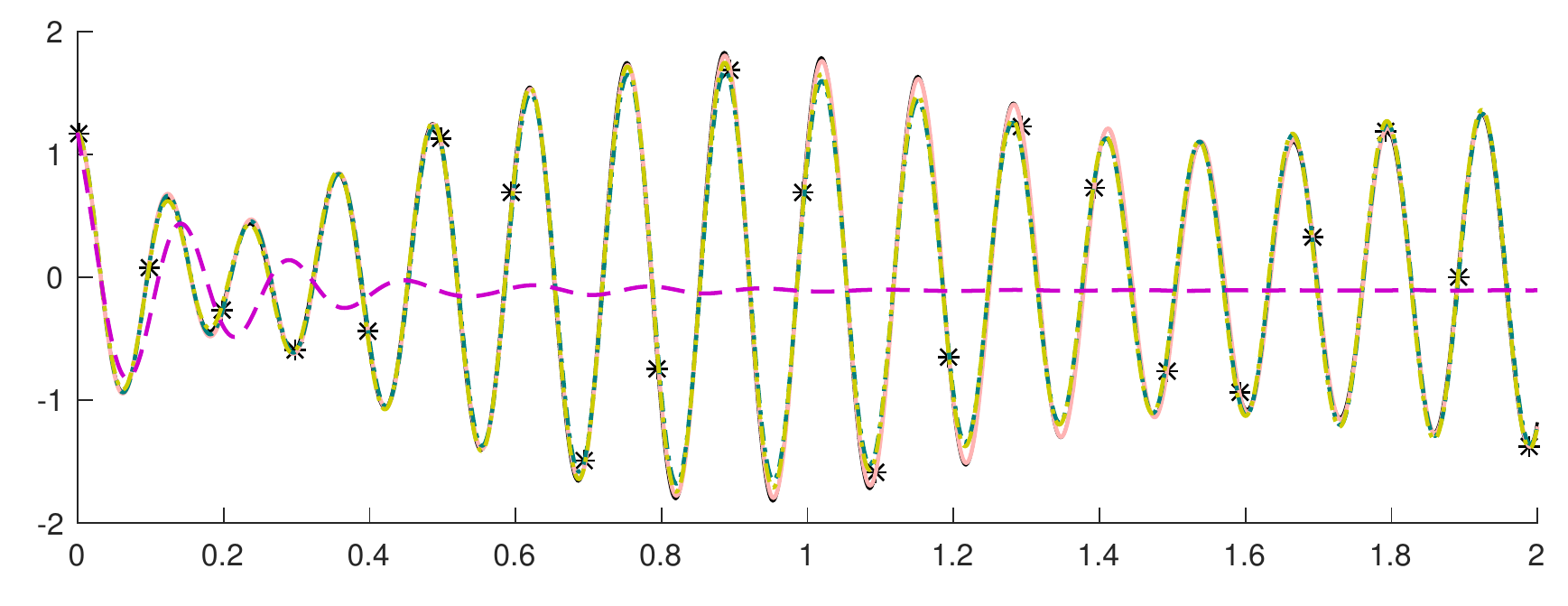}} 
\subfigure[$\mathbb{E}\,\left (\tr \left( \mathbf{R}_1 \boldsymbol{\rho}_t \right)\right )$ with $\Delta=  2^{-9} $]{
\includegraphics[height= 1in,width=3.17in]{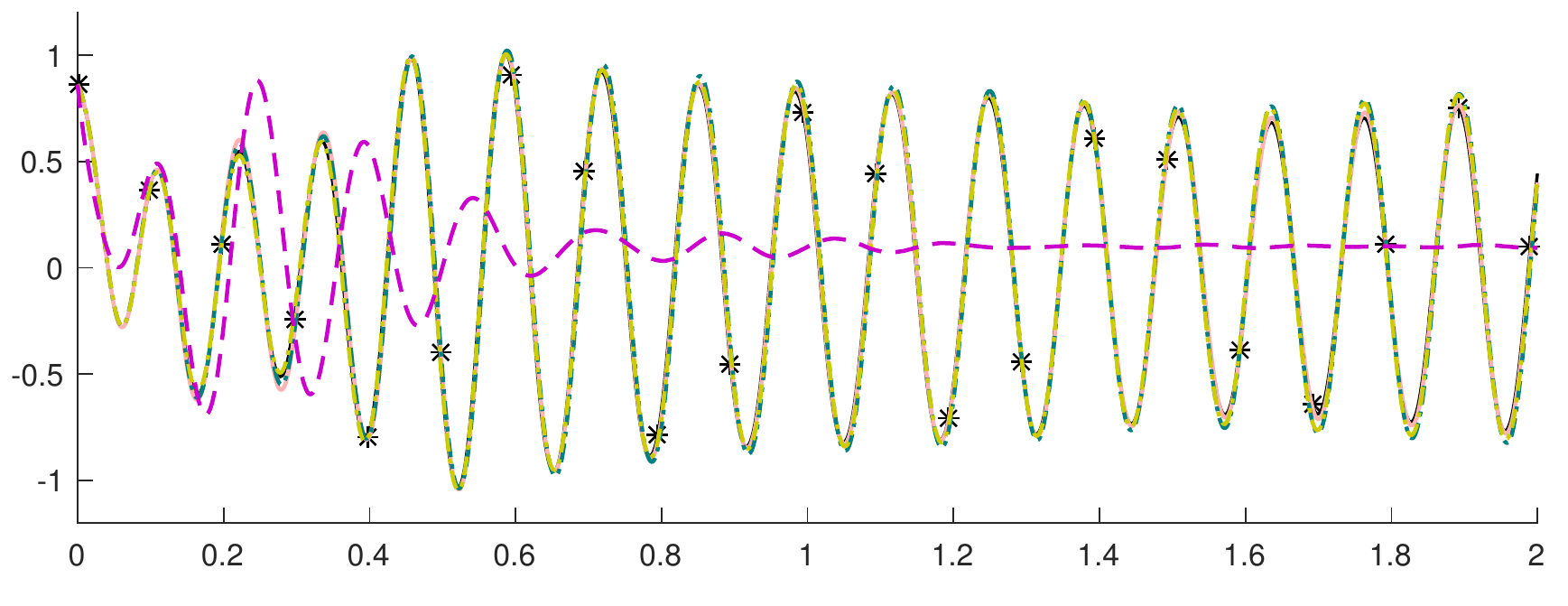}} \\
\subfigure[$\mathbb{E}\,\left (\tr \left( \mathbf{L}_1 \boldsymbol{\rho}_t \right)\right )$ with $\Delta=2^{-12}$]{
\includegraphics[height= 1in,width=3.17in]{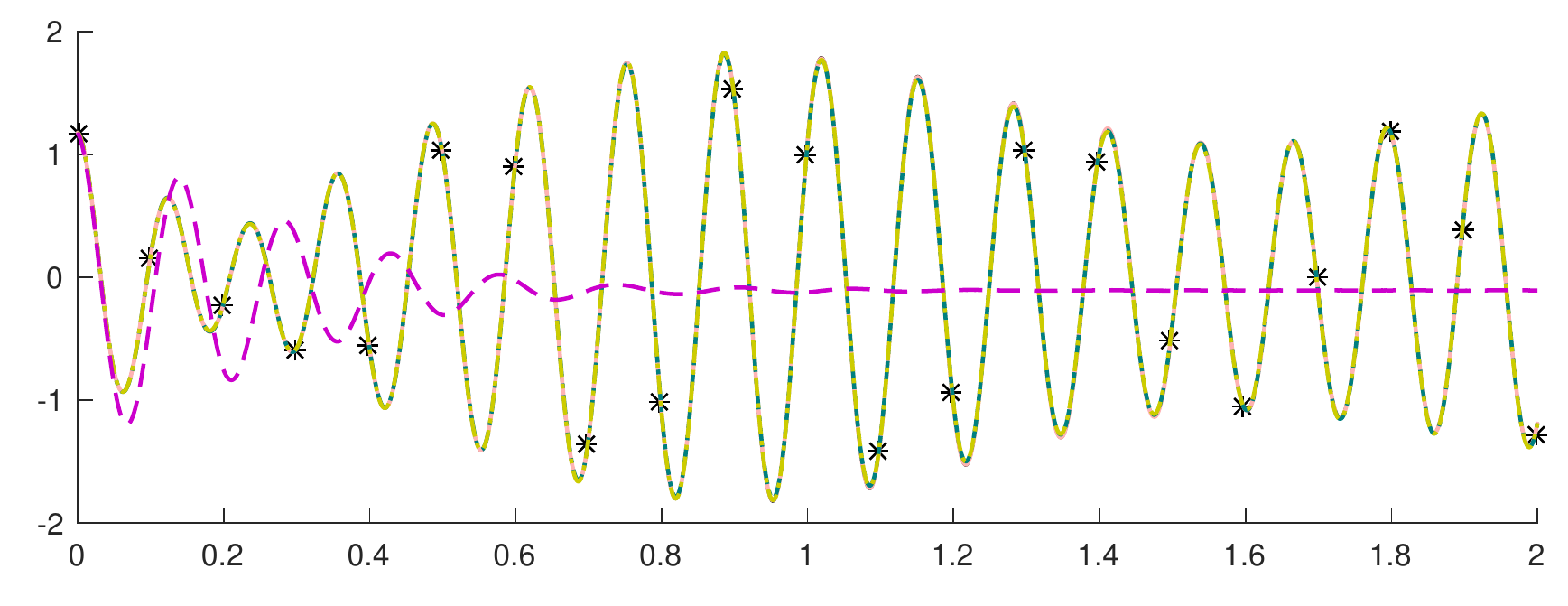}} 
\subfigure[$\mathbb{E}\,\left (\tr \left( \mathbf{R}_1 \boldsymbol{\rho}_t \right)\right )$ with $\Delta=2^{-12}$]{
\includegraphics[height= 1in,width=3.17in]{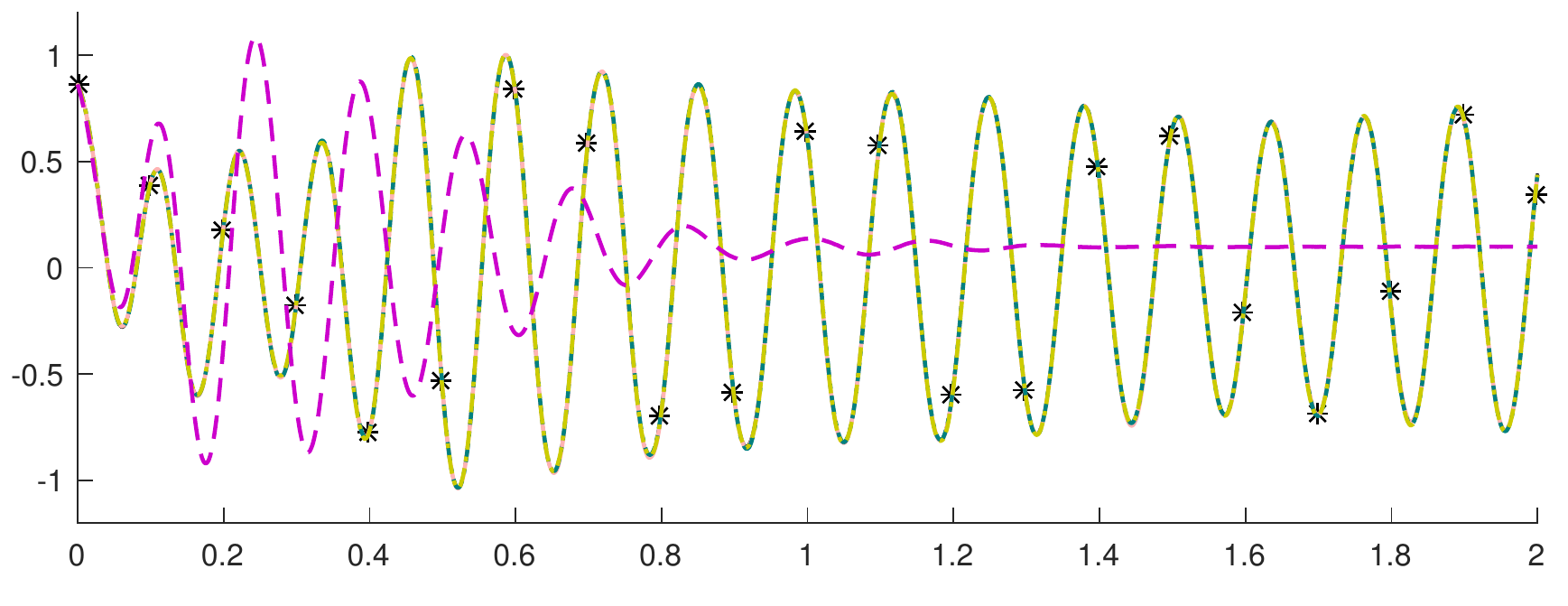}} \\
\caption{ 
Consider System \ref{Ex:DrivenSQMEapp} with $d=50$.
We present the values of 
$\mathbb{E}\,\left (\tr \left( \mathbf{L}_1 \boldsymbol{\rho}_t \right)\right )$ 
and $\mathbb{E}\,\left (\tr \left( \mathbf{R}_1 \boldsymbol{\rho}_t \right)\right )$
obtained from $20000$  samples of Schemes \ref{scheme:SME_Euler_Exponential}, \ref{scheme:SME_Exponential_Integral}, \ref{scheme:SME_ExponentialExponential_Integral_GD} and \ref{scheme:Euler_Implicit} with step-size $\Delta$.
The reference values are represented by a solid line with stars, 
 and the time scale is in nanoseconds.} 
\label{Fig:Ej2A2}
\end{figure}

\renewcommand{\arraystretch}{1.3}
\begin{table}[tb]
\begin{center}

\begin{tabular}{ccccccccccc}
\hline\hline
& $\Delta$ &$2^{- 6}$&$2^{- 7}$&$2^{- 8}$&$2^{- 9}$&$2^{-10}$&$2^{-11}$&$2^{-12}$&$2^{-13}$&$2^{-14}$ \\
\hline\hline
\multirow{4}*{\rotatebox[origin=c]{90}{$\epsilon_2\left (\Delta\right )$}} & Scheme \ref{scheme:SME_Euler_Exponential} 
         & $0.615$ & $0.315$ & $0.159$ & $0.083$ & $0.038$ & $0.029$ & $0.015$ & $0.012$ & $0.010$  \\
& Scheme \ref{scheme:SME_Exponential_Integral} 
         & $0.677$ & $0.532$ & $0.354$ & $0.200$ & $0.107$ & $0.052$ & $0.031$ & $0.016$ & $0.009$ \\
& Scheme \ref{scheme:SME_ExponentialExponential_Integral_GD} 
         & $0.638$ & $0.426$ & $0.281$ & $0.170$ & $0.093$ & $0.046$ & $0.035$ & $0.012$ & $0.007$ \\
& Scheme \ref{scheme:Euler_Implicit} 
         & $1.878$ & $1.929$ & $1.953$ & $1.948$ & $1.943$ & $1.934$ & $1.920$ & $1.905$ & $1.879$  \\
\hline
\multirow{4}*{\rotatebox[origin=c]{90}{$\epsilon_3\left (\Delta\right )$}} & Scheme \ref{scheme:SME_Euler_Exponential} 
         & $0.620$ & $0.345$ & $0.176$ & $0.088$ & $0.051$ & $0.024$ & $0.016$ & $0.012$ & $0.014$  \\
& Scheme \ref{scheme:SME_Exponential_Integral} 
         & $0.594$ & $0.366$ & $0.222$ & $0.118$ & $0.059$ & $0.032$ & $0.021$ & $0.017$ & $0.013$  \\
& Scheme \ref{scheme:SME_ExponentialExponential_Integral_GD} 
         & $0.664$ & $0.337$ & $0.162$ & $0.086$ & $0.045$ & $0.026$ & $0.019$ & $0.011$ & $0.009$  \\
& Scheme \ref{scheme:Euler_Implicit} 
         & $1.064$ & $1.134$ & $1.270$ & $1.388$ & $1.504$ & $1.596$ & $1.650$ & $1.696$ & $1.739$ \\
\hline
\multirow{4}*{\rotatebox[origin=c]{90}{$\tau\left (\Delta\right )$}}& Scheme \ref{scheme:SME_Euler_Exponential} 
         & $    1.25$ & $    2.45$ & $    4.71$ & $    9.09$ & $   18.08$ & $   34.07$ & $   67.17$ & $  134.44$ & $  262.67$ \\
& Scheme \ref{scheme:SME_Exponential_Integral} 
         & $    1.60$ & $    2.91$ & $    5.56$ & $   10.41$ & $   20.65$ & $   38.22$ & $   76.52$ & $  153.14$ & $  305.68$ \\
& Scheme \ref{scheme:SME_ExponentialExponential_Integral_GD} 
         & $    1.89$ & $    3.45$ & $    6.40$ & $   11.47$ & $   22.65$ & $   43.21$ & $   86.23$ & $  173.21$ & $  343.59$ \\
& Scheme \ref{scheme:Euler_Implicit} 
         & $    1.00$ & $    1.96$ & $    3.75$ & $    7.53$ & $   14.89$ & $   29.88$ & $   59.73$ & $  119.22$ & $  238.68$ \\
\hline\hline
\end{tabular}
\caption{Estimated values of $\epsilon_2 \left (\Delta \right )$ and $\epsilon_3 \left (\Delta \right )$ for Schemes \ref{scheme:SME_Euler_Exponential} -  \ref{scheme:Euler_Implicit},
obtained from $20000$ observations of System \ref{Ex:DrivenSQMEapp} with $d=50$. }
\label{Tab:Ej2OpA1}
\end{center}
\end{table}

\begin{figure}[tb]
\centering
\subfigure{
\includegraphics[height= 0.25in,width=3.17in]{LegendFigure2.pdf}}\\\subfigure[$\Delta=  2^{-6} $]{
\includegraphics[height= 1in,width=3.17in]{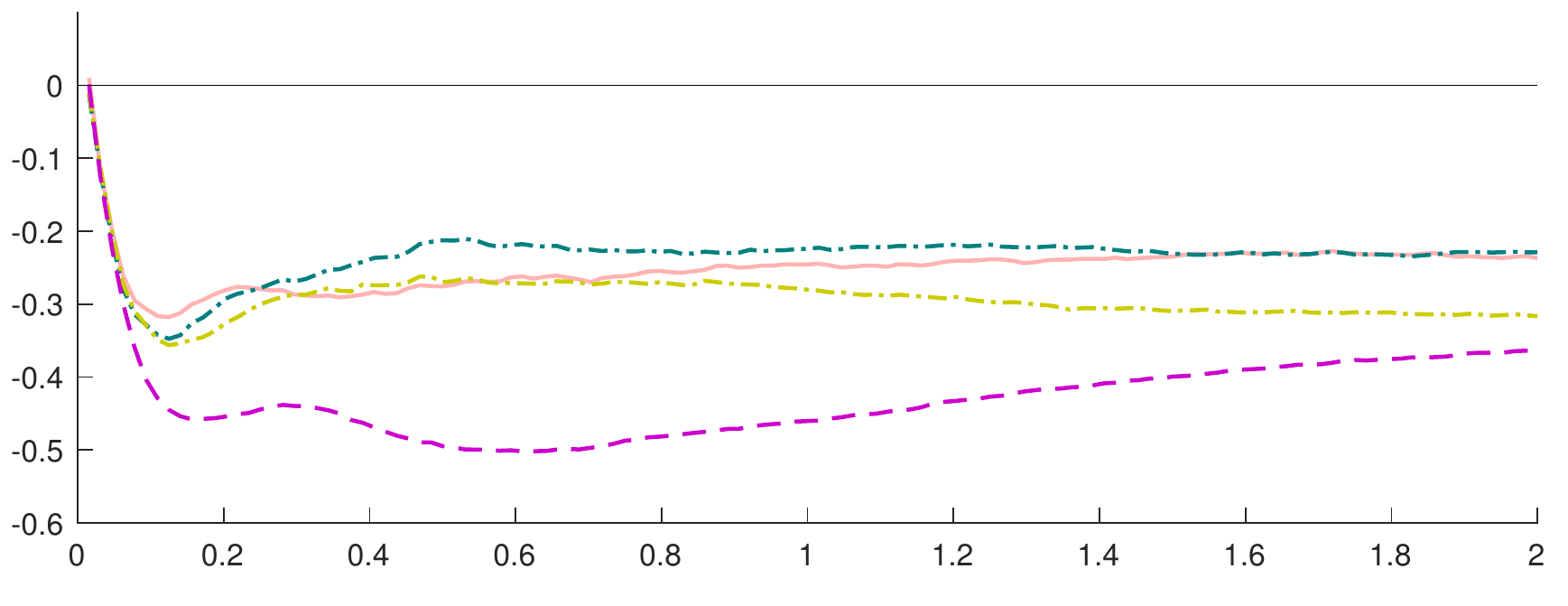}} 
\subfigure[$\Delta=  2^{-9} $]{
\includegraphics[height= 1in,width=3.17in]{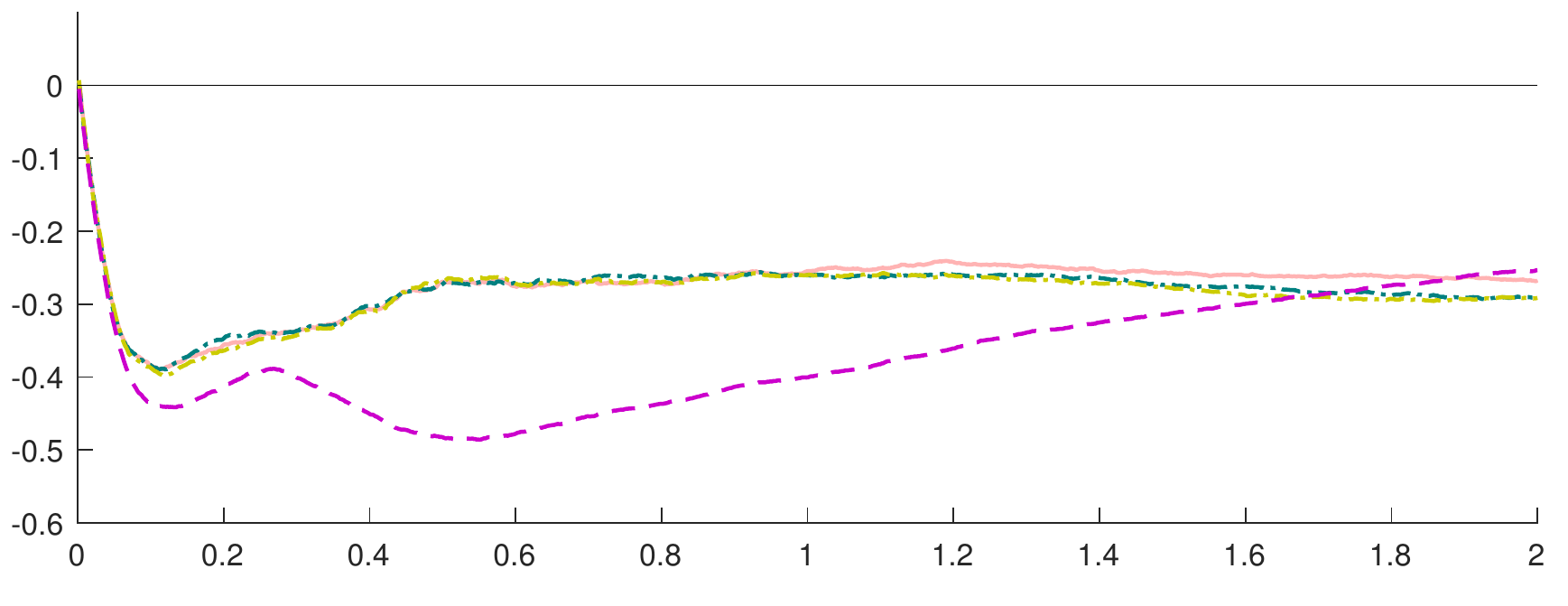}} \\
\subfigure[$\Delta=2^{-12}$]{
\includegraphics[height= 1in,width=3.17in]{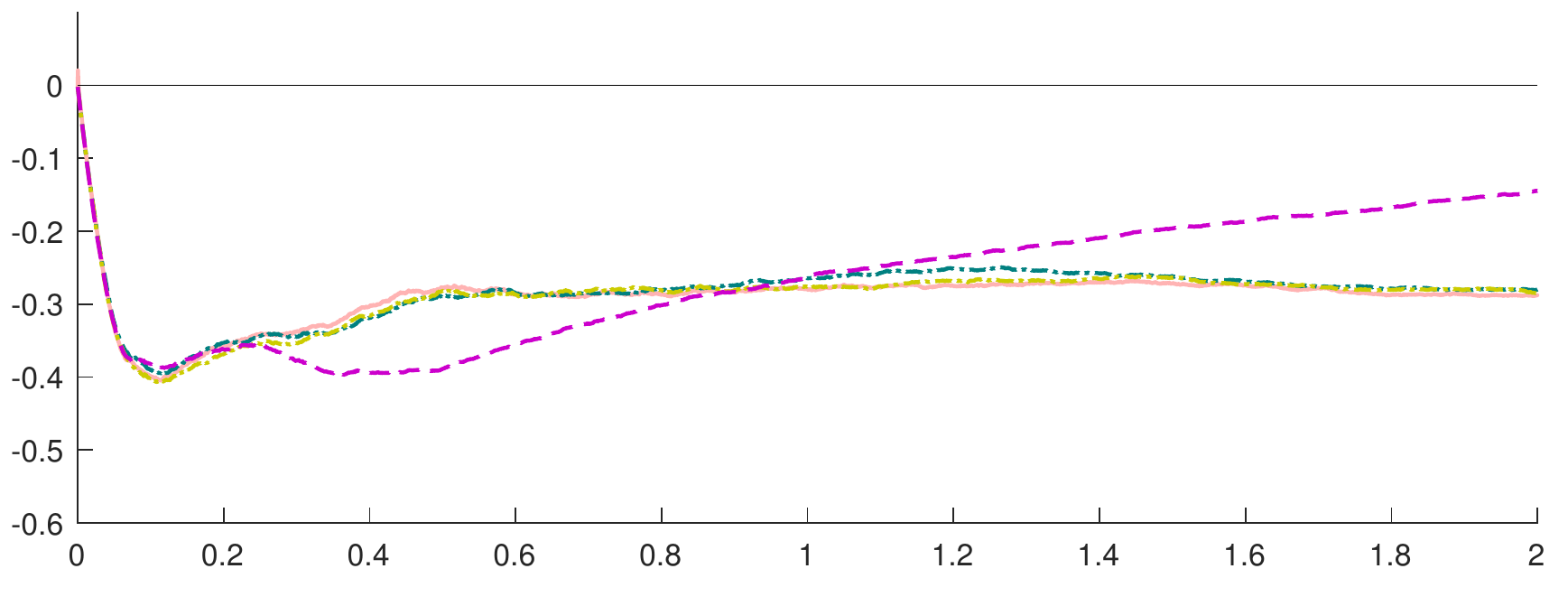}} 
\subfigure[$\Delta=2^{-15}$]{
\includegraphics[height= 1in,width=3.17in]{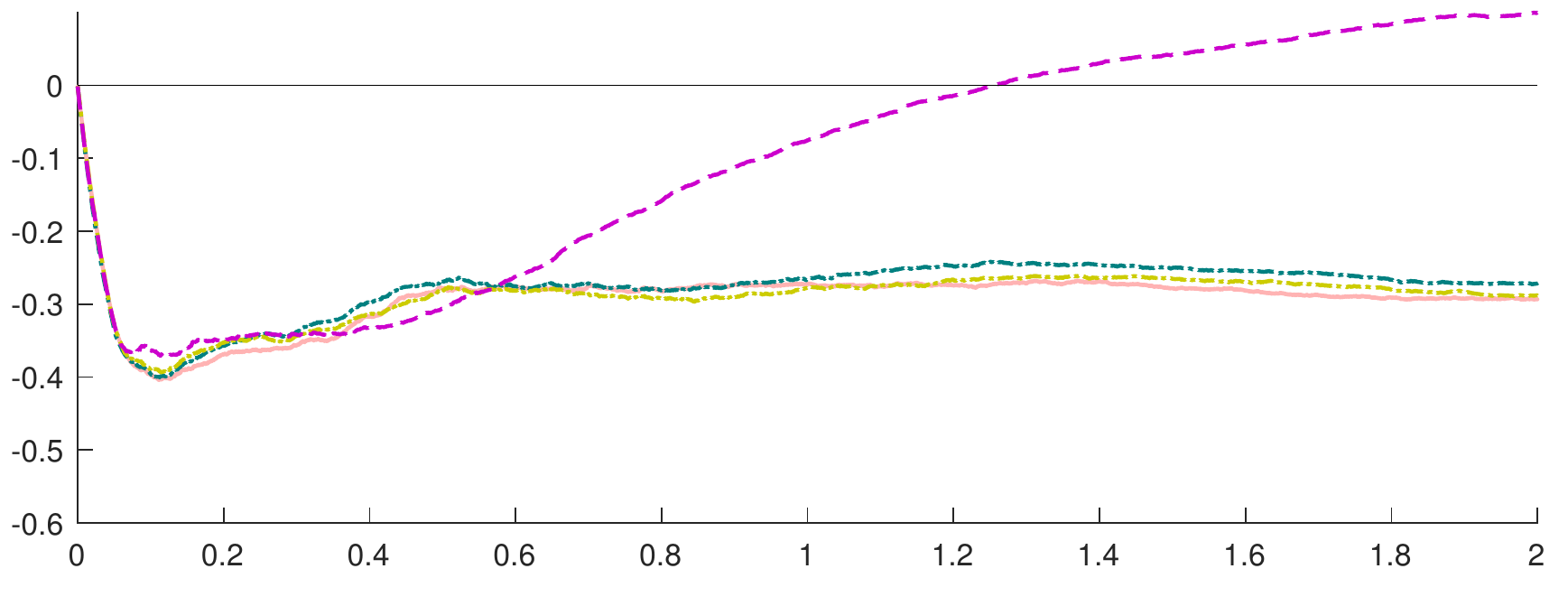}} \\
\caption{Computation of the Mandel-Q parameter \eqref{eq:MandelQ} corresponding to System \ref{Ex:DrivenSQMEapp} with  $d=50$ and $\omega_3 = 1.1\, \omega_2$.
We use $20000$  realizations of Schemes \ref{scheme:SME_Euler_Exponential}, \ref{scheme:SME_Exponential_Integral}, \ref{scheme:SME_ExponentialExponential_Integral_GD} and \ref{scheme:Euler_Implicit}.
}
\label{Fig:MandelQ}
\end{figure}

Using Remark \ref{rem:aproximacion} 
we approximate System \ref{Ex:DrivenSQME} by the following stochastic master equation.

\begin{system}
\label{Ex:DrivenSQMEapp}
Adopt the stochastic quantum master equation \eqref{eq:StochasticQME} with 
$\boldsymbol{\mathfrak{h}}  =  \ell^2_d \otimes \mathbb{C}^{2} $,
\begin{align*}
\mathbf{H}
& =
{\omega_1} \,\boldsymbol{\sigma}^z/2 + \omega_2 \, \mathbf{P}_d  \, \mathbf{a}^{\dag} \mathbf{a}  \, \mathbf{P}_d + g   \, \mathbf{P}_d \left( \mathbf{a}^{\dag} + \mathbf{a} \right ) \, \mathbf{P}_d  \, \boldsymbol{\sigma}^x
\\
& \quad +
\epsilon \mathbf{P}_d\left( \exp{ \left (- \mathrm{i} \left( \omega_3 t + \phi \right) \right )}  \mathbf{a}^{\dag}  
+ \exp{\left ( \mathrm{i} \left( \omega_3 t + \phi \right) \right )}  \mathbf{a} \right) \mathbf{P}_d  ,
\end{align*}
Furthermore, set $ \mathbf{L}_1  = \sqrt{\alpha_{1}}  \, \mathbf{P}_{d} \mathbf{a} \mathbf{P}_{d}$, 
$ \mathbf{L}_2 = \sqrt{\alpha_{2}} \, \mathbf{P}_d \mathbf{a}^{\dag} \mathbf{P}_d $, 
$ \mathbf{R}_1 = \sqrt{\beta_1  \gamma} \boldsymbol{\sigma}^-$,
$ \mathbf{R}_2 = \sqrt{\beta_1 \left( 1 -  \gamma \right)} \boldsymbol{\sigma}^-$,
$ \mathbf{R}_3 = \sqrt{\beta_2} \boldsymbol{\sigma}^- 
$, 
$ \mathbf{R}_4 = \sqrt{\beta_3} \boldsymbol{\sigma}^+ $ and
$ \mathbf{R}_5 = \sqrt{\beta_4} \boldsymbol{\sigma}^z 
$.
As in System \ref{Ex:DiffusiveSQMEapp},
$\ell^2_d$ denotes the linear span of $\mathbf{e}_0, \ldots, \mathbf{e}_d$ for any $d \in \mathbb{N}$,
and 
$\mathbf{P}_d$ is the orthogonal projection of  $\ell^2 \left(\mathbb{Z}_+ \right)$ onto $\ell^2_d$.
\end{system}

We simulate numerically System \ref{Ex:DrivenSQMEapp} with $d=50$,
which is a good approximation of System \ref{Ex:DrivenSQME},
by applying 
Schemes \ref{scheme:SME_Euler_Exponential}-\ref{scheme:Euler_Implicit}.
In order to test the accuracy of Schemes \ref{scheme:SME_Euler_Exponential}-\ref{scheme:Euler_Implicit},
we compute 
 $\mathbb{E} \,\tr \left (\mathbf{L}_1 \, \boldsymbol{\rho}_t\right )$ and $\mathbb{E} \,\tr \left ( \mathbf{R}_1 \, \boldsymbol{\rho}_t\right )$,
where $0\leq t \leq 2$.
The references values have been obtained  by solving  \eqref{eq:1.3} 
via the ode45 MATLAB program,
which is based on the Dormand-Prince method.
Figure \ref{Fig:Ej2A2} shows estimations of 
$\mathbb{E}\,\tr\left (\mathbf{L}_1\boldsymbol{\rho}_t\right )$ and $\mathbb{E}\,\tr\left (\mathbf{R}_1\boldsymbol{\rho}_t\right )$
obtained from sampling $2\cdot10^4$ times Schemes \ref{scheme:SME_Euler_Exponential}-\ref{scheme:Euler_Implicit}
with step sizes $\Delta$ equal to $2^{-6}$, $2^{-9}$ and $2^{-12}$. 
Table \ref{Tab:Ej2OpA1} provides the errors 
$$
\epsilon_2 \left( \Delta \right)
 =
\max \left\{ 
\left\vert \mathbb{E} \left( \tr \left( \mathbf{L}_1  \bar{\boldsymbol{\rho}}_{n}  \right) \right) - \mathbb{E} \left( \tr \left( \mathbf{L}_1  \boldsymbol{\rho}_{n \Delta}  \right)  \right) \right\vert
:
n = 0 ,  \ldots,  \lfloor 1 / \Delta \rfloor \right\} 
$$
and 
$$
\epsilon_3 \left( \Delta \right)
 =
\max \left\{ 
\left\vert \mathbb{E} \left( \tr \left( \mathbf{R}_1  \bar{\boldsymbol{\rho}}_{n}  \right) \right) - \mathbb{E} \left( \tr \left( \mathbf{R}_1  \boldsymbol{\rho}_{n \Delta}  \right)  \right) \right\vert
:
n = 0 ,  \ldots,  \lfloor 1 / \Delta \rfloor  \right\} ,
$$
as well as the relative mean CPU time $\tau \left (\Delta\right )$ used for processing 100 realizations of Schemes \ref{scheme:SME_Euler_Exponential}-\ref{scheme:Euler_Implicit} with step size $\Delta$, which was calculated as in Table \ref{Tab:Error}.

Figure \ref{Fig:Ej2A2} and Table \ref{Tab:Ej2OpA1}
illustrate the good accuracy of 
Schemes \ref{scheme:SME_Euler_Exponential}, \ref{scheme:SME_Exponential_Integral} and \ref{scheme:SME_ExponentialExponential_Integral_GD}.
In contrast to Scheme \ref{scheme:Euler_Implicit},
Schemes \ref{scheme:SME_Euler_Exponential}-\ref{scheme:SME_ExponentialExponential_Integral_GD}
reproduce very well the oscillatory behavior 
as well as the amplitud modulation of the interference between 
the natural frequency of the small system and the driving coherent field.

Now, we compute the Mandel $Q$-parameter:
\begin{equation}
\label{eq:MandelQ}
Q_1 \left( t \right) 
=
\frac{ \mathbb{E} \left( \left( N^1_t \right)^2 \right)  - \left( \mathbb{E} N^1_t \right)^2  }
{\mathbb{E} \left(  N^1_t \right)}
- 1 ,
\end{equation}
which characterizes the departure of the number of photocounts from Poisson statistics
(see, e.g., \cite{BarchielliGregoratti2012,Loudon2000,MandelWolf1995}).
Figure \ref{Fig:MandelQ} displays  estimations of $Q_1 \left( t \right)$
obtained from Schemes \ref{scheme:SME_Euler_Exponential}-\ref{scheme:Euler_Implicit}
with step sizes 
$\Delta = 2^{-6}, 2^{-9}, 2^{-12}, 2^{-15}$.
Schemes  \ref{scheme:SME_Euler_Exponential}-\ref{scheme:SME_ExponentialExponential_Integral_GD} 
consistently show that $Q_1 \left( t \right)$ takes negative values stabilizing around $-0.28$.
These negative values indicate sub-Poissonian statistics that is an indicator of non-clasical effects,
in a good agreement with the Physics of  System \ref{Ex:DrivenSQME}.
On the other hand,
Scheme \ref{scheme:Euler_Implicit} 
change its behavior from sub-Poissonian statistics to Poissonian statistics as the step size $\Delta$ decrease.

Finally,   
we study the behavior of System \ref{Ex:DrivenSQME} 
for different frequencies $\omega_3$ of the monochromatic laser light,
while keeping the other parameters unchanged.
To this end,
we simulate numerically System \ref{Ex:DrivenSQMEapp} with $d=100$ 
by using Scheme \ref{scheme:SME_Euler_Exponential}.
Figure \ref{Fig:MandelQfrecs} presents the Mandel $Q$-parameter $Q_1 \left( t \right)$ 
when 
$\omega_3  = 0$, $0.4\, \omega_2$, $0.8\, \omega_2$, $\omega_2$, $1.2\, \omega_2$, $1.6\, \omega_2$, $2\, \omega_2$;
we recall that the time $t$ is given in  in nanoseconds.
If $\omega_3$ is near the resonance frequency $\omega_2$,
then Figure \ref{Fig:MandelQfrecs} shows that $Q_1 \left( t \right)$ takes negative values.
In the other case, 
Schemes  \ref{scheme:SME_Euler_Exponential}
predicts that $Q_1 \left( t \right)$ varies from negative to positive values as the system evolves in time.

\begin{figure}[tb]
\centering
\includegraphics[height= 0.3in,width=3.17in]{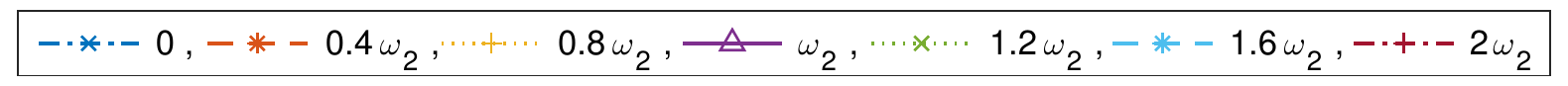}
\\
\includegraphics[height= 2.0in,width=5.5in]{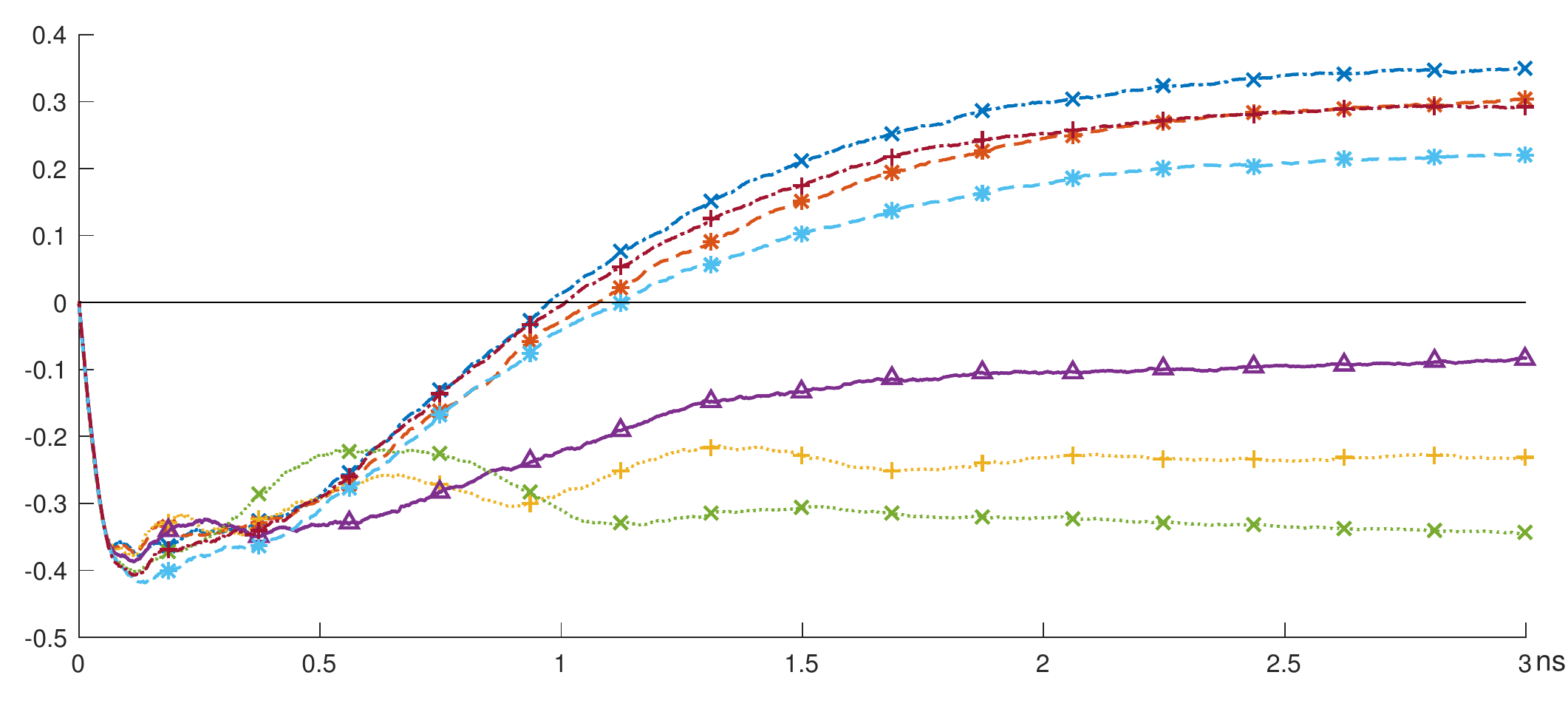}
\caption{Consider System \ref{Ex:DrivenSQMEapp} with  $d=100$.
We show the Mandel-Q parameter \eqref{eq:MandelQ} for different values of $\omega_3$
obtained from  $20000$ observations of Scheme \ref{scheme:SME_Euler_Exponential}.
We take $\omega_3$ equal to $0$, $0.4\, \omega_2$, $0.8\, \omega_2$, $\omega_2$, $1.2\, \omega_2$, $1.6\, \omega_2$ and $2\, \omega_2$.
}
\label{Fig:MandelQfrecs}
\end{figure}


\section{Proofs}
\label{sec:Proofs}


\subsection{Proof of Theorem \ref{th:Representation}}
\label{sec:Proof:Representation}

\begin{proof}

As
$ \Delta \left[ N^m, N^{\bar{m}} \right] = \left( \Delta N^m \right) \left( \Delta N^{\bar{m}} \right)$ (see, e.g., \cite{Protter2005})
we have 
$$
\left[ N^m, N^{\bar{m}} \right] 
=
\begin{cases}
 N^m
 &
\text{if } m = \bar{m}
\\
0
&
\text{otherwise} ,
\end{cases}
$$
because 
$\Delta N^m_t  \in \left\{ 0, 1 \right\}$
and
$N^1,\ldots, N^M$ are counting process having no common jumps.
Since 
$$
\ketbra{\mathbf{X}^k_{s}}{\mathbf{X}^k_{s}} - \ketbra{\mathbf{X}^k_{s-}}{\mathbf{X}^k_{s-}} -  \ketbra{\mathbf{X}^k_{s-}}{\Delta \mathbf{X}^k_{s}}
-  \ketbra{\Delta \mathbf{X}^k_{s} }{\mathbf{X}^k_{s-}} -  \ketbra{\Delta \mathbf{X}^k_{s} }{ \Delta \mathbf{X}^k_{s} }
=
0 ,
$$
applying It\^o's formula for complex-valued semimartingales (see, e.g., \cite{Protter2005}) yields  
\begin{align*}
& \ketbra{\mathbf{X}^k _t}{\mathbf{X}^k _t}
 =
\ketbra{\mathbf{X}^k _0}{\mathbf{X}^k _0} 
+ 
\int_{0}^t \left( 
\ketbra{\mathbf{G} \left ( s - , \mathbf{X}^k_{s-}  \right)}{\mathbf{X}^k_{s-}} + \ketbra{\mathbf{X}^k_{s-}}{\mathbf{G} \left ( s - , \mathbf{X}^k_{s-}  \right)} 
\right) ds 
\\ & \quad 
+
\sum_{j=1}^J \int_{0}^t 
\left(
\ketbra{\mathbf{L}_{j}  \left( s \right) \mathbf{X}^k_{s-}  - \ell_j \left( s - \right) \mathbf{X}^k_{s-} }{\mathbf{X}^k_{s-}} 
+
\ketbra{\mathbf{X}^k_{s-}}{\mathbf{L}_{j}  \left( s \right) \mathbf{X}^k_{s-}  - \ell_j \left( s - \right) \mathbf{X}^k_{s-} } 
\right)
dW^j_s
\\ & \quad 
+
\sum_{j=1}^J \int_{0}^t 
\left(
\ketbra{\mathbf{L}_{j}  \left( s \right) \mathbf{X}^k_{s-}  - \ell_j \left( s - \right) \mathbf{X}^k_{s-} }
{\mathbf{L}_{j}  \left( s \right) \mathbf{X}^k_{s-}  - \ell_j \left( s -\right) \mathbf{X}^k_{s-}}
\right)
ds  
\\ & \quad 
+
 \sum_{m=1}^M \int_{0+}^t   \left(  
 \ketbra{ \dfrac{ \mathbf{R}_{m}  \left( s \right) \mathbf{X}^k_{s-}}{ \sqrt{r_{m} \left( s - \right)} } - \mathbf{X}^k_{s-} }{\mathbf{X}^k_{s-}} 
 +
 \ketbra{\mathbf{X}^k_{s-}}{ \dfrac{ \mathbf{R}_{m}  \left( s \right) \mathbf{X}^k_{s-}}{ \sqrt{r_{m} \left( s - \right)}} - \mathbf{X}^k_{s-} } 
 \right) dN^m_s 
 \\ & \quad 
+
 \sum_{m=1}^M \int_{0+}^t   \left(  
 \ketbra{ \dfrac{ \mathbf{R}_{m}  \left( s \right) \mathbf{X}^k_{s-}}{ \sqrt{r_{m} \left( s - \right)} } - \mathbf{X}^k_{s-} }
 {\dfrac{ \mathbf{R}_{m}  \left( s \right) \mathbf{X}^k_{s-}}{ \sqrt{r_{m} \left( s - \right)}} - \mathbf{X}^k_{s-}} 
 \right) dN^m_s .
\end{align*} 
Using simple algebraic manipulations we get
\begin{align}
\label{eq:3.5}
& \ketbra{\mathbf{X}^k _t}{\mathbf{X}^k _t}
 =
\ketbra{\mathbf{X}^k _0}{\mathbf{X}^k _0} 
+
 \sum_{m=1}^M  \int_{0}^t   \mathbf{R}_{m} \left( s - \right) \ketbra{\mathbf{X}^k_{s-} }{ \mathbf{X}^k_{s-}} \, ds
\\ & \quad 
\nonumber
+ 
\int_{0}^t \left( 
\ketbra{\mathbf{G} \left ( s \right) \mathbf{X}^k_{s-} }{\mathbf{X}^k_{s-}} + \ketbra{\mathbf{X}_{s-} }{\mathbf{G} \left ( s \right) \mathbf{X}^k_{s-}}
+
\sum_{j=1}^J 
\ketbra{ \mathbf{L}_{j}  \left( s \right) \mathbf{X}^k_{s-} } {\mathbf{L}_{j}  \left( s \right) \mathbf{X}^k_{s-} } 
\right) ds 
\\ & \quad 
\nonumber
+
\sum_{j=1}^J \int_{0}^t 
\left(
\ketbra{\mathbf{L}_{j}  \left( s \right) \mathbf{X}^k_{s-}}{\mathbf{X}^k_{s-}}  + \ketbra{ \mathbf{X}^k_{s-}}{\mathbf{L}_{j}  \left( s \right) \mathbf{X}^k_{s-}} 
- 2  \ell_j \left( s - \right) \ketbra{ \mathbf{X}^k_{s-}}{\mathbf{X}^k_{s-}}
\right)
dW^j_s
\\ & \quad 
\nonumber
+
 \sum_{m=1}^M \int_{0+}^t   \left(  
 \dfrac{  \ketbra{  \mathbf{R}_{m}  \left( s  \right) \mathbf{X}^k_{s-}}{\mathbf{R}_{m}  \left( s \right) \mathbf{X}^k_{s-}} }
 {r_{m} \left( s - \right)}
 -
 \ketbra{\mathbf{X}^k_{s-} }{ \mathbf{X}^k_{s-}} 
 \right) dN^m_s .
\end{align}

Define 
$
\boldsymbol{\rho}^k_t = \ketbra{\mathbf{X}^k _t}{\mathbf{X}^k _t}
$
and 
$
\boldsymbol{\varrho}_t =  \sum_{k=1}^{\mu}  \boldsymbol{\rho}^k_t
$.
Then
\begin{equation}
\label{eq:3.6}
\ell_j \left( s \right) 
= \sum_{k=1}^{\mu}  \Re \left( \tr \left( \mathbf{L}_{j}  \left( s \right) \boldsymbol{\rho}^k_{s}  \right)  \right)
= \Re \left( \tr \left( \mathbf{L}_{j}  \left( s \right)  \boldsymbol{\varrho}_{s}  \right) \right)
\end{equation}
and
\begin{equation}
\label{eq:3.7}
r_{m} \left( s \right) 
= \sum_{k=1}^{\mu}   \tr \left( \mathbf{R}_{m} \left( s \right)^*  \mathbf{R}_{m} \left( s \right)  \boldsymbol{\rho}^k_{s}  \right)
=  \tr \left( \mathbf{R}_{m} \left( s \right)^*  \mathbf{R}_{m} \left( s \right)  \boldsymbol{\varrho}_{s}  \right) .
\end{equation}
Thus,
the stochastic intensity of $N^m$ is 
$\tr \left( \mathbf{R}_{m} \left( s \right) ^*\mathbf{R}_{m} \left( s \right) \boldsymbol{\rho}_{s-} \right) $
for any $m=1,\ldots,M$.
Combining \eqref{eq:3.5}, \eqref{eq:3.6} and \eqref{eq:3.7} gives 
\begin{align*}
\boldsymbol{\rho}^k_t 
&  =
\ketbra{\mathbf{X}^k _0}{\mathbf{X}^k _0} 
+
\sum_{m=1}^M  \int_{0}^t   \tr \left( \mathbf{R}_{m} \left( s \right)^*  \mathbf{R}_{m} \left( s \right)  \boldsymbol{\varrho}_{s-}  \right) \boldsymbol{\rho}^k_{s-} \, ds
\\ 
& \quad
+ \int_{0}^t \left( 
\mathbf{G} \left ( s \right) \boldsymbol{\rho}^k_{s-} + \boldsymbol{\rho}^k_{s-}  \mathbf{G} \left ( s \right)^*
+
\sum_{j=1}^J  \mathbf{L}_{j}  \left( s \right) \boldsymbol{\rho}^k_{s-}  \mathbf{L}_{j}  \left( s \right)^*
\right) ds 
\\ & \quad 
+
\sum_{j=1}^J \int_{0}^t 
\left(
\mathbf{L}_{j}  \left( s \right) \boldsymbol{\rho}^k_{s-}  + \boldsymbol{\rho}^k_{s-}  \mathbf{L}_{j}  \left( s \right)^* 
- 2  \Re \left( \tr \left( \mathbf{L}_{j}  \left( s \right)  \boldsymbol{\varrho}_{s-}  \right) \right) \boldsymbol{\rho}^k_{s-} 
\right)
dW^j_s 
\\ & \quad
+
 \sum_{m=1}^M \int_{0+}^t   \left(  
 \dfrac{   \mathbf{R}_{m}  \left( s \right) \boldsymbol{\rho}^k_{s-}  \mathbf{R}_{m}  \left( s \right)^* }
 { \tr \left( \mathbf{R}_{m} \left( s \right)^*  \mathbf{R}_{m} \left( s \right)  \boldsymbol{\varrho}_{s-}  \right)}
 -
\boldsymbol{\rho}^k_{s-}
\right) dN^m_s .
\end{align*} 
Therefore, $\boldsymbol{\rho}_t $ satisfies \eqref{eq:StochasticQME}.
\end{proof}

\subsection{Proof of Theorem \ref{th:Pade1}}
\label{sec:Proof:Pade1}

\begin{proof}
Since 
$$
\begin{pmatrix}
  \mathbf{G}  & \mathbf{g}  \\
  \mathbf{0}_{1 \times d}  & 0
\end{pmatrix}^j
=
\begin{pmatrix}
  \mathbf{G}^j  &  \mathbf{G}^{j-1} \mathbf{g}  \\
  \mathbf{0}_{1 \times d}  & 0
\end{pmatrix}
\hspace{2cm}
\forall j \in \mathbb{N} ,
$$
simple algebraic manipulations yield
$
\mathbf{N}_{p q} \left( \frac{\mathbf{A}}{2^m} \right)
=
\begin{pmatrix}
  \mathbf{N}_{p q} \left( \frac{\mathbf{G}}{2^m} \right)  & \widetilde{\mathbf{N}}_{ p q} \left( \frac{\mathbf{G}}{2^m} \right) \frac{\mathbf{g}}{2^m}  \\
  \mathbf{0}_{1 \times d}  & 1
\end{pmatrix} 
$
and
$$
\mathbf{D}_{p q} \left( \frac{\mathbf{A}}{2^m} \right)
=
\begin{pmatrix}
  \mathbf{D}_{p q} \left( \frac{\mathbf{G}}{2^m} \right)  & \widetilde{\mathbf{D}}_{ p q} \left( \frac{\mathbf{G}}{2^m} \right) \frac{\mathbf{g}}{2^m}  \\
  \mathbf{0}_{1 \times d}  & 1
\end{pmatrix} .
$$
Hence
$
\mathbf{D}_{p q} \left( \frac{\mathbf{A}}{2^m} \right)^{-1}
=
\begin{pmatrix}
  \mathbf{D}_{p q} \left( \frac{\mathbf{G}}{2^m} \right)^{-1} & - \mathbf{D}_{p q} \left( \frac{\mathbf{G}}{2^m} \right)^{-1} \widetilde{\mathbf{D}}_{ p q} \left( \frac{\mathbf{G}}{2^m} \right) \frac{\mathbf{g}}{2^m}  \\
  \mathbf{0}_{1 \times d}  & 1
\end{pmatrix} 
$,
and so
$$
\mathbf{D}_{p q} \left( \mathbf{A}/{2^m} \right)^{-1} \mathbf{N}_{p q} \left( \mathbf{A}/{2^m} \right)
=
\begin{pmatrix}
 \mathbf{P}_{p q} \left( \frac{\mathbf{G}}{2^m} \right) 
 & 
 \mathbf{D}_{p q} \left( \frac{\mathbf{G}}{2^m} \right)^{-1} 
\left(
 \widetilde{\mathbf{N}}_{ p q} \left( \frac{\mathbf{G}}{2^m} \right) -  \widetilde{\mathbf{D}}_{ p q} \left( \frac{\mathbf{G}}{2^m} \right) 
\right) \frac{\mathbf{g}}{2^m}  \\
  \mathbf{0}_{1 \times d}  & 1
\end{pmatrix} .
$$
Hence
\begin{align*}
& \left( \mathbf{D}_{p q} \left( \frac{\mathbf{A}}{2^m} \right)^{-1} \mathbf{N}_{p q} \left( \frac{\mathbf{A}}{2^m} \right) \right)^{2^m}
\\
& \quad =
\begin{pmatrix}
 \mathbf{P}_{p q} \left( \frac{\mathbf{G}}{2^m} \right)^{2^m}
 & 
\left(  \sum_{j=0}^{2^m - 1} \mathbf{P}_{p q} \left( \frac{\mathbf{G}}{2^m} \right)^{j} \right)
\mathbf{D}_{p q} \left( \frac{\mathbf{G}}{2^m} \right)^{-1} 
\left(
 \widetilde{\mathbf{N}}_{ p q} \left( \frac{\mathbf{G}}{2^m} \right) -  \widetilde{\mathbf{D}}_{ p q} \left( \frac{\mathbf{G}}{2^m} \right) 
\right) \frac{\mathbf{g}}{2^m}  \\
  \mathbf{0}_{1 \times d}  & 1
\end{pmatrix} .
\end{align*}
Using that
$
\sum_{j=0}^{2^m-1} \mathbf{X}^j =  \prod^{m-1}_{j=0} \left( \mathbf{I} + \mathbf{X}^{2^j} \right)
$
for any matrix $\mathbf{X}$,
we obtain the assertion of the theorem.
\end{proof}

\subsection{Proof of Theorem \ref{th:Pade2}}
\label{sec:Proof:Pade2}

\begin{proof} Fix $\mu \in \mathbb{N}$.
Due to 
$$
\mathbf{A}^j =
\begin{pmatrix}
  \mathbf{G}^j  & \mathbf{G}^{j-1} \mathbf{a} & \mathbf{G}^{j-2} \left( \mathbf{G}\mathbf{g} + \mathbf{a}\,b  \right) \\
  \mathbf{0}_{1 \times d}  & 0 & 0\\
  \mathbf{0}_{1 \times d}  & 0 & 0 
\end{pmatrix} 
\hspace{2cm}
\forall j \geq 2 ,
$$
algebraic manipulations lead to 
$$
\mathbf{N}_{p q} \left( \frac{\mathbf{A}}{\mu} \right)
=
\begin{pmatrix}
 \mathbf{N}_{p q} \left( \frac{\mathbf{G}}{ \mu } \right) 
 &  
 \widetilde{\mathbf{N}}_{p q} \left( \frac{\mathbf{G}}{ \mu } \right) \frac{\mathbf{a}}{ \mu }
 &
  \widetilde{\mathbf{N}}_{p q} \left( \frac{\mathbf{G}}{ \mu } \right) \frac{\mathbf{g}}{ \mu }
  +  \widehat{\mathbf{N}}_{p q} \left( \frac{\mathbf{G}}{ \mu } \right) \frac{ \mathbf{a} \, b }{ \mu^2 }
 \\
  \mathbf{0}_{1 \times d}  & 1 & \frac{ b \, p  }{ \mu \, \left( p + q \right) }
   \\
  \mathbf{0}_{1 \times d}  & 0 & 1
\end{pmatrix} 
$$
and
$$
\mathbf{D}_{p q} \left( \frac{\mathbf{A}}{\mu} \right)
=
\begin{pmatrix}
 \mathbf{D}_{p q} \left( \frac{\mathbf{G}}{ \mu } \right) 
 &  
 \widetilde{\mathbf{D}}_{p q} \left( \frac{\mathbf{G}}{ \mu } \right) \frac{\mathbf{a}}{ \mu }
 &
  \widetilde{\mathbf{D}}_{p q} \left( \frac{\mathbf{G}}{ \mu } \right) \frac{\mathbf{g}}{ \mu }
  +  \widehat{\mathbf{D}}_{p q} \left( \frac{\mathbf{G}}{ \mu } \right) \frac{ \mathbf{a} \, b }{ \mu^2 }
\\
   \mathbf{0}_{1 \times d}  & 1 & - \frac{ b \, q  }{ \mu \, \left( p + q \right) }
\\
  \mathbf{0}_{1 \times d}  & 0 & 1
\end{pmatrix} .
$$
Therefore, 
$$
\mathbf{D}_{p q} \left(  \frac{\mathbf{A}}{ \mu } \right)^{-1}
=
\begin{pmatrix}
 \mathbf{D}_{p q} \left( \frac{\mathbf{G}}{ \mu } \right)^{-1}
 &  
  - \mathbf{D}_{p q} \left( \frac{\mathbf{G}}{ \mu } \right)^{-1} \widetilde{\mathbf{D}}_{p q} \left( \frac{\mathbf{G}}{ \mu } \right) \frac{\mathbf{a}}{ \mu }
 &
 \ast
 \\
   \mathbf{0}_{1 \times d}  & 1 &  \frac{ b \, q  }{ \mu \, \left( p + q \right) }
\\
  \mathbf{0}_{1 \times d}  & 0 & 1
\end{pmatrix} 
$$
with 
$\ast = - \mathbf{D}_{p q} \left( \frac{\mathbf{G}}{ \mu } \right)^{-1} \left( 
 \widetilde{\mathbf{D}}_{p q} \left( \frac{\mathbf{G}}{ \mu } \right)  \left(  \frac{\mathbf{g}}{ \mu } + \frac{ \mathbf{a} \, b \, q }{ \mu^2 \, \left( p + q \right) }  \right)
  +  \widehat{\mathbf{D}}_{p q} \left( \frac{\mathbf{G}}{ \mu } \right) \frac{ \mathbf{a} \, b }{ \mu^2 }
 \right)
$.
This gives
\begin{equation}
\label{eq:6.1}
 \mathbf{P}_{p q} \left(  \frac{\mathbf{A}}{ \mu } \right)
=
\begin{pmatrix}
 \mathbf{P}_{p q} \left( \frac{\mathbf{G}}{ \mu } \right) 
 &  
  \widetilde{\mathbf{r}} \left( \mu \right) \frac{\mathbf{a}}{ \mu }
 &
  \widetilde{\mathbf{r}} \left( \mu \right) \frac{\mathbf{g}}{ \mu }
  +  \widehat{\mathbf{r}} \left( \mu \right)  \frac{ \mathbf{a} \, b }{ \mu^2 }
 \\
  \mathbf{0}_{1 \times d}  & 1 & \frac{ b }{ \mu }
   \\
  \mathbf{0}_{1 \times d}  & 0 & 1
\end{pmatrix} ,
\end{equation}
where
$
\widetilde{\mathbf{r}} \left( \mu \right)
=
\mathbf{D}_{p q} \left( \frac{\mathbf{G}}{ \mu } \right)^{-1}
 \left(\widetilde{\mathbf{N}}_{p q}\left( \frac{\mathbf{G}}{ \mu } \right) - \widetilde{\mathbf{D}}_{p q} \left( \frac{\mathbf{G}}{ \mu } \right)\right) 
$
and 
$$
\widehat{\mathbf{r}} \left( \mu \right)
=
\mathbf{D}_{p q} \left( \frac{\mathbf{G}}{ \mu } \right)^{-1}
 \left(\widehat{\mathbf{N}}_{p q}\left( \frac{\mathbf{G}}{ \mu } \right) 
 - \widehat{\mathbf{D}}_{p q} \left( \frac{\mathbf{G}}{ \mu } \right)
 - \widetilde{\mathbf{D}}_{p q} \left( \frac{\mathbf{G}}{ \mu } \right)\right) .
$$

From \eqref{eq:6.1} we obtain
$$
\mathbf{P}_{p q} \left(  \frac{\mathbf{A}}{ \mu } \right)^{\mu}
=
\begin{pmatrix}
 \mathbf{P}_{p q} \left( \frac{\mathbf{G}}{ \mu } \right) ^{\mu}
 &  
\mathbf{r} \left( \mu \right)  \widetilde{\mathbf{r}} \left( \mu \right) \frac{\mathbf{a}}{ \mu }
 &
\mathbf{r} \left( \mu \right) 
\left(
 \widetilde{\mathbf{r}} \left( \mu \right) \frac{\mathbf{g}}{ \mu }
  +  \widehat{\mathbf{r}} \left( \mu \right)  \frac{ \mathbf{a} \, b }{ \mu^2 }
 \right)
+ 
\widehat{\mathbf{r}} \left( \mu \right) \widetilde{\mathbf{r}} \left( \mu \right) \frac{ \mathbf{a} \, b }{ \mu^2}
\\
  \mathbf{0}_{1 \times d}  & 1 &  b 
 \\
  \mathbf{0}_{1 \times d}  & 0 & 1
\end{pmatrix} ,
$$
with
$
\mathbf{r} \left( \mu \right) = \sum_{j=0}^{\mu-1}  \mathbf{P}_{p q} \left( \mathbf{G}/{ \mu } \right) ^{j} 
$
and
$
\widehat{\mathbf{r}} \left( \mu \right) = \sum_{j=0}^{\mu-1}  \left( \mu - 1 - j \right) \mathbf{P}_{p q} \left( \mathbf{G}/{ \mu } \right) ^{j} 
$.
Using algebraic manipulations we deduce that
$ m \mapsto \mathbf{r} \left( 2^m \right) $ and $ m \mapsto \widehat{\mathbf{r}} \left( 2^m \right)$
satisfy the recurrence relations \eqref{eq:2.21}.
\end{proof}

\section{Conclusion} 
\label{sec:Conclusion}
The numerical simulation of quantum measurement processes in continuous time, with mixed initial states,
leads to solve stochastic quantum master equations (SQMEs for short).
In a wide range of physical situations,
the SQMEs  are (or are approximated by)
stiff stochastic differential equations in high-dimensional linear operator space.
In order to overcome the difficulties arising in the direct numerical integration of SQMEs,
we find a new representation of the solution to the jump-diffusion SQME 
by means of coupled non-linear stochastic Schr\"odinger equations.
This has allowed us to design numerical methods for SQMEs
based on new exponential schemes for jump-diffusion stochastic Schr\"odinger equations.
Thus,
we develop a set of numerical methods that accurately calculate
stochastic density operators describing continuous measurements of quantum systems,
even if we use a typical desktop computer and the Hilbert state space has dimension in the order of hundreds of thousand.
We show the good performance of the new numerical integrators  
by simulating the continuous monitoring of two open quantum systems
formed by a quantized electromagnetic field interacting with a two-level system,
under the effect of the environment.

\section*{Acknowledgements}
J.F.  thanks the support of CONICYT grant CONICYT-PCHA/Doctorado Nacional/2013-21130792.
C.M.M. was partially supported by FONDECYT Grant 1140411 and BASAL Grant  PFB-03.
C.M.M. is grateful to Alberto Barchielli and  Matteo Gregoratti (Politecnico di Milano)
for asking him about the numerical solution of stochastic quantum master equations,
as well as for useful discussions on this topic.
The authors thank the referees for their valuable comments and suggestions on the manuscript.

\section*{References}

\end{document}